\pdfoutput=1
\PassOptionsToPackage{nameinlink,capitalize}{cleveref}
\documentclass[a4paper,UKenglish,cleveref,autoref,thm-restate,numberwithinsect]{lipics-v2021}

\pdfoutput=1 \nolinenumbers \hideLIPIcs

\usepackage[utf8]{inputenc}
\usepackage{mathpartir}
\usepackage{import}
\usepackage{multirow}
\usepackage{amssymb}
\usepackage{stmaryrd}
\usepackage{amsmath}
\usepackage{amsfonts}
\usepackage{mathtools}
\usepackage{amsthm}
\usepackage[makeroom]{cancel}
\usepackage[usenames,dvipsnames]{xcolor}
\usepackage{verbatim}
\usepackage{bussproofs}
\usepackage{url}
\usepackage{graphicx}
\hypersetup{
  colorlinks,
  filecolor=magenta,
  citecolor=Blue,
  linkcolor=Red,
  urlcolor=Blue}
\usepackage{xspace}
\usepackage{lmodern}
\usepackage{ragged2e}
\usepackage{blindtext}
\usepackage{cmll}
\usepackage{pdflscape}
\usepackage{float}
\usepackage{tikz}

\usepackage{graphicx}
\graphicspath{ {figures/} }
\usepackage{array}
\newcolumntype{P}[1]{w{c}{#1}}
\usepackage{blindtext}
\usepackage{multicol}

\usepackage{framed}

\usepackage{array}
\usepackage{booktabs}
\usepackage[bibliography=common]{apxproof}
\newtheoremrep{theoremapx}[definition]{Theorem}
\newtheoremrep{propositionapx}[definition]{Proposition}
\newtheoremrep{lemmaapx}[definition]{Lemma}
\newtheoremrep{corollaryapx}[definition]{Corollary}
\Crefname{lemmaapx}{Lemma}{Lemmata}

\allowdisplaybreaks
\crefname{section}{\S}{\S\S}
\usepackage{macros}
\title{Around Classical and Intuitionistic Linear Processes} 

\titlerunning{Around Classical and Intuitionistic Linear Processes} 

\author{Juan C. Jaramillo}{Unversity of Groningen, The Netherlands}{j.c.jaramillo.londono@rug.nl}{https://orcid.org/0009-0003-0973-4123}{Ministry of Science of Colombia (Minciencias).}\author{Dan Frumin}{Unversity of Groningen, The Netherlands}{dan@groupoid.moe}{https://orcid.org/0000-0001-5864-7278}{}
\author{Jorge A. P\'erez}{Unversity of Groningen, The Netherlands}{j.a.perez@rug.nl}{https://orcid.org/0000-0002-1452-6180}{Support of the Dutch Research Council (NWO) under project No.016.Vidi.189.046 (Unifying Correctness for Communicating Software) is gratefully acknowledged.}

\authorrunning{J.\,C.\,Jaramillo, D.\,Frumin, and J.\,A.\,P\'erez} 

\Copyright{Juan\,C.\,Jaramillo, Dan\,Frumin, and Jorge\,A.\,P\'erez} 

\ccsdesc{Theory of computation~Linear logic}
\ccsdesc{Theory of computation~Type theory}
\ccsdesc{Theory of computation~Process calculi}

\keywords{Process calculi, session types, linear logic} 

\category{}

\funding{}

\acknowledgements{We are most grateful to Bas van den Heuvel for initial discussions on the topic of this paper. We also thank the anonymous reviewers for their helpful suggestions. }

\EventEditors{Rupak Majumdar and Alexandra Silva}
\EventNoEds{2}
\EventLongTitle{35th International Conference on Concurrency Theory (CONCUR 2024)}
\EventShortTitle{CONCUR 2024}
\EventAcronym{CONCUR}
\EventYear{2024}
\EventDate{September 9--13, 2024}
\EventLocation{Calgary, Canada}
\EventLogo{}
\SeriesVolume{311}
\ArticleNo{7}

\begin{document}
\maketitle

\begin{abstract}
 Curry-Howard correspondences between Linear Logic (\LL) and session types
 provide a firm foundation for concurrent processes. 
As the correspondences hold for intuitionistic and classical versions of \LL (\ILL and \CLL), we obtain two different  families of type systems for concurrency. 
An open question remains: \emph{how do these two families exactly relate to each other?} 
 Based upon a translation from \CLL to \ILL due to Laurent, we provide two complementary answers, in the form of full abstraction results based on a typed observational equivalence due to Atkey. Our results elucidate hitherto missing formal links between seemingly related yet different type systems for concurrency. 
\end{abstract}

\section{Introduction}
\label{sec: introduction}

We address an open question on the logical foundations of concurrency, as resulting from Curry-Howard correspondences between linear logic (\LL) and session types. 
These correspondences, often referred to as \emph{\PaS}, connect
\LL propositions and session types, 
proofs in \LL and $\pi$-calculus processes, as well as 
cut-elimination  in \LL and process synchronization.
The result is type systems that elegantly ensure important forms of communication correctness for processes. 
The correspondence was discovered by Caires and Pfenning, who relied on an \emph{intuitionistic} presentation of \LL~(\ILL)~\cite{DBLP:conf/concur/CairesP10};   Wadler later presented it using  \emph{classical} \LL~(\CLL)~\cite{PropositionsAsSessions}. These two works triggered the emergence of multiple type systems for concurrency with firm logical foundations, based on (variants of) \ILL (e.g.,~\cite{DBLP:conf/ppdp/ToninhoCP11,DBLP:conf/fossacs/PfenningG15,DBLP:journals/pacmpl/BalzerP17,DBLP:conf/concur/CairesPPT19}) and \CLL (e.g.,~\cite{DBLP:conf/esop/CairesP17,DBLP:conf/fossacs/DardhaG18,DBLP:journals/pacmpl/KokkeMP19,DBLP:conf/concur/0001KDLM21,DBLP:journals/pacmpl/QianKB21}). While key differences between these two families of type systems,  intuitionistic and classical, have been observed~\cite{PiUll}, in this paper we ask:
can we formally relate them from the standpoint of (typed) process calculi?

From a logical standpoint, the mere existence of two different families of type systems may seem surprising---after all, the relationship between \ILL and \CLL is well understood~\cite{journal/jlc/Schellinx91,JudgementalAnalysis,Laurent}.
Laurent has given a thorough account of these relationships, including a translation from \CLL to \ILL~\cite{Laurent}.
A central insight in our work is the following: while translations from \CLL to \ILL are useful, they alone do not entail formal results for typed processes, and a satisfactory answer from the \PaS perspective must include process calculi considerations.

Let us elaborate.
Given some context $\Delta$ and a formula $A$, let us write 
$\CPJud \Delta$
and 
$\Delta \PiDiLLJud A$
to denote sequents in \CLL and \ILL, respectively.
Under 
the concurrent interpretation induced by 
\PaS, these sequents are annotated as
$P \CPJud \Delta$
and
$\Delta \PiDiLLJud P :: x:A$, respectively, 
where $P$ is a process, $x$ is a name, and  $\Delta$ is now a finite collection of assignments  $x_1 : A_1, \ldots, x_n : A_n$.
An assignment specifies a name's intended session protocol. This way, e.g., `$x: A \Tensor B$' (resp. `$x: A \parrType B$') says that   $x$ \emph{outputs} (resp. \emph{inputs}) a name of type $A$ before continuing as described by $B$. 
Also, `$x: \unit$' says that $x$ has completed its protocol and has no observable behavior.
The judgment $\Delta \PiDiLLJud P :: x:A$ has a rely-guarantee flavor: 
``$P$ \emph{relies} on the behaviors described by $\Delta$  to \emph{offer} a protocol $A$ on $x$''. 
Hence, the assignment $x:A$ in the right-hand side plays a special role: this is \emph{the} only observation made about the behavior of $P$. 
Differently, the judgment $P \CPJud \Delta$ simply reads as ``$P$ implements the behaviors described by $\Delta$''; as such, all assignments in $\Delta$ are equally relevant for observing the behavior of $P$.

Unsurprisingly, these differences between intuitionistic and classical processes arise in their associated   (typed) behavioral equivalences~\cite{LogicalRelations,DBLP:journals/pacmpl/KokkeMP19,ObservationsCP,DBLP:conf/lics/DerakhshanBJ21}.
For intuitionistic processes, theories of \emph{logical relations}~\cite{LogicalRelations,DBLP:conf/lics/DerakhshanBJ21} induce contextual equivalences in which only the right-hand side assignment matters in comparisons; the assignments in $\Delta$ are used to construct appropriate contexts.
For classical processes, we highlight Atkey's \emph{observed communication semantics}~\cite{ObservationsCP}, whose induced  observational equivalence  accounts  for the entire typing context.

Laurent's \emph{negative translation}  from \CLL to \ILL~\cite{Laurent}, denoted $\TranLauR{(-)}$,
translates formulas using the parameter $\myR$ (an arbitrary formula in \ILL) as a ``residual'' element. We have, e.g.,:
\[
\TranLauR{(A\otimes B)}=((\TranLauR{A}\multimap \myR)\otimes(\TranLauR{B}\multimap \myR))\multimap \myR
\]
As  
\cref{f:Transformations} shows, using $\TranLauR{(-)}$ we can transform  $\CPJud \Delta$ into $\TranLauR{(\Delta)}\PiDiLLJud \myR$.
Now, from the view of `propositions-as-sessions', we see that $\TranLauR{(-)}$ increases the size of formulas/protocols and that fixing $\myR = \unit$ results into the simplest residual protocol. 
Given this, Laurent's translation transforms $P \CPJud \Delta$
into $\TranLauR{(\Delta)}\PiDiLLJud \TranLauP{P} :: w : \unit$ (for some fresh $z$), 
where   $\TranLauP{P}$ is a process that reflects the translation.
The translation has an unfortunate effect, however: a classical process $P$ with observable behavior given by $\Delta$ is transformed into an intuitionistic process $\TranLauP{P}$ \emph{without observable behavior} (given by $w: \unit$). We conclude that, independently of the chosen $\myR$, the translation $\TranLauR{(-)}$ alone does not adequately  relate the concurrent interpretations of \CLL and \ILL, as it does not uniformly account for  observable behavior in $P$ and $\TranLauP{P}$. 

Our goal is to complement the scope of Laurent's translation, in a way that is consistent with existing theories of (typed) behavioral equivalence for logic-based processes.

\begin{figure}[t!]
\begin{tikzpicture}[thick,scale=0.9, every node/.style={scale=0.9}]
		\node  (0) at (2, 0) {$\boxed{\CPJud \Delta}$};
		\node (1) at (4.5, -2.25) {$\boxed{\TranLauR{\Delta}\PiDiLLJud \myR}$};
		\node (2) at (7, 0) {$\boxed{\CPJud \DualTranLauR{\Delta},\myR}$};
         \node(3) at (8.75,-0.65){\textcolor{darkgray}{\CLL}};
         \node(4) at (8.75,-1.35){\textcolor{darkgray}{\ILL}};
        \draw [dotted] (1.45,-1) -- (9.25,-1);
		\draw[thick,->] (0) -- (1) node[midway, below left]{$\TranLauR{(-)}$~\cite{Laurent}};
		\draw[thick, dashed,->] (0) -- (2)node[midway, above]{$\DualTranLauR{(-)}$~~[\Cref{sec:Laurent's denotations}]};
		\draw[thick,->]  (1) -- (2)node[midway, below right]{$\Dual{(-)}$};
\end{tikzpicture}
    \centering
\caption{Translations between \CLL and \ILL. In this paper, we shall fix $\myR = \unit$.}
\label{f:Transformations}
\vspace{-2mm}
\end{figure}
We proceed in two steps, shown in \cref{f:Transformations}. In the following,  we shall fix $\myR = \unit$ and omit `$\myR$' when clear from the context.
First,  there is a well-known {translation} from \ILL to \CLL, denoted $\Dual{(-)}$, under which a sequent 
$\Delta \PiDiLLJud A$
is transformed into
$\CPJud \Dual{\Delta}, A$. 
Our observation is that $\DualTranLau{(-)}$ (the composition of the two translations) goes from \CLL into  itself,
translating $\CPJud \Delta$ into $\CPJud \DualTranLau{\Delta}, \unit$.
At the level of processes, this allows us to consider the corresponding processes $P$ and $\TranLauP{P}$ in the common setting of classical processes.
To reason about their observable behavior we employ Atkey's  observational  equivalence~\cite{ObservationsCP}, denoted~`$\ObsEquivAtk$'.

Our second step leads to our \textbf{main contributions}: two full abstraction results that connect  
$\CPJud \Delta$
and
$\CPJud \DualTranLau{\Delta}, \unit$
from the perspective of \PaS.
\begin{itemize}
\item 
The first result, given in \secRef{sec:Laurent's denotations}, adopts a \emph{denotational} approach to ensure that $P$~(typable with $\CPJud \Delta$) and $\TranLauP{P}$ (typable with \emph{both} $\TranLau{\Delta} \PiDiLLJud \unit$ \emph{and} $\CPJud \DualTranLau{\Delta}, \unit$) are behaviorally equivalent. This full abstraction result ensures that 
$P \ObsEquivAtk Q$ iff $\TranLauP{P}\ObsEquivAtk \TranLauP{Q}$ (\Cref{cor:fullabsden}).

\item The second result, given in \secRef{sec:Transformers}, is an \emph{operational} bridge between $\CPJud \Delta$ and $\CPJud \DualTranLau{\Delta}, \unit$: \Cref{corollary: full abstraction transformers} ensures that $P \ObsEquivAtk Q$ iff $C[P] \ObsEquivAtk C[Q]$, where $C$ is a so-called \emph{transformer context}, which ``adapts'' observable behavior in processes using types in $\Delta$. 
\end{itemize}

\noindent Next, we recall \CP (Wadler's concurrent interpretation of \CLL), Atkey's observational equivalence, and Laurent's translation.
\secRef{sec:Laurent's denotations} and \secRef{sec:Transformers} develop our full abstraction results.
\secRef{s:conc} further discusses our contributions; in particular, we discuss how they are related to the \emph{locality} principle---one of the known distinguishing features between typed processes based on \PaS~\cite{PiUll,LinearLogicPropositions,DBLP:journals/jfp/Wadler14}. 

\section{Background}
\label{s:back}

\subparagraph*{Propositions-as-Sessions / Classical Processes (\CP).}
We shall work with \emph{classical processes}~(\CP) as proposed by  Wadler~\cite{PropositionsAsSessions}. 
Assuming an infinite set of \emph{names} ($x, y, z,\dots$), the set of \emph{processes} ($P, Q,\dots$) is defined as follows:
\begin{align*}
    P,Q::= \inact&\Sep (\nu x)P\Sep P\para Q \Sep \forward{x}{y}\Sep  \boutt{y}{x}(P\para Q) \Sep \impp{x}{y}P \Sep \Server{x}{y}P \Sep \Client{x}{y}P \\
    &  \Sep \choice{x}{i}P \Sep \case{x}{P}{Q}   \Sep \Emptyoutt{x}\Sep \Emptyimpp{x}P   \quad\quad \text{for } i\in \{1,2\}  
\end{align*}
\noindent
We write $P\Substitution{x}{y}$ to denote the capture-avoiding substitution of $y$ for $x$ in $P$.
We have usual constructs for inaction, restriction, and parallel composition. 
The forwarder $\forward{x}{y}$ equates $x$ and $y$.
We then have 
$\boutt{y}{x}(P\para Q)$ (send the restricted name $y$ along $x$, proceed as $P \para Q$)
and 
$\impp{x}{y}P$ (receive a name $z$ along $x$, proceed as $P\subst{z}{y}$).
Processes $\Server{x}{y}P$ and $\Client{x}{y}P$ denote a replicated input (server) and a client request, respectively.
Process $\choice{x}{i}P$ denotes the selection of one of the two alternatives of a corresponding branching process $\case{x}{P}{Q}$.
Processes $\Emptyoutt{x}$ and $\Emptyimpp{x}P$ enable coordinated closing of the session along $x$. In a statement, a name is \emph{fresh} if it is not among the names of the objects of the statement  (e.g., processes).

In $(\nu x)P$, name $x$ is bound in $P$; also, in $\boutt{y}{x}(P\para Q)$, $\impp{x}{y}P$, $\Server{x}{y}P$, and $\Client{x}{y}P$, name $y$ is bound in $P$ but not in $Q$.
The types are assigned to names and correspond to the following formulas of \CLL:
\[A,B::= \unit \Sep  \bot \Sep  A\Tensor B \Sep  A\parrType B\Sep  A\choiceType B\Sep  A\caseType B \Sep  \ServerTypeSingle{A}\Sep  \ClientTypeSingle{A}\]
The assignment $x:A$ says that the session protocol through $x$ goes as described by~$A$.
As we have seen, $x: A\Tensor B$ and $x:A\parrType B$ are read as sending and receiving along $x$, respectively.
Also, $x:A\choiceType B$ denotes the selection of either $A$ or $B$ along $x$, whereas $x:A\caseType B$ denotes the offer of $A$ and $B$ along $x$.
Finally, $x:\ServerTypeSingle{A}$ and $x:\ClientTypeSingle{A}$ assign server and client behaviors to $x$, respectively.
There then is a clear duality in the interpretation of the following pairs: $\Tensor$ and $\parrType$; $\choiceType$ and $\caseType$; and $\ServerTypeSingle{\empty}$ and $\ClientTypeSingle{\empty}$. It reflects reciprocity between the behavior of a name: when a process on one side sends, the process on the opposite side must receive, and vice versa.
 Formally, the dual type of $A$, denoted $\Dual{A}$, is defined as

\begin{tabular}{c c c c}
    $\Dual{\unit}:= \bot$ & $\Dual{(A\Tensor B)}:=\Dual{A} \parrType \Dual{B}$ & $\Dual{(A\caseType B)}:= \Dual{A} \choiceType \Dual{B}$ & $\Dual{(\ServerTypeSingle{A})}:=\ClientTypeSingle{\Dual{A}}$\\
    $\Dual{\bot}:=\unit$ & $\Dual{(A\parrType B)}:= \Dual{A} \Tensor \Dual{B}$  & $\Dual{(A \choiceType B)}:=  \Dual{A} \caseType \Dual{B}$ & $\Dual{(\ClientTypeSingle{A})}:=\ServerTypeSingle{\Dual{A}}$\\
\end{tabular}

Duality is an involution, i.e., $\Dual{(\Dual{A})}=A$. 
We write $\Delta, \Gamma$ to denote \emph{contexts}, a finite collection of assignments $x:A$.
The empty context is denoted `$\,\cdot\,$'.
The typing judgments are then of the form $P \CPJud \Delta$, with typing rules as in \Cref{fig:CP Rules}.
For technical convenience, we shall consider mix principles (rules \MixZero and \MixTwo, not included in~\cite{PropositionsAsSessions}), which enable the typing of useful forms of process composition.
This way, in \secRef{sec:Laurent's denotations}
we will consider \CPMixZero: the variant of \CP with \MixZero; 
in \secRef{sec:Transformers} we will consider \CPMixZeroTwo: the extension of  \CPMixZero with \MixTwo.

\begin{figure}
    \centering
\begin{mathpar}
    \inferrule*[Right=Id]{ }{\forward{x}{y}\CPJud x:A, y:A^\bot} 
    \and
    \inferrule*[Right=1]{ }{\Emptyoutt{x} \CPJud x:\unit}
    \and
    \inferrule*[Right=$\bot$]{P\CPJud \Gamma}{\Emptyimpp{x}P\CPJud \Gamma,x:\Lbot}\\
    \and
    \inferrule*[Right=$\Tensor$]{P\CPJud\Gamma,y:A \and Q\CPJud \Delta,x:B}{\boutt{y}{x}(P\para Q)\CPJud \Gamma,\Delta,x:A\Tensor B}
    \and
    \inferrule*[Right=$\parrType$]{P\CPJud \Gamma ,y:A, x:B}{\impp{x}{y}P\CPJud \Gamma, x:A \parrType B}
    \and
 \inferrule*[Right=$\choiceType_i$]{P\CPJud \Gamma,x:A_i}{\choice{x}{i}P\CPJud \Gamma,x:A_1\choiceType A_2}
 \and
 \inferrule*[Right=$\caseType$]{P\CPJud\Gamma,x:A_1 \and Q\CPJud \Gamma,x:A_2}{\case{x}{P}{Q}\CPJud \Gamma,x:A\caseType B}
 \and
 \inferrule*[Right=$\oc$]{P\CPJud \ClientTypeSingle{\Delta},y:A}{\Server{x}{y}P\CPJud\ClientTypeSingle{\Delta},x:\ServerTypeSingle{A}}\\
 \and
  \inferrule*[Right=$\textcolor{Black}{?}$]{P\CPJud \Delta,x:A}{\Client{x}{y}P\CPJud\Delta,x:\ClientTypeSingle{A}}
 \and
   \inferrule*[Right=C]{P\CPJud\Delta,x_1:\ClientTypeSingle{A},x_2:\ClientTypeSingle{A}}{P\Substitution{x_1}{x_2}\CPJud\Delta,x_1:\ClientTypeSingle{A}}
 \and
    \inferrule*[Right=W]{P\CPJud \Gamma}{P\CPJud \Gamma,x:\ClientTypeSingle{A}}
\and
    \inferrule*[Right=Cut]{P\CPJud \Gamma,x:A \and Q\CPJud \Delta,x:A^\bot}{\cut{x}{P}{Q}\CPJud \Gamma,\Delta}
    \and 
    \inferrule*[Right=\MixTwo]{P \CPJud \Delta \and  Q \CPJud \Gamma}{P \para Q \CPJud \Delta, \Gamma}
    \and     
            \inferrule*[Right=\MixZero]{ }{\inact\CPJud \cdot}    
    \end{mathpar}
    \caption{Typing rules. $\CP$ does not include `mix'. $\CPMixZero$ is $\CP + \MixZero$, $\CPMixZeroTwo$ is $\CPMixZero + \MixTwo$.}
    \label{fig:CP Rules}
\end{figure}

Note that the type system $\CPMixZero$ corresponds exactly to the sequent calculus for \CLL if we ignore name and process annotations.
This correspondence goes beyond typing: an important aspect of \PaS is that the dynamic behavior of processes (process reductions) corresponds to simplification of proofs (cut elimination).
In the following we will not need this reduction semantics, which can be found in, e.g., \cite{PropositionsAsSessions}.
Rather, we will use the denotational semantics of $\CPMixZero$ as defined by Atkey~\cite{ObservationsCP}, which we recall next.

\subparagraph*{Denotational Semantics for \CP.}
We adopt Atkey's denotational semantics for \CPMixZero~\cite{ObservationsCP}, which allows us to reason about observational equivalence of processes.
Here we recall the notions of configurations, observations, and denotations as needed for our purposes.
\begin{figure}[t]  
   \begin{mathpar}
        \inferrule[C:Proc]{P\CPJud \Gamma}{P\ConfJud{\Gamma}{\cdot}}
        \and
        \inferrule[C:Cut]{C_1\ConfJud{\Gamma_1, x:A}{\Theta_1}
          \and C_2 \ConfJud{\Gamma_2, x:A^\bot}{\Theta_2}}
        {C_1 \paraCut{x} C_2 \ConfJud{\Gamma_1,\Gamma_2}{\Theta_1,\Theta_2,x:A}}
        \and
        \inferrule[C:0]{ }{\ZeroConf\ConfJud{\cdot}{\cdot}}        
        \and
        \inferrule[C:W]{C\ConfJud{ \Gamma} { \Theta}}{C\ConfJud{ \Gamma, x:\ClientTypeSingle{A}} {\Theta}}
        \and
        \inferrule[C:Con]{C\ConfJud{\Gamma,x_1:\ClientTypeSingle{A},x_2:\ClientTypeSingle{A}}{\Theta}}{C\Substitution{x_1}{x_2}\ConfJud{\Gamma,x_1:\ClientTypeSingle{A}}{\Theta}}
    \end{mathpar}
    \caption{Classical Processes: Configurations.}
    \label{fig: Configurations}
\end{figure}
\begin{figure}[t]
  \begin{mathpar}
    \inferrule*[right=Stop]{ }{\ZeroConf\Downarrow ()}
    \and\inferrule*[right=Link]{\ConfigurationC{\Configuration{C'}{\Substitution{x}{y}}}\ObservationUpdate{x\mapsto a}}
    {\ConfigurationC{\forward{x}{y}\paraCut{x} C'}\ObservationUpdate{x\mapsto a, y\mapsto a}}
    \and
    \inferrule*[right=Comm]{\ConfigurationC{P\paraCut{x}Q}\ObservationUpdate{x\mapsto a}}
    {\ConfigurationC{\cut{x}{P}{Q}}\Observation}
    \and
    \inferrule*[right=0]{\ConfigurationC{\ZeroConf}\Observation}
    {\ConfigurationC{\inact}\Observation}
    \and
    \inferrule*[right=\text{$\Tensor\parrType$}]{\ConfigurationC{P\paraCut{y}(Q\paraCut{x}R)}\ObservationUpdate{x\mapsto a, y\mapsto b}}
    {\ConfigurationC{\boutt{y}{x}(P\para Q)\paraCut{x}\impp{x}{y}R}\ObservationUpdate{x\mapsto(a,b)}}
    \and
        \inferrule*[right=\text{$1\bot$}]{\ConfigurationC{P}\Observation}
    {\ConfigurationC{\Emptyoutt{x}\paraCut{x}\Emptyimpp{x}P}\ObservationUpdate{x\mapsto *}}
    \and
    \inferrule*[right=$\choiceType\caseType$]{\ConfigurationC{P\paraCut{x}Q_i}\ObservationUpdate{x\mapsto a}}
    {\ConfigurationC{\choice{x}{i}P\paraCut{x}\case{x}{Q_0}{Q_1}}\ObservationUpdate{x\mapsto (i,a)}}
    \and
    \inferrule*[right=$!?$]{\ConfigurationC{P\paraCut{y}Q}\ObservationUpdate{y \mapsto a}}
    {\ConfigurationC{\Server{x}{y}P\paraCut{x}\Client{x}{y}Q}\ObservationUpdate{x\mapsto \bag{a}}}
    \and
    \inferrule*[right=!W]{\ConfigurationC{C'}\Observation}
    {\ConfigurationC{\Server{x}{y}P\paraCut{x}C'}\ObservationUpdate{x\mapsto \emptyset}}
    \and
    \inferrule*[right=!C]{\ConfigurationC{\Server{x_1}{y}P\paraCut{x_1}(\Server{x_2}{y}P\paraCut{x_2}C')}\ObservationUpdate{x_1\mapsto \alpha,x_2\mapsto \beta}}
    {\ConfigurationC{\Server{x_1}{y}P\paraCut{x_1}C'\Substitution{x_1}{x_2}}\ObservationUpdate{x_1\mapsto \alpha \uplus \beta}}  
    \and
    \inferrule*[right=$\equiv$]{C'\Observation \qquad C\equiv C'}
    {C\Observation}
  \end{mathpar}
  \caption{Classical Processes: Observations.}
  \label{fig:Observations}
\end{figure}

The observational equivalence on processes relies on the notion of \emph{configuration}, which is a process that has some of its names selected for the purposes of ``observations''.
Configurations, defined in \Cref{fig: Configurations}, are typed as $C \ConfJud{\Delta}{\Theta}$, where $\Delta$ contains free/unconnected names, and $\Theta$ contains the names that we intend to observe.
Rule \rulenamestyle{C:Cut}, for example, states that we can compose two configurations along a name $x$, and make the name observable.

Observations for configurations are given in \Cref{fig:Observations}.
The observation relation $C \Downarrow \theta$ is defined for closed configurations $C \ConfJud{\cdot}{\Theta}$, and the shape observation $\theta \in \AtkeyDenotations{\Theta}$ is defined based on the shape of types in $\Theta$.
For a type $A$, the set of observations $\AtkeyDenotations{A}$ is defined as
\begin{align*}
    \AtkeyDenotations{\unit} = \AtkeyDenotations{\bot} &= \{*\} &
      \AtkeyDenotations{\ServerTypeSingle{A}} =  \AtkeyDenotations{\ClientTypeSingle{A}}  &= \mathcal{M}_f(\AtkeyDenotations{A}) \\
    \AtkeyDenotations{A\otimes B}= \AtkeyDenotations{A\parrType B}  &= \AtkeyDenotations{A}\times \AtkeyDenotations{B} &
    \AtkeyDenotations{A_0\oplus A_1} = \AtkeyDenotations{A_0\caseType A_1}  &= \sum_{i\in \{0,1\}} \AtkeyDenotations{A_i}
\end{align*}
and we set $\AtkeyDenotations{\Theta}=\AtkeyDenotations{x_1:A_1,\cdots,x_n:A_n}=\AtkeyDenotations{A_1}\times \cdots\times\AtkeyDenotations{A_n}$.
Here $\mathcal{M}_f(X)$ denotes finite multisets with elements from $X$.
We use the standard notations $\emptyset$, $\uplus$, and $\bag{a_1,\dots,a_n}$ to denote the empty multiset, multiset union, and multiset literals, respectively. 

If $\theta \in \AtkeyDenotations{x_1:A_1,\cdots,x_n:A_n}$, then we write $\theta[x_{i}\mapsto \theta_{i}]$ for the observation which is identical to $\theta$, except that its $i$th component is set to $\theta_{i}$.
In \Cref{fig:Observations}, Rule Stop says that $\inact$ has no observations; Rule $\Tensor\parrType$ collects observations $a$ and $b$ into a single observation $(a,b)$.

Using these notions, there is an immediate canonical notion of observational equivalence.
In the following, we write $P,Q\CPJud\Gamma$ whenever 
$P \CPJud\Gamma$ and $Q\CPJud\Gamma$ hold.
\begin{definition}[Observational equivalence~\cite{ObservationsCP}]\label{def: observational equivalence}
Let $P,Q\in \CPMixZero$ such that $P,Q\CPJud\Gamma$. They are observationally equivalent, written $P\ObsEquivAtk Q$, if for all configurations-process context $C[-]$ where $C[P],C[Q]\CPJud \cdot\para\Theta$, and all $\theta\in \AtkeyDenotations{\Theta}$, $C[P]\Downarrow\theta\;\Leftrightarrow \;C[Q]\Downarrow\theta$. 
\end{definition}
We take this notion of observational equivalence as \emph{the} equivalence of \CPMixZero, sometimes writing $P \ObsEquivAtk Q \CPJud \Gamma$ to emphasize the typing of processes we are comparing.
Establishing observational equivalence of two processes directly is complicated, due to the universal quantification over all potential configurations $C$.
To establish equivalence in a compositional way, we recall Atkey's notion of denotational semantics for \CP  in \Cref{fig:DenotationalSemantics}: it assigns to each process $P \CPJud \Delta$ a denotation $\AtkeyDenotations{P \CPJud \Delta}$ as a subset of observations $\AtkeyDenotations{\Delta}$ on its names.
When the typing of a process $P$ is clear from the context we simply write $\AtkeyDenotations{P} \subseteq \AtkeyDenotations{\Delta} $.
\begin{figure}[t]
  \begin{mathpar}
    \AtkeyDenotations{\forward{x}{y}\CPJud x:A,y:A^\bot}=\{(a,a)\para a\in \AtkeyDenotations{A}\}
    \and
   \AtkeyDenotations{\Emptyoutt{x}\CPJud x:\unit}=\{(*)\}
   \and
   \AtkeyDenotations{\inact\CPJud}=\{()\}
   \and
   \AtkeyDenotations{\Emptyimpp{x}P\CPJud \Gamma,x:\bot}=\{(\gamma,*)\para\gamma \in \AtkeyDenotations{P\CPJud \Gamma}\}
   \and
    \begin{aligned}[t]
    \AtkeyDenotations{\cut{x}{P}{Q}\CPJud\Gamma,\Delta}
    =\{(\gamma,\delta)\para (\gamma,a)\in \AtkeyDenotations{P\CPJud\Gamma,x:A}, (\delta,a)\in \AtkeyDenotations{Q\CPJud \Delta,x:A^\bot}\}
    \end{aligned}
    \and
    \AtkeyDenotations{\boutt{y}{x}(P\para Q)\CPJud\Gamma,\Delta, x:A\Tensor B}
    =\left\{(\gamma,\delta,(a,b))\para
    \begin{gathered}[c]
      (\gamma,a)\in \AtkeyDenotations{P\CPJud\Gamma,y:A},\\[-1.5em]
      (\delta,b)\in \AtkeyDenotations{Q\CPJud \Delta,x:B}      
    \end{gathered}
\right\}
    \and
    \AtkeyDenotations{\impp{x}{y}P\CPJud \Delta, x: A\parrType B}
      =\{(\gamma,\delta,(a,b))\para (\gamma,a,b)\in \AtkeyDenotations{P\CPJud \Delta,y:A,x:B}\}
    \and
    \AtkeyDenotations{\choice{x}{i}P\CPJud \Gamma, x:A_1\choiceType A_2}=\{(\gamma,(i,a))\para (\gamma,a)\in \AtkeyDenotations{P\CPJud \Gamma, x:A_i}\}
    \and
    \AtkeyDenotations{\case{x}{P_1}{P_2}\CPJud \Gamma, x: A_1\caseType A_2}
      =\bigcup_{i\in\{1,2\}}\{(\gamma,(i,a))\para (\gamma,a)\in \AtkeyDenotations{P_i\CPJud \Gamma, x: A_i}\}
    \and
    \AtkeyDenotations{\Server{x}{y}P\CPJud \ClientTypeSingle{\Delta},x: \ServerTypeSingle{A}}
    =\left\{
      \begin{aligned}
        &(\uplus_{j=1}^k\alpha_j^1,\dots,\uplus_{j=1}^k\alpha_j^n,\bag{a_1,\dots,a_n})\\[-1.5em]
         \para & \forall i\in \{1,\dots,k\}. (\alpha_i^1,\dots,\alpha_i^k,a_i)\in \AtkeyDenotations{P\CPJud \ClientTypeSingle{\Delta},y:A}
      \end{aligned}
\right\}
    \and
    \AtkeyDenotations{\Client{x}{y}P\CPJud \Gamma, x:\ClientTypeSingle{A}}=\{(\gamma,\bag{a})\para (\gamma,a) \in \AtkeyDenotations{P\CPJud \Gamma, y:A} \}
    \and

    \AtkeyDenotations{P\CPJud \Gamma, x:\ClientTypeSingle{A}}=\{(\gamma,\emptyset)\para \gamma \in \AtkeyDenotations{P\CPJud \Gamma}\}\\
    \and
    \AtkeyDenotations{P\Substitution{x_1}{x_2}\CPJud \Gamma,x_1:\ClientTypeSingle{A}}
      =\{(\gamma,\alpha_1\uplus\alpha_2)\para (\gamma,\alpha_1,\alpha_2)
          \in \AtkeyDenotations{P\CPJud \Gamma,x_1:\ClientTypeSingle{A},x_2:\ClientTypeSingle{A}}
                \}
  \end{mathpar}
  \caption{Classical Processes: Denotational Semantics}
  \label{fig:DenotationalSemantics}
\end{figure}
The denotational semantics are sound and complete w.r.t. the observations:
\begin{theorem}[Adequacy~\cite{ObservationsCP}]\label{theorem: adequacy}
If $C\CPJud \cdot \para \Theta$, then $C\Downarrow \theta$ iff $\theta \in\AtkeyDenotations{C\CPJud \cdot\para \Theta}$.
\end{theorem}
Hence, we can use denotational semantics to prove observational equivalence:
\begin{corollary}[\cite{ObservationsCP}]\label{ObservationalEquivalenceCorollary}
    If $P,Q\CPJud \Gamma$ and $\AtkeyDenotations{P}=\AtkeyDenotations{Q}$, then $P\ObsEquivAtk Q$.
\end{corollary}
Above, the condition $P,Q\CPJud \Gamma$ is important, as there are processes with different types that have the same denotations. Examples are
$\choice{x}{1}\Emptyoutt{x}\CPJud x:\unit\choiceType \unit $ and $\case{x}{\Emptyoutt{x}}{\Emptyoutt{x}}\CPJud x:\unit\caseType \unit$.
\begin{remark}
	\label{theorem: observational equivalences}
	Atkey shows that $\ObsEquivAtk$ captures many equalities on processes induced by proof transformations, such as cut permutations and commuting conversions; see~\cite[Sect.\,5]{ObservationsCP}.
\end{remark}

\subparagraph{Laurent's Translation  $\TranLauR{(-)}$.}
\label{sec:Laurent Transformation}
As mentioned above, Laurent gives a parametric  translation from $\CLL$ to $\ILL$. 
Here we recall this translation following \cite[\S\,2.1]{Laurent}, 
considering only the class of formulas needed for our purposes.
The formulas of $\ILL$ are built using the grammar:
\[I,J::= \unit  \Sep I\Tensor J \Sep I\multimap J\Sep I\choiceType J\Sep I\caseType J \Sep \ServerTypeSingle{I}\]
The sequent calculus for $\ILL$ (omitted for space reasons) works on the judgments of the form $\Delta \PiDiLLJud I$. 
Let $\myR$ be a fixed but arbitrary formula in $\ILL$. 
We have the following derivable rules:
    \begin{center}
\begin{tabular}{c c}
    \begin{minipage}{5cm}
        \begin{prooftree}
            \AxiomC{$\Gamma,I\PiDiLLJud \myR$}
            \RightLabel{$\myR_R$}
            \UnaryInfC{$\Gamma\PiDiLLJud I \multimap \myR$}
        \end{prooftree}
    \end{minipage}   
&
    \begin{minipage}{5cm}
    \begin{prooftree}
        \AxiomC{$\Gamma\PiDiLLJud I$}
        \AxiomC{$\myR \PiDiLLJud \myR$}
        \RightLabel{$\myR_L$}
        \BinaryInfC{$\Gamma, I\multimap \myR \PiDiLLJud \myR$}  
         \end{prooftree}
    \end{minipage}
\end{tabular}
\end{center}

By using Rule $\myR_R$, the formula $I$ in the left-hand side of $\PiDiLLJud$ becomes $I\multimap\myR$ on the right-hand side. 
 Similarly, by using Rule $\myR_L$, the formula $I$ on the right-hand side of $\PiDiLLJud$ becomes $I\multimap \myR$ on the left-hand side. 
 Moving the formula $I$ from one side to the other of $\PiDiLLJud$ results in $I\multimap \myR$, which allows us to mimic in $\ILL$ the one-sided sequents of $\CLL$.
The translation in $\ILL$ of a $\CLL$ formula $F$, denoted $\TranLauR{(F)}$, is inductively defined using this movement of formulas; see \Cref{Translation and its dual} (second column).
\begingroup
\setlength{\extrarowheight}{1.5pt}
\begin{table}[t]
\begin{center}
\begin{tabular}{|P{1cm}| P{3.5cm}| P{3.7cm}| P{3.7cm}|}
\hline
    $F$ & $\TranLauR{F}$ (\ILL)  & $\TranLauR{F}$ (\CLL) & $\DualTranLauR{F}$\\
\hline \hline
      $\bot$& $\unit $ &$\unit$& $\bot$\\
      \hline
       $\unit$ & $\unit\multimap \myR $&$\bot\parrType \myR$ & $\unit\otimes\Dual{\myR}$
       \\
       \hline
       $A\otimes B $&$
       \begin{array}[c]{c}
         ((\TranLauR{A}\multimap \myR)\otimes(\TranLauR{B}\multimap \myR))\\
         \multimap \myR
       \end{array}
       $ &$
       \begin{array}[c]{c}
          ((\DualTranLauR{A}\parrType \myR)\otimes(\DualTranLauR{B}\parrType \myR))^\bot\\
          \parr\, \myR
       \end{array}
       $ &
       $\begin{array}[c]{c}
            ((\DualTranLauR{A}\parrType \myR)\otimes(\DualTranLauR{B}\parrType \myR))\\
            \otimes\, \Dual{\myR}
       \end{array}
       $
        \\
        \hline
       $A \parrType B$ &$\TranLauR{A}\otimes\TranLauR{B}$ &$\TranLauR{A}\otimes\TranLauR{B}$& $\DualTranLauR{A} \parrType \DualTranLauR{B}$
       \\
       \hline
         $A\oplus B $& $
         \begin{array}{c}
           ((\TranLauR{A}\multimap \myR)\oplus(\TranLauR{B}\multimap \myR))\\
           \multimap \myR 
         \end{array}$ &
                       $\begin{array}{c}
                         ((\DualTranLauR{A}\parrType \myR)\oplus(\DualTranLauR{B}\parrType \myR))^\bot\\
                         \parr\, \myR
                       \end{array}$
                        &$\begin{array}{c}
                        ((\DualTranLauR{A}\parrType \myR)\oplus(\DualTranLauR{B}\parrType \myR))\\
                        \otimes\, \Dual{\myR}
                        \end{array}
                        $
    \\
       \hline
      $A\caseType B$ &$\TranLauR{A}\oplus\TranLauR{B}$ &$\TranLauR{A}\oplus\TranLauR{B}$& $\DualTranLauR{A} \caseType \DualTranLauR{B}$
      \\
\hline
 $\ServerTypeSingle{A}$ & $!(\TranLauR{A}\multimap \myR)\multimap\myR $ & $(!(\DualTranLauR{A}\parrType \myR))^\bot\parrType \myR$ & $!(\DualTranLauR{A}\parrType \myR)\otimes\Dual{\myR}$
 \\
 \hline
 $\ClientTypeSingle{A}$ & $ !((\TranLauR{A}\multimap \myR)\multimap \myR)$&$!((\DualTranLauR{A}\parrType \myR)^\bot\parrType \myR)$ & $?((\DualTranLauR{A}\parrType \myR)\otimes\Dual{\myR})$
 \\
    \hline
\end{tabular}
\end{center}
    \caption{Translations $\TranLauR{(-)}$ and  $\DualTranLauR{(-)}$.}
    \label{Translation and its dual}
\end{table}
\endgroup

The amount of (nested) occurrences of 
`$\multimap \myR$' indicates how many times a formula has to be moved. Not all connectives require such transformations; we will expand on this in \cref{sec:Laurent's denotations}.
This translation extends to contexts as expected; it is correct, in the following sense:
\begin{theorem}[\cite{Laurent}]\label{theorem: transformation laurent}
    If $\CPJud \Delta$ is provable in $\CLL$ then $\TranLauR{\Delta}\PiDiLLJud \myR$ is provable in $\ILL$.
\end{theorem}

Given an $\myR$ such that $\bigotimes_n \myR\PiDiLLJud \myR$ (for all $n>0$), the theorem extends to $\CLLMixZeroTwo$---$\CLL$ with the corresponding $\mbox{\MixZero}$ and $\mbox{\MixTwo}$ rules (obtained from \CPMixZeroTwo in \Cref{fig:CP Rules}).
The following result considers the case $\myR=\unit$; it will be useful in \Cref{sec:Transformers}, where we use $\CPMixZeroTwo$.
\begin{lemma}[\cite{Laurent}]
\label{theorem: laurent theorem 2.14}
  Let $\myR=\unit$.  $\CPJud \Delta$ is provable in $\CLLMixZeroTwo$ iff $\TranLauR{\Delta}\PiDiLLJud \myR$ is provable in $\ILL$.
\end{lemma}
As already mentioned, since we interpret propositions as sessions, we pick the simplest residual formula/protocol that satisfies the premise of \Cref{theorem: laurent theorem 2.14}, i.e., we fix $\myR = \unit$.
Considering this, in the remainder of the paper we refer to the translation simply as $\TranLau{(-)}$.
\section{A Denotational Characterization of Laurent's Translation}
\label{sec:Laurent's denotations}
Here we study the effect of Laurent's  translation $\TranLau{(-)}$ on  processes typed under $\CPMixZero$.
We prove our first full abstraction result (\Cref{cor:fullabsden}) by lifting $\TranLau{(-)}$ to the level of denotations.
\subparagraph{The Composed Translation.}
As discussed in \secRef{sec: introduction},  we wish to compare processes in the uniform setting of $\CLL$. 
We know that  if $\Delta\PiDiLLJud A$ is provable in $\ILL$, then $\CPJud\Delta^\bot,A$ is provable in $\CLL$ (see, e.g.,~\cite{Laurent}). 
Hence, we can interpret sequents in $\ILL$ as sequents in $\CLL$. 
Notice that formulas in $\ILL$ can be treated as formulas in $\CLL$ by letting $A\multimap B := A^\bot\parrType B$.
Using this transformation within the composition of $\TranLau{(-)}$ and $\Dual{(-)}$, we obtain the desired transformation on $\CLL$ proofs.   From now on, we shall write $\DualTranLau{(-)}$ to denote the translation 
given in \Cref{Translation and its dual} (rightmost column).
\begin{theoremapxrep}\label{transformation of proof cll to cll}
    If $\CPJud \Delta$ is provable in $\CLLMixZero$, then $\CPJud \DualTranLau{\Delta},\unit$ is provable in $\CLLMixZero$.
\end{theoremapxrep}
\begin{proof}
Given a provable sequent $\CPJud \Delta$ in \CLL, then taking $\myR=\unit$ in $\TranLauR{(-)}$  $\TranLau{\Delta}\PiDiLLJud \unit$, thus taking formulas in $A\multimap B $ in \ILL as $\Dual{A}\parrType B$ in \CLL, we can prove $\CPJud\DualTranLau{\Delta},\unit$. The details of the transformation can be found in \Cref{ProofTransformation}.

\begin{table}
    \centering
\begin{tabular}{|w{l}{3cm} w{c}{3cm} w{c}{8cm}|}
\hline
    \begin{minipage}{3cm}
    \begin{prooftree}
    \AxiomC{\empty}
        \UnaryInfC{$\CPJud \cdot$}
    \end{prooftree}
\end{minipage}&
$\ProofTran$&
\begin{minipage}{10cm}
    \begin{prooftree}
    \AxiomC{\empty}
        \UnaryInfC{$\CPJud \unit$}
    \end{prooftree}
\end{minipage}\\
& & \\
    \begin{minipage}{3cm}
   \begin{prooftree}
        \AxiomC{$\CPJud \unit$}
    \end{prooftree}
\end{minipage}&
$\ProofTran$&
\begin{minipage}{10cm}
    \begin{prooftree}
        \AxiomC{$\CPJud \unit$}
        \AxiomC{$\CPJud {\unit},{\bot}$}
        \BinaryInfC{$\CPJud x' :\unit\Tensor \bot, w:\unit$}
    \end{prooftree}
\end{minipage}\\
& & \\
\begin{minipage}{3cm}
      \begin{prooftree}
                \AxiomC{$\CPJud A^\bot,A$}
            \end{prooftree}
\end{minipage}& 
$\ProofTran$&
\begin{minipage}{10cm}
                \begin{prooftree}
             \AxiomC{$\CPJud A^\bot,A$}
                \AxiomC{$\CPJud {\unit}, .{\bot}$}
                \BinaryInfC{$\CPJud A\Tensor \bot,A^\bot,\unit$}
                 \end{prooftree}  
\end{minipage}\\
& & \\
\begin{minipage}{3cm}
    \begin{prooftree}
        \AxiomC{$\CPJud \Delta$}
        \UnaryInfC{$\CPJud \Delta,\bot$}
    \end{prooftree}       
\end{minipage}&
$\ProofTran$&
\begin{minipage}{10cm}
    \begin{prooftree}
        \AxiomC{$\CPJud \DualTranLau{\Delta},{\unit}$}
        \UnaryInfC{$\CPJud \DualTranLau{\Delta},\bot,{\unit}$}
    \end{prooftree}
\end{minipage}\\
& & \\
\begin{minipage}{3cm}
    \begin{prooftree}
        \AxiomC{$\CPJud \Delta,A,B$}
        \UnaryInfC{$\CPJud \Delta,A\parrType B$}
    \end{prooftree}
\end{minipage}&
$\ProofTran$&
\begin{minipage}{10cm}
    \begin{prooftree}
        \AxiomC{$\CPJud \DualTranLau{\Delta},\DualTranLau{A},\DualTranLau{B},{\unit}$}
        \UnaryInfC{$\CPJud \DualTranLau{\Delta},\DualTranLau{A}\parrType\DualTranLau{B},{\unit}$}
    \end{prooftree}
\end{minipage}\\
& & \\

\begin{minipage}{3cm}
    \begin{prooftree}
        \AxiomC{$\CPJud \Delta,A$}
        \AxiomC{$\CPJud \Delta,B$}
        \BinaryInfC{$\CPJud \Delta,A\caseType B$}
    \end{prooftree}
\end{minipage}&
$\ProofTran$&
\begin{minipage}{10cm}
    \begin{prooftree}
        \AxiomC{$\CPJud \DualTranLau{\Delta},\DualTranLau{A},\unit$}
        \AxiomC{$\CPJud \DualTranLau{\Delta},\DualTranLau{B},\unit$}
        \BinaryInfC{$\CPJud \DualTranLau{\Delta},\DualTranLau{A}\caseType \DualTranLau{B},\unit$}
    \end{prooftree}
\end{minipage}\\
& & \\
\begin{minipage}{3cm}
{\footnotesize
    \begin{prooftree}
        \AxiomC{$\CPJud \Delta,A$}
        \AxiomC{$\CPJud \Gamma,B$}
        \BinaryInfC{$\CPJud \Delta, A\Tensor B$}
    \end{prooftree}
    }
\end{minipage}&
$\ProofTran$&
\begin{minipage}{10cm}
{\footnotesize
    \begin{prooftree}
        \AxiomC{$\CPJud \DualTranLau{\Delta},\DualTranLau{A},{\unit}$}
        \UnaryInfC{$\CPJud \DualTranLau{\Delta},\DualTranLau{A}\parrType {\unit}$}
          \AxiomC{$\CPJud \DualTranLau{\Gamma},\DualTranLau{B},{\unit}$}
        \UnaryInfC{$\CPJud \DualTranLau{\Gamma},\DualTranLau{B}\parrType {\unit}$}
        \BinaryInfC{$ \CPJud \DualTranLau{\Delta},\DualTranLau{\Gamma},(\DualTranLau{A}\parrType {\unit})\otimes(\DualTranLau{B}\parrType {\unit})$}
        \AxiomC{$\CPJud {\unit},{\bot}$}
        \BinaryInfC{ $\CPJud {\DualTranLau{\Delta}}, \DualTranLau{(A\Tensor B)}, \unit$}
    \end{prooftree}
    }
\end{minipage}
\\
& & \\
\begin{minipage}{3cm}
{\footnotesize
    \begin{prooftree}
        \AxiomC{$\CPJud \Delta,A$}
        \UnaryInfC{$\CPJud \Delta,\Gamma, A\oplus B$}
    \end{prooftree}
    }
\end{minipage}&
$\ProofTran$&
\begin{minipage}{10cm}
{\footnotesize
   \begin{prooftree}
       \AxiomC{$\CPJud\DualTranLau{\Delta},\DualTranLau{A},\unit$}
       \UnaryInfC{$\CPJud\DualTranLau{\Delta},\DualTranLau{A}\parrType \unit$}
       \UnaryInfC{$\CPJud\DualTranLau{\Delta},(\DualTranLau{A}\parrType \unit)\oplus(\DualTranLau{B}\parrType \unit)$}
       \AxiomC{$\CPJud \bot,\unit$}
       \BinaryInfC{$\CPJud\DualTranLau{\Delta},((\DualTranLau{A}\parrType \unit)\oplus(\DualTranLau{B}\parrType \unit))\otimes \bot,\unit$}
   \end{prooftree}
    }
\end{minipage}\\[5mm]
\begin{minipage}{3cm}
    \begin{prooftree}
        \AxiomC{$\CPJud \ClientTypeSingle{\Delta},A$}
        \UnaryInfC{$\Server{x}{y}P\CPJud\ClientTypeSingle{\Delta},x:\ServerTypeSingle{A}$}
    \end{prooftree}
\end{minipage}&
$\ProofTran$&
\begin{minipage}{10cm}
    \begin{prooftree}
        \AxiomC{$\CPJud ?((\DualTranLau{\Delta}\parrType \unit)\otimes\bot),\DualTranLau{A},\unit$}
        \UnaryInfC{$\CPJud ?((\DualTranLau{\Delta}\parrType \unit)\otimes\bot),\DualTranLau{A}\parrType \unit$}
        \UnaryInfC{$\CPJud ?((\DualTranLau{\Delta}\parrType \unit)\otimes\bot),!(\DualTranLau{A}\parrType \unit)$}
        \AxiomC{$\CPJud \bot,\unit$}
        \BinaryInfC{$\CPJud\DualTranLau{(\ClientTypeSingle{\Delta})},\DualTranLau{(\ServerTypeSingle{A})}, \unit$}
    \end{prooftree}
\end{minipage}\\[5mm]
    \begin{minipage}{3cm}
    \begin{prooftree}
        \AxiomC{$\CPJud \Delta,A$}
        \UnaryInfC{$\CPJud \Delta,\ClientTypeSingle{A}$}
    \end{prooftree}
\end{minipage}&
$\ProofTran$&
\begin{minipage}{10cm}
   \begin{prooftree}
  \AxiomC{$\CPJud\DualTranLau{\Delta},\DualTranLau{A},\unit$}
  \UnaryInfC{$\CPJud\DualTranLau{\Delta},\DualTranLau{A}\parrType \unit$}
  \AxiomC{$\CPJud \bot,\unit$}
  \BinaryInfC{$\CPJud \DualTranLau{\Delta},(\DualTranLau{A}\parrType \unit)\otimes\bot,\unit$}
  \UnaryInfC{$\CPJud \DualTranLau{\Delta},?((\DualTranLau{A}\parrType \unit)\otimes\bot),\unit$}
    \end{prooftree}
\end{minipage}\\[5mm]
    \begin{minipage}{3cm}
    \begin{prooftree}
       \AxiomC{$\CPJud\Delta$}
       \UnaryInfC{$\CPJud\Delta, \ClientTypeSingle{A}$}
    \end{prooftree}
\end{minipage}&
$\ProofTran$&
\begin{minipage}{10cm}
    \begin{prooftree}
        \AxiomC{$\CPJud \DualTranLau{\Delta},\unit$}
        \UnaryInfC{$\CPJud \DualTranLau{\Delta},?((\DualTranLau{A}\parrType\unit)\otimes\bot),\unit$}
    \end{prooftree}
\end{minipage}\\[5mm]
    \begin{minipage}{3cm}
    \begin{prooftree}
       \AxiomC{$\CPJud\Delta,\ClientTypeSingle{A},\ClientTypeSingle{A}$}
       \UnaryInfC{$\CPJud\Delta, \ClientTypeSingle{A}$}
    \end{prooftree}
\end{minipage}&
$\ProofTran$&
\begin{minipage}{10cm}
    \begin{prooftree}
        \AxiomC{$\CPJud \DualTranLau{\Delta},?((\DualTranLau{A}\parrType\unit)\otimes\bot),?((\DualTranLau{A}\parrType\unit)\otimes\bot),\unit$}
        \UnaryInfC{$\CPJud \DualTranLau{\Delta},?((\DualTranLau{A}\parrType\unit)\otimes\bot),\unit$}
    \end{prooftree}
\end{minipage}\\[5mm]
\hline
\end{tabular}
    \caption{Proof transformations and proof equivalences}
    \label{ProofTransformation}
\end{table}
\end{proof}
We shall write $A\in \CPMixZero$, when is clear from the context that $A$ is a type. Similarly, $A \in \CLLMixZero$ says that $A$ is a formula in $\CLLMixZero$. 

\subparagraph{The Translation on Processes.}
  $\DualTranLau{(-)}$ induces a translation on processes, denoted $\TranLauP{-}$, which is defined inductively on typing derivations  (\Cref{def: Tranformation of processes}).
This translation  is the computational interpretation of the composition of the two steps in \Cref{f:Transformations}.
Before detailing its definition, we examine two illustrative cases: output and input. 
Let us first consider the process $P=\boutt{y}{x}({P_1}\para {P_2})$, which is  typed as follows:
\begin{prooftree}
    \AxiomC{${P_1}\CPJud \Delta, y:A $}
    \AxiomC{${P_2}\CPJud \Gamma, x:B$}
    \RightLabel{$\Tensor$}
    \BinaryInfC{$\boutt{y}{x}({P_1}\para {P_2})\CPJud \Delta,\Gamma,x:A\Tensor B$}
\end{prooftree}
From the standpoint of the \PaS interpretation, by \Cref{theorem: transformation laurent} 
there  exists a process $\TranLauP{P}$ and fresh names $z$ and $w$ 
such that 
$\TranLau{\Delta},z:\TranLau{(A\Tensor B)}\PiDiLLJud \TranLauP{P}::w:\unit$. 
As \Cref{Translation and its dual} (second column) shows, $\TranLau{(A\Tensor B)}=((\TranLau{A}\multimap \unit)\otimes(\TranLau{B}\multimap \unit))\multimap \unit$. To determine the shape of $\TranLauP{P}$, we can reason inductively and apply \Cref{theorem: transformation laurent} to the judgments ${P_1}\CPJud \Delta, y:A $ and ${P_2}\CPJud \Gamma, x:B$. This gives us $\TranLau{\Delta},y:\TranLau{A} \PiDiLLJud \TranLauP{{P_1}}:: z_1:\unit$ and $\TranLau{\Gamma},x:\TranLau{B} \PiDiLLJud \TranLauP{{P_2}}:: z_2:\unit$, respectively. We can then obtain the following typing derivation (and shape) for $\TranLauP{P}$:
\begin{prooftree}
    \AxiomC{$\TranLau{\Delta},y:\TranLau{A} \PiDiLLJud \TranLauP{{P_1}}:: z_1:\unit$}
    \UnaryInfC{$\TranLau{\Delta} \PiDiLLJud \impp{z_1}{y} \TranLauP{{P_1}}:: z_1:\TranLau{A}\LinearArrow \unit$}
     \AxiomC{$\TranLau{\Gamma},x:\TranLau{B} \PiDiLLJud \TranLauP{{P_2}}:: z_2:\unit$}
     \UnaryInfC{$\TranLau{\Gamma} \PiDiLLJud \impp{z_2}{x} \TranLauP{{P_2}}:: z_2:\TranLau{B}\LinearArrow \unit$}
    \BinaryInfC{$\TranLau{\Delta},\TranLau{\Gamma}\PiDiLLJud \boutt{z_1}{z_2}(\impp{z_1}{y} \TranLauP{{P_1}}\para \impp{z_2}{x} \TranLauP{{P_2}}):: z_2:(\TranLau{A}\LinearArrow \unit)\Tensor(\TranLau{B}\LinearArrow \unit)$}
    \AxiomC{$(\star)$}
    \BinaryInfC{$\TranLau{\Delta},\TranLau{\Gamma}, x':\TranLau{(A\Tensor B)}\PiDiLLJud
     \underbrace{ \boutt{z_2}{x'}( \boutt{z_1}{z_2}(\impp{z_1}{y} \TranLauP{{P_1}}\para \impp{z_2}{x} \TranLauP{{P_2}})\para \fwd{x'}{w})}_{\TranLauP{P}}::
       w:\unit$} 
\end{prooftree}
 where $(\star)$ stands for  $x':\unit\PiDiLLJud \forward{x'}{w}::w:\unit$.
Above, we see how each  nested `$\multimap \unit$' induced by Laurent's translation entails extra actions on the level of processes, due to the interpretation of $\multimap$ as input (when  introduced on the right, as in this case): moving $y : \TranLau{A}$ and $x : \TranLau{B}$ to the right-hand side induces the inputs along $z_1$ and $z_2$, respectively. Subsequently, we use the $\Tensor$ rule on the right, which produces the output on $z_2$; we then move the resulting assignment for $z_{2}$ back to the left, finally obtaining $x' : \TranLau{(A \Tensor B)}$. 
This last movement adds the final `$\multimap \unit$': because it is introduced on the left, we obtain an output along $x'$. At this point, we can return to the classical setting by applying the translation $(-)^\bot$ to the derivation above, 
      which leads to the following typing derivation for $\TranLauP{P}$ in \CPMixZero:
    \begin{prooftree}
    \AxiomC{$\TranLauP{{P_1}} \CPJud {\DualTranLau{\Delta}},y:\DualTranLau{A},  z_1: \unit$}
    \UnaryInfC{$\impp{z_1}{y} \TranLauP{{P_1}}\CPJud {\DualTranLau{\Delta}},  z_1:{\DualTranLau{A}}\parrType \unit$}
     \AxiomC{$\TranLauP{{P_2}}\CPJud {\DualTranLau{\Gamma}},x:{\DualTranLau{B}} , z_2: \unit$}
      \UnaryInfC{$\impp{z_2}{x} \TranLauP{{P_2}} \CPJud {\DualTranLau{\Gamma}}, z_2:{\DualTranLau{B}}\parrType \unit$}
      \BinaryInfC{$\boutt{z_1}{z_2}(\impp{z_1}{y} \TranLauP{{P_1}}\para \impp{z_2}{x} \TranLauP{{P_2}})\CPJud {\DualTranLau{\Delta}},{\DualTranLau{\Gamma}} , z_2:({\DualTranLau{A}}\parrType \unit)\Tensor({\DualTranLau{B}}\parrType \unit)$}
    \AxiomC{$(\star \star )$}
 \BinaryInfC{$\underbrace{\boutt{z_2}{x'}( \boutt{z_1}{z_2}(\impp{z_1}{y} \TranLauP{{P_1}}\para \impp{z_2}{x} \TranLauP{{P_2}})\para \fwd{x'}{w})
}_{\TranLauP{P}}\CPJud {\DualTranLau{\Delta}},{\DualTranLau{\Gamma}}, x':\DualTranLau{(A\Tensor B)},
           w:\unit$}
           \end{prooftree}
            where $(\star \star )$ stands for $\forward{x'}{w}\CPJud x':\bot,w:\unit$. Importantly, while the translation $(-)^{\bot}$ modifies the types for $\TranLauP{P}$, it does not change its shape. Also, it is worth noticing how the output on $x$ in $P$ is mimicked by $\TranLauP{P}$ through the output on $z_2$, not by the output  on $x'$.
      
Now consider the case when $Q=\impp{x}{y}Q_1$, which is typed as:
    \begin{prooftree}
        \AxiomC{$Q_1\CPJud \Delta,y:A,x:B$}
        \UnaryInfC{$\impp{x}{y}Q_1\CPJud \Delta,x:A\parrType B$}
    \end{prooftree}
  In this case, we expect to obtain a process $\TranLauP{Q}$ such that 
$\TranLau{\Delta},x':\TranLau{(A\parrType B)}\PiDiLLJud \TranLauP{Q}::w:\unit$, where $\TranLau{(A\parrType B)}=\TranLau{A}\otimes\TranLau{B}$ (cf.   \Cref{Translation and its dual}, second column). By reasoning inductively on   $Q_1\CPJud \Delta,y:A,x:B$, we obtain  $\TranLau{\Delta}, y:\TranLau{A},x':\TranLau{B}\PiDiLLJud \TranLauP{Q_1}:: w:\unit$, which enables us to obtain the following derivation:
\begin{prooftree}
   \AxiomC{$\TranLau{\Delta}, y:\TranLau{A},x':\TranLau{B}\PiDiLLJud \TranLauP{Q_1}:: w:\unit$}
        \UnaryInfC{$\TranLau{\Delta}, x':\TranLau{A}\otimes\TranLau{B}\PiDiLLJud \underbrace{\impp{x'}{y}\TranLauP{Q_1}}_{\TranLauP{Q}}:: w:\unit$}
    \end{prooftree}
    Differently from the case of $\Tensor$, here the transformation $\TranLau{(-)}$ does not add any `$\multimap\unit$'. This is relevant, because it ensures that the process $\TranLauP{Q}$ does not have input/output actions in front of the input on $x'$, which mimics the input on $x$ in $Q$.
By applying the translation $(-)^\bot$ to the derivation above, we obtain the  following typing derivation for $\TranLauP{Q}$ in \CPMixZero:
        \begin{prooftree}
        \AxiomC{$\TranLauP{Q_1}\CPJud \DualTranLau{\Delta},y:\DualTranLau{A},x':\DualTranLau{B},w:\unit$}
        \UnaryInfC{$\underbrace{\impp{x'}{y}\TranLauP{Q_1}}_{\TranLauP{Q}}\CPJud \DualTranLau{\Delta},x':\DualTranLau{(A\parrType B)},w:\unit$}
    \end{prooftree}
    Once again, notice that the translation $(-)^{\bot}$ does not modify the shape of $\TranLauP{Q}$.

A key observation is that although 
$P=\boutt{y}{x}({P_1}\para {P_2})$
and 
$Q=\impp{x}{y}Q_1$
are compatible (i.e., they have complementary actions on $x$), their translations $\TranLauP{P}$ and $\TranLauP{Q}$
are {not}. 
In general, given two composable processes
       $P\CPJud \Delta,x:A$ and $Q\CPJud \Gamma,x:\Dual{A}$, we have: 
\[
	\TranLauP{P} \CPJud\DualTranLau{\Delta},x':\DualTranLau{A},w:\unit
	\qquad 
	\TranLauP{Q} \CPJud\DualTranLau{\Gamma},x':\DualTranLau{\Dual{A}},z:\unit
\]
and so 
$\TranLauP{P}$ and $\TranLauP{Q}$ cannot be composed directly: the types of $x'$  are not dual ($\Dual{(\DualTranLau{A})}\neq \DualTranLau{\Dual{A}}$).
To circumvent this difficulty, we shall consider  
\emph{synchronizer processes} $\Synchronizer{z,w}{A}$ such that 
$$\Synchronizer{z,w}{A}\CPJud w:\TranLau{A}\Tensor\bot,z:\TranLau{\Dual{A}}\Tensor\bot,s:\unit$$
Using synchronizers, a mediated composition between $\TranLauP{P}$ and $\TranLauP{Q}$ is then possible: 
\begin{align*}
        \cut{w}{\cut{z}{\impp{z}{y}\TranLauP{Q}}{\Synchronizer{z,w}{A}}}{\impp{w}{x'}\TranLauP{P}}
\end{align*}
Synchronizer processes have a purely logical origin: 
Laurent~\cite{Laurent} shows that for any $A$ the sequent $\CPJud  \TranLau{A}\Tensor\bot,\TranLau{A^\bot}\Tensor \bot,\unit$ is provable; using this result, the definition of synchronizers (given next) arises by reading off the process associated with this proof.

\begin{definition}[Synchronizer]\label{def: synchronizers}
Given $F \in \CPMixZero$ and names $z$, $w$, and $s$, we define 
the synchronizer process $\Synchronizer{z,w}{A}$, satisfying
 $\Synchronizer{z,w}{A}\CPJud z: \TranLau{A}\Tensor\bot,w :\TranLau{\Dual{A}}\Tensor\bot, s:\unit$, by recursion on $A$.
\end{definition}

Armed with the notion of synchronizer processes, we can finally define:
\begin{definition}[Laurent's translation on processes]\label{def: Tranformation of processes}
    Let $\Synchronizer{z,w}{A}$  be a synchronizer as in \Cref{def: synchronizers}.
Given a typed process $P\CPJud \Delta$, we define $\TranLauP{P}$ inductively in \Cref{f: Tranformation of processes}. 
\end{definition}

\begin{figure}[!t]
    \begin{align*}
    \TranLauP{\fwd{x}{y}}&=\boutt{x}{x'}(\fwd{x}{y}\para \fwd{x'}{w})
    \\
    \TranLauP{\Emptyimpp{x}P}&=\Emptyimpp{x'}\TranLauP{P}
\\
\TranLauP{\Emptyoutt{x}}&=\boutt{x}{x'}(\Emptyoutt{x}\para \fwd{x'}{w})
\\
\TranLauP{\impp{x}{y}P}&= \impp{x'}{y}\TranLauP{P}
\\
\TranLauP{{\boutt{y}{x}(\textcolor{Black}{P}\para \textcolor{Black}{Q})}}&={\boutt{z_2}{x'}( \boutt{z_1}{z_2}(\impp{z_1}{y} \TranLauP{\textcolor{Black}{P}}\para\impp{z_2}{x} \TranLauP{\textcolor{Black}{Q}})\para \forward{x'}{w})}
\\	
\TranLauP{\case{x}{P_1}{P_2}}&= \case{x'}{\TranLauP{P_1}}{\TranLauP{P_2}}
    \\
    \TranLauP{\choice{x}{i}P}&=\boutt{z}{x'}(\choice{z}{i}\impp{z}{y}\TranLauP{P}\para \fwd{x'}{w})
           \\
         \TranLauP{\Server{x}{y}P}&=\boutt{x}{x'}(\Server{x}{v}\impp{v}{y}\TranLauP{P}\para \fwd{x'}{w})
    \\
            \TranLauP{\Client{x}{y}P}&=\Client{x'}{m}\boutt{v}{m}(\impp{v}{y}\TranLauP{P}\para \fwd{m}{w})
           \\
    \TranLauP{\cut{x}{P}{Q}}
    &=\cut{w}{\cut{z}{\impp{z}{x'}\TranLauP{Q}}{\Synchronizer{z,w}{A}}}{\impp{w}{x'}\TranLauP{P}}   
\end{align*}
\caption{Laurent's translation on \CP processes (\Cref{def: Tranformation of processes}). \label{f: Tranformation of processes}}
\end{figure}

The  next lemma  ensures that for a given \CPMixZero process $P$, $\TranLauP{P}$ is well-typed.
\begin{lemmaapxrep}\label{lemma: transformation on processes}
   Let $P\CPJud \Delta$ be a process in $\CPMixZero$, then $\TranLauP{P}\CPJud \DualTranLau{\Delta},w:\unit$.
\end{lemmaapxrep}
     
\begin{example}
     To illustrate mediated composition, consider the 
 processes $\Emptyoutt{x}\CPJud x:\unit$ and $\Emptyimpp{x}P\CPJud \Delta,x:\bot$.
  By \Cref{lemma: transformation on processes}, we have $\boutt{x}{x'}(\Emptyoutt{x}\para \fwd{x'}{w})\CPJud x':\DualTranLau{\unit}, w:\unit$ and $ \impp{z}{m}\Emptyimpp{m}\TranLauP{P}\CPJud \DualTranLau{\Delta}, m:\bot, z:\unit$, respectively.
  These two processes can be composed with the synchronizer for $A = \unit$:
  $$ \Synchronizer{w,z}{\unit}=\boutt{x'}{w}(\impp{x'}{x}\boutt{m_1}{z}(\forward{m_1}{x}\para\forward{x'}{z} )\para \forward{w}{s})
$$
By expanding \Cref{def: Tranformation of processes}, we obtain the following observational equivalence:
 \begin{align*}
     \TranLauP{\cut{x}{\Emptyoutt{x}}{\Emptyimpp{x}P}}&=\cut{z}{\cut{w}{\impp{w}{x'}\TranLauP{\Emptyoutt{x}}}{\Synchronizer{z,w}{\unit}}}{\impp{z}{m}\TranLauP{\Emptyimpp{x}P}}
     \\
     &=(\nu z)((\nu w)(\impp{w}{x'}\boutt{x}{x'}(\Emptyoutt{x}\para \fwd{x'}{w})
     \\
     &\qquad\para \boutt{x'}{w}(\impp{x'}{x}\boutt{m}{z}(\forward{m}{x}\para\forward{x'}{z} )\para \forward{w}{s}))
     \\
     &\qquad\para \impp{z}{m}\Emptyimpp{m}\TranLauP{P})\\
     &\ObsEquivAtk\TranLauP{P}\Substitution{s}{z} 
 \end{align*}
\end{example}

\subparagraph{Properties.}
The logical translations  strongly suggest that $P$ and $\TranLauP{P}$ should be equivalent in some sense. How to state this relation?
Our technical insight is to bring $\DualTranLau{(-)}$ to the level of denotations: 
we define the function $\FDenotationsInd{A}{-}:\,\AtkeyDenotations{ A}\to \AtkeyDenotations{   \DualTranLau{A}}$, which   ``saturates''  $\AtkeyDenotations{A}$ by adding as many `$*$' (the observation of $\unit$ and $\bot$)
as residual $\bot$s and $\unit$s are induced by $\DualTranLau{(-)}$. 
\begin{figure}[!t]
\begin{align*}
\FDenotationsInd{\bot}{*}&=* 
&
\FDenotationsInd{\ServerTypeSingle{A}}{\bag{a_1,a_2}}&= (\bag{(\FDenotationsInd{A}{a_1}, *),  (\FDenotationsInd{A}{a_2}, *)},*)
\\
\FDenotationsInd{\unit}{*}&=(*,*) 
& 
\FDenotationsInd{\ClientTypeSingle{A}}{\bag{a_1,a_2}}&=\bag{((\FDenotationsInd{A}{a_1},*),*), ((\FDenotationsInd{A}{a_2},*),*)}
\\
\FDenotationsInd{A\parrType B}{(a,b)}&= (\FDenotationsInd{A}{a},\FDenotationsInd{B}{b}) 
&
\FDenotationsInd{A\Tensor B}{(a,b)}&= ((\FDenotationsInd{A}{a},*),(\FDenotationsInd{B}{b},*),*) 
\\
 \FDenotationsInd{A_1\caseType A_2}{(i,a)}&=(i,\FDenotationsInd{A_i}{a})
& 
\FDenotationsInd{A_1\oplus A_2}{(i,a)}&=((i,(\FDenotationsInd{A_i}{a},*)),*)
\end{align*}	
\caption{Transformation on denotations induced by Laurent's translation (\Cref{def: transformation on denotations}).
The generalized definitions $\FDenotationsInd{\ServerTypeSingle{A}}{\bag{a_1,\ldots, a_n}}$ and $\FDenotationsInd{\ClientTypeSingle{A}}{\bag{a_1,\ldots, a_n}}$ arise as expected. 
 \label{t:transden}}
\end{figure}

\begin{definition}[Transformations on Denotations]\label{def: transformation on denotations}
Given $A\in\CPMixZero$, 
$\FDenotationsInd{A}{-}:\AtkeyDenotations{A}\mapsto\AtkeyDenotations{\DualTranLau{A}}$
is  defined in \Cref{t:transden}. 
Given $\Delta=x_1:A_1,\dots,x_n:A_n$, 
we define 
$\FDenotations{\Delta}{-}:\AtkeyDenotations{\Delta} \mapsto \AtkeyDenotations{\DualTranLau{\Delta},1}$
as
       
    \[\FDenotations{\Delta}{(a_1,\dots,a_n)}= (\FDenotationsInd{A_1}{a_1},\dots,\FDenotationsInd{A_n}{a_n},*)\]  
\end{definition}
\noindent
Our goal is to show that 
$\FDenotations{\Delta}{\AtkeyDenotations{P\CPJud\Delta}}=\AtkeyDenotations{\TranLauP{P}\CPJud\DualTranLau{\Delta},w:\unit}$
(\Cref{theorem: laurent transformation on denotations}).
In particular we use synchronizers (and their observations) to prove the property for processes with cut. 
Also, the following two properties will be useful:

\begin{lemmaapxrep}\label{lemma: injectivity Function L formula}
    For any $A$ and $\Delta$ in $\CPMixZero$, 
 both $\FDenotationsInd{A}{-}$ and $\FDenotations{\Delta}{-}$ are injective. 
\end{lemmaapxrep}
\begin{proof}

The first part, on $\FDenotationsInd{A}{-}$, proceeds by induction on the structure of $A$.

    The cases for $\unit$ and $\bot$ follow trivially.
    The rest of the cases follow by inductive hypothesis.
    
Thus, for the case $\ServerTypeSingle{A}$, given $\bag{a_1,\dots,a_n},\bag{a'_1,\dots,a'_n}\in \AtkeyDenotations{\ServerTypeSingle{A}}$
\begin{align*}
    \FDenotationsInd{\ServerTypeSingle{A}}{\bag{a_1,\dots,a_n}}\quad&=\quad \FDenotationsInd{\ServerTypeSingle{A}}{\bag{a'_1,\dots,a'_n}}\\
   \Leftrightarrow\quad (\bag{(\FDenotationsInd{A}{a_1}, *), \dots, (\FDenotationsInd{A}{a_n}, *)},*)\quad&=\quad (\bag{(\FDenotationsInd{A}{a'_1}, *), \dots, (\FDenotationsInd{A}{a'_n}, *)},*) \tag{by \Cref{def: transformation on denotations}}\\
    \Leftrightarrow\quad \bag{(\FDenotationsInd{A}{a_1}, *), \dots, (\FDenotationsInd{A}{a_n}, *)}\quad&=\quad \bag{(\FDenotationsInd{A}{a'_1}, *), \dots, (\FDenotationsInd{A}{a'_n}, *)}\\
    \Leftrightarrow\quad \FDenotationsInd{A}{a_1}=\FDenotationsInd{A}{a'_1}\quad &\dots\quad \FDenotationsInd{A}{a_n}=\FDenotationsInd{A}{a'_n}\\
    \Rightarrow\quad a_1=a'_1\quad &\dots\quad  a_n=a'_n\tag{by the I.H.}\\
    \Rightarrow\quad \bag{a_1,\dots,a_n}\quad&=\quad\bag{a'_1,\dots,a'_n}\\
\end{align*}

The second part, on $\FDenotations{\Delta}{-}$, follows using the first part.
    Let $\Delta =A_1,\dots,A_n$, and $(a_1,\dots,a_n),(a'_1,\dots,a'_n)\in\AtkeyDenotations{\Delta}$, then
        \begin{align*}
            \FDenotations{\Delta}{a_1,\dots,a_n}\quad&=\quad\FDenotations{\Delta}{a'_1,\dots,a'_n}\\
            \Leftrightarrow\quad (\FDenotationsInd{A_1}{a_1},\dots,\FDenotationsInd{A_n}{a_n},*)\quad&=\quad(\FDenotationsInd{A_1}{a'_1},\dots,\FDenotationsInd{A_n}{a'_n},*)\\
        \Leftrightarrow\quad \FDenotationsInd{A_1}{a_1}=\FDenotationsInd{A_1}{a'_1}\quad&\dots \quad\FDenotationsInd{A_n}{a_n}=\FDenotationsInd{A_n}{a'_n}\\
        \Rightarrow\quad a_1=a'_1 \quad&\dots\quad a_n=a'_n\tag{by \Cref{lemma: injectivity Function L formula}}\\
        \Rightarrow\quad (a_1,\dots,a_n)\quad&=\quad(a'_1,\dots,a'_n)\\
        \end{align*}
    \end{proof}
\begin{lemmaapxrep}\label{lemma: union of multisets}
    Let $\ClientTypeSingle{A}\in\CPMixZero$. Suppose $\alpha_j\in\AtkeyDenotations{\ClientTypeSingle{A}}$ for all $j=1,\dots,k$. Then
    $\FDenotationsInd{\ClientTypeSingle{A}}{\uplus^k_{j=1}\alpha_j}= \uplus^k_{j=1}\FDenotationsInd{\ClientTypeSingle{A}}{\alpha_j}$.
\end{lemmaapxrep}
\begin{proofsketch}
    It follows by \Cref{def: transformation on denotations} and definition of $\uplus$.
\end{proofsketch}
\begin{proof}
    $\alpha_i=\bag{a_1^i,\dots,a_{m_i}^i}$ for all $i=1,\dots,k$, then 

\begin{align*}
    \FDenotationsInd{\ClientTypeSingle{A}}{\uplus^k_{j=1}\alpha_j}&=\FDenotationsInd{\ClientTypeSingle{A}}{\bag{a_1^1,\dots,a_{m_1}^1,\dots,a_1^k,\dots,a_{m_k}^k}}&\text{by def. of $\uplus$}\\
                                                &=\bag{((\FDenotationsInd{A}{a_1^1},*),*),\dots,((\FDenotationsInd{A}{a_{m_k}^k},*),*)}&\text{by def.~\ref{def: transformation on denotations}}\\
                                                &=\uplus^k_{j=1}\bag{((\FDenotationsInd{A}{a_1^j},*),*),\dots,((\FDenotationsInd{A}{a_{m_j}^j},*),*)}&\text{by def. of $\uplus$}\\
                                                 &=\uplus^k_{j=1}\FDenotationsInd{\ClientTypeSingle{A}}{\alpha_j}&\text{by def.~\ref{def: transformation on denotations}}
\end{align*}
\end{proof}
As explained above, synchronizers mediate between the translation of two processes. The following lemma ensures that a synchronizer 
$\Synchronizer{w,z}{A}$ acts as a forwarder:
\begin{lemmaapxrep}\label{lemma: denotations of synchronizers}
Let $\Delta = z: \TranLau{A}\Tensor\bot,w :\TranLau{\Dual{A}}\Tensor\bot, s:\unit$, for some 
    $A\in \CPMixZero$.
    Then \\$\AtkeyDenotations{\Synchronizer{z,w}{A}\CPJud \Delta}=\{((\FDenotationsInd{A}{a},*),(\FDenotationsInd{\Dual{A}}{a},*),*) \para a\in \AtkeyDenotations{A}\}$.
 \end{lemmaapxrep}

    \begin{proofsketch}
        The proof follows by induction on $A$.
    \end{proofsketch}

\begin{proof}
    The proof follows by induction on the structure of  $A$.
    \begin{itemize}
        \item Base case: $\unit$.
          By \Cref{def: synchronizers}, \[\Synchronizer{w,z}{\unit}=\boutt{m_2}{w}(\impp{m_2}{z_1}\boutt{m_2'}{z}(\forward{m_2'}{m_2}\para\forward{z_1}{z} )\para \forward{w}{s})\]
 By \Cref{fig:DenotationalSemantics}:

\begin{align*}
    \AtkeyDenotations{\forward{m_2'}{m_2}\CPJud m_2':\unit, m_2:\bot}&=
    \AtkeyDenotations{\forward{z_1}{z}\CPJud z_1:\unit, z:\bot}\\
    &=\AtkeyDenotations{\forward{w}{s}\CPJud s:\unit, w:\bot}\\
     &=\{(*,*)\}
     \end{align*}
 We calculate the denotations of $\Synchronizer{w,z}{\unit}$  by calculating the denotations of its subproceses as follows:
     \begin{align*}
(1)\quad &\AtkeyDenotations{\boutt{m_2'}{z}(\forward{m_2'}{m_2}\para\forward{z_1}{z} )\CPJud z_1:\unit, m_2:\bot, z:\unit\Tensor\bot}=\{((*,*),*,*)\}\\
(2)\quad &\AtkeyDenotations{\underbrace{\impp{m_2}{z_1}\boutt{m_2'}{z}(\forward{m_2'}{m_2}\para\forward{z_1}{z} )}_{R}\CPJud m_2:\bot\parrType\unit, z:\unit\Tensor\bot}=\{((*,*),(*,*))\}\\
(3)\quad &\AtkeyDenotations{\boutt{m_2}{w}(R\para \forward{w}{s})\CPJud w:(\bot\parrType\unit)\Tensor\bot,z:\unit\Tensor\bot,s:\unit}=\{(((*,*),*),(*,*),*)\}\\
\phantom{\Rightarrow}\quad &\phantom{\AtkeyDenotations{\boutt{m_2}{w}(R\para \forward{w}{s})\CPJud w:(\bot\parrType\unit)\Tensor\bot,z:\unit\Tensor\bot,s:\unit}}=\{((\FDenotationsInd{\unit}{*},*),(\FDenotationsInd{\bot}{*},*),*)\}
\end{align*}
   
\item Case $\ServerTypeSingle{A}$.
By \Cref{def: synchronizers}:
\begin{align*}
    \Synchronizer{w,z}{\ServerTypeSingle{A}}=\boutt{m}{z}(\impp{m}{m_2}\boutt{w_2'}{w}(\Server{w_2'}{w'}\Client{m_2}{m_2'}\impp{w'}{m_1'}\Synchronizer{m_1',m_2'}{A}\para \forward{w}{m})\para \forward{z}{s})
\end{align*}

By I.H.:
\begin{align*}
        &\AtkeyDenotations{\Synchronizer{m_1',m_2'}{A}\CPJud m_2':\TranLau{A}\Tensor\bot,m_1':\TranLau{\Dual{A}}\Tensor\bot, w':\unit}\\
        &=\{((\FDenotationsInd{A}{a},*),(\FDenotationsInd{\Dual{A}}{a},*),*) \para a\in \AtkeyDenotations{A}\}
    \end{align*}
Thus,  we calculate the denotations of $\Synchronizer{w,z}{\ServerTypeSingle{A}}$ by repeatedly applying the rules in \Cref{fig:DenotationalSemantics}. Essentially, in each step we calculate the denotations of a larger subprocess of $\Synchronizer{w,z}{\ServerTypeSingle{A}}$. We proceed as follows:
\begin{align*}
 (1)
\quad  &\AtkeyDenotations{\impp{w'}{m_1'}\Synchronizer{m_1',m_2'}{A}\CPJud w':(\TranLau{\Dual{A}}\Tensor\bot)\parrType\unit, m_2':\TranLau{A}\Tensor\bot}\\
    &=\{((\FDenotationsInd{A}{a},*),((\FDenotationsInd{\Dual{A}}{a},*),*)) \para \\
    &\qquad((\FDenotationsInd{A}{a},*),(\FDenotationsInd{\Dual{A}}{a},*),*) \in \AtkeyDenotations{\Synchronizer{m_1',m_2'}{A}\CPJud m_2':\TranLau{A}\Tensor\bot,m_1':\TranLau{\Dual{A}}\Tensor\bot, w':\unit}\}\\
(2)
\quad&\AtkeyDenotations{\underbrace{\Client{m_2}{m_2'}\impp{w'}{m_1'}\Synchronizer{m_1',m_2'}{A}}_{R_1}\CPJud w':(\TranLau{\Dual{A}}\Tensor\bot)\parrType\unit, m_2:\ClientType{\TranLau{A}\Tensor\bot}}\\
 &=\{(\bag{(\FDenotationsInd{A}{a},*)},((\FDenotationsInd{\Dual{A}}{a},*),*)) \para \\
  &\qquad((\FDenotationsInd{A}{a},*),((\FDenotationsInd{\Dual{A}}{a},*),*)) \in 
  \AtkeyDenotations{\impp{w'}{m_1'}\Synchronizer{m_1',m_2'}{A}\CPJud w':(\TranLau{\Dual{A}}\Tensor\bot)\parrType\unit, m_2':\TranLau{A}\Tensor\bot}\}\\
  (3)
\quad&\AtkeyDenotations{\Server{w_2'}{w'}R_1\CPJud w_2':\ServerType{(\TranLau{\Dual{A}}\Tensor\bot)\parrType\unit},m_2:\ClientType{\TranLau{A}\Tensor\bot} }\\
 &=\{ (\uplus_{j=1}^k(\bag{(\FDenotationsInd{A}{a_j},*)}, \bag{(\FDenotationsInd{\Dual{A}}{a_1},*),*),\dots,(\FDenotationsInd{\Dual{A}}{a_k},*),*)})\para\\
&\qquad\forall i\in \{1,\dots,k\}. (\bag{(\FDenotationsInd{A}{a_i},*)},a_i)\in \AtkeyDenotations{R_1\CPJud\CPJud w':(\TranLau{\Dual{A}}\Tensor\bot)\parrType\unit, m_2:\ClientType{\TranLau{A}\Tensor\bot}}\}\\
(4)
\quad&\AtkeyDenotations{\underbrace{\boutt{w_2'}{w}(\Server{w_2'}{w'}R_1\para \forward{w}{m})}_{R_2}\CPJud w:(\ServerType{(\TranLau{\Dual{A}}\Tensor\bot)\parrType\unit})\Tensor\bot,m_2:\ClientType{\TranLau{A}\Tensor\bot},m:\unit }\\
 &=\{ (\uplus_{j=1}^k(\bag{(\FDenotationsInd{A}{a},*)}, (\bag{(\FDenotationsInd{\Dual{A}}{a_1},*),*),\dots,(\FDenotationsInd{\Dual{A}}{a_k},*),*)},*),*)\para\\
&\qquad(\uplus_{j=1}^k(\bag{(\FDenotationsInd{A}{a},*)}, \bag{(\FDenotationsInd{\Dual{A}}{a_1},*),*),\dots,(\FDenotationsInd{\Dual{A}}{a_k},*),*)})\\
&\qquad\in \AtkeyDenotations{\Server{w_2'}{w'}R_1\CPJud w_2':\ServerType{(\TranLau{\Dual{A}}\Tensor\bot)\parrType\unit},m_2:\ClientType{\TranLau{A}\Tensor\bot} }\}\\
(5)
\quad&\AtkeyDenotations{\impp{m}{m_2}R_2\CPJud w:(\ServerType{(\TranLau{\Dual{A}}\Tensor\bot)\parrType\unit})\Tensor\bot,m:(\ClientType{\TranLau{A}\Tensor\bot})\parrType\unit}\\
& =\{  (((\uplus_{j=1}^k\bag{(\FDenotationsInd{A}{a},*)},*), (\bag{(\FDenotationsInd{\Dual{A}}{a_1},*),*),\dots,(\FDenotationsInd{\Dual{A}}{a_k},*),*)},*),*)\para \\
&\qquad \in \AtkeyDenotations{R_2\CPJud w:(\ServerType{(\TranLau{\Dual{A}}\Tensor\bot)\parrType\unit})\Tensor\bot,m_2:\ClientType{\TranLau{A}\Tensor\bot},m:\unit}\}\\
(6)
\quad& \AtkeyDenotations{\boutt{m}{z}(\impp{m}{m_2}R_2\para \forward{z}{s})\CPJud w:(\ServerType{(\TranLau{\Dual{A}}\Tensor\bot)\parrType\unit})\Tensor\bot,z:((\ClientType{\TranLau{A}\Tensor\bot})\parrType\unit)\bot,s:\unit}\\
&=\{  ((((\uplus_{j=1}^k\bag{(\FDenotationsInd{A}{a_j},*)},*),*), (\bag{(\FDenotationsInd{\Dual{A}}{a_1},*),*),\dots,(\FDenotationsInd{\Dual{A}}{a_k},*),*)},*),*)\para\\
 &\qquad(((\uplus_{j=1}^k\bag{(\FDenotationsInd{A}{a_j},*)},*), (\bag{(\FDenotationsInd{\Dual{A}}
{a_1},*),*),\dots,(\FDenotationsInd{\Dual{A}}{a_k},*),*)},*),*) \\
&\qquad\in  
 \AtkeyDenotations{\impp{m}{m_2}R_2\CPJud w:(\ServerType{(\TranLau{\Dual{A}}\Tensor\bot)\parrType\unit})\Tensor\bot,m:(\ClientType{\TranLau{A}\Tensor\bot})\parrType\unit}\}\\
 &=\{((((\uplus_{j=1}^k\bag{(\FDenotationsInd{A}{a_j},*)},*),*), (\bag{(\FDenotationsInd{\Dual{A}}{a_1},*),*),\dots,(\FDenotationsInd{\Dual{A}}{a_k},*),*)},*),*)\para\\
 &\qquad \forall j\in\{1,\dots,k\}.a_j\in \AtkeyDenotations{A}\}\\
 &=\{  ((((\FDenotationsInd{\ServerTypeSingle{A}}{\bag{a_1,\dots,a_k}},*) ((\FDenotationsInd{\ClientType{\Dual{A}}}{\bag{a_1,\dots,a_k}},*),*)\para\\
 &\qquad \bag{a_1,\dots,a_k} \in \AtkeyDenotations{\ServerTypeSingle{A}}\}\tag{by \Cref{def: transformation on denotations,lemma: union of multisets}}
\end{align*}

What concludes the proof since $\Synchronizer{w,z}{\ServerTypeSingle{A}}=\boutt{m}{z}(\impp{m}{m_2}R_2\para \forward{z}{s})$.

The other cases proceed similarly.
     \end{itemize}
\end{proof}
\noindent    
The next lemma is crucial to ensure that the denotations of a composed process correspond  (in the sense of  \Cref{def: transformation on denotations}) with those of its translation:  
\begin{lemmaapx}[Synchronizers are well-behaved]\label{theorem: synchronizers are well behaved}
Let $P,P',Q,Q'\in\CPMixZeroTwo$, such that
\begin{align*}
    \AtkeyDenotations{P'\CPJud\Delta',x':\DualTranLau{A},w:\unit}&=\FDenotations{\Delta,x:A}{\AtkeyDenotations{P\CPJud\Delta,x:A}}\\
    \AtkeyDenotations{Q'\CPJud\Gamma',x':\DualTranLau{\Dual{A}},z:\unit}&=\FDenotations{\Gamma,y:\Dual{A}}{\AtkeyDenotations{Q\CPJud\Gamma, x:\Dual{A}}}
\end{align*}
Then:
 $(\delta,\gamma)\in\AtkeyDenotations{\cut{x}{P}{Q}} \; \Leftrightarrow \; (\FDenotationsInd{\Delta}{\delta},\FDenotationsInd{\Gamma}{\gamma},*)\in \AtkeyDenotations{\cut{w}{\impp{w}{x'}P'}{\cut{z}{\impp{z}{x'}Q'}{\Synchronizer{z,w}{A}}}}$.
\end{lemmaapx}
\begin{proof}
    \begin{align*}
        &(\delta,\gamma)\in \AtkeyDenotations{\cut{x}{P}{Q}}\\
        &\;\Leftrightarrow\; (\delta,a)\in \AtkeyDenotations{P}\wedge (\gamma,a)\in \AtkeyDenotations{Q}\tag{by \Cref{fig:DenotationalSemantics}}\\
        &\;\Leftrightarrow\; (\FDenotationsInd{\Delta}{\delta},\FDenotationsInd{A}{a},*)\in\AtkeyDenotations{P'}\wedge (\FDenotationsInd{\Gamma}{\gamma},\FDenotationsInd{\Dual{A}}{a},*) \in \AtkeyDenotations{Q'}\tag{by assumption}\\
          &\;\Leftrightarrow\; (\FDenotationsInd{\Delta}{\delta},(\FDenotationsInd{A}{a},*))\in\AtkeyDenotations{\impp{w}{x'}P'}\wedge (\FDenotationsInd{\Gamma}{\gamma},(\FDenotationsInd{\Dual{A}}{a},*)) \in \AtkeyDenotations{\impp{z}{x'}Q'}\tag{by \Cref{fig:DenotationalSemantics}}\\
          &\;\Leftrightarrow\; (\FDenotationsInd{\Delta}{\delta},(\FDenotationsInd{\Dual{A}}{a},*),* ) \in \AtkeyDenotations{\cut{w}{\impp{w}{x'}P'}{\Synchronizer{z,w}{A}}}\\
          &\quad\quad\quad\quad\quad\wedge (\FDenotationsInd{\Gamma}{\gamma},(\FDenotationsInd{\Dual{A}}{a},*)) \in \AtkeyDenotations{\impp{z}{x'}Q'}\tag{by \Cref{lemma: denotations of synchronizers}}\\
          & \; \Leftrightarrow \; (\FDenotationsInd{\Delta}{\delta},\FDenotationsInd{\Gamma}{\gamma},*)\in \AtkeyDenotations{\cut{w}{\impp{w}{x'}P'}{\cut{z}{\impp{z}{x'}Q'}{\Synchronizer{z,w}{A}}}}\tag{by \Cref{fig:DenotationalSemantics}}
    \end{align*}
\end{proof}
\noindent The next result, \Cref{theorem: laurent transformation on denotations}, states that the lifting of Laurent's transformation $\DualTranLau{(-)}$ 
to the level of denotations is correct.

\begin{theoremapxrep}\label{theorem: laurent transformation on denotations} 
Let $P\CPJud \Gamma$ be a $\CPMixZero$  process. Then 
    $\FDenotations{\Gamma}{\AtkeyDenotations{P\CPJud\Gamma}}=\AtkeyDenotations{\TranLauP{P}\CPJud\DualTranLau{\Gamma},w:\unit}$.
\end{theoremapxrep}
\begin{proofsketch}
    By induction on the structure of $P\CPJud \Gamma$, with a case analysis in the last rule applied. 
    We give a representative case. Consider $\Server{x}{y}P\CPJud\ClientTypeSingle{\Delta},x:\ServerTypeSingle{A}$, with $\Delta=x_1:A_1,\dots,x_n:A_n$.
    In one direction, we apply \Cref{lemma: union of multisets,def: transformation on denotations} to show
    \begin{align*}
           &\FDenotations{\ClientTypeSingle{\Delta},x:\ServerTypeSingle{A}}{\AtkeyDenotations{\Server{x}{y}P\CPJud{\ClientTypeSingle{\Delta}},x:\ServerTypeSingle{A}}}\\
           &=\{\FDenotations{\ClientTypeSingle{\Delta},x:\ServerType{A}}{\uplus^k_{j=1}\alpha_j^1,\cdots,\uplus^k_{j=1}\alpha_j^n,\Lbag a_1,,\cdots,a_k\Rbag}\para\\
           &\qquad\forall i\in \{1,\cdots,k\}. (\alpha_i^1,\cdots,\alpha_i^n,a_i)\in \AtkeyDenotations{P\CPJud \ClientTypeSingle{\Delta},y:A}\}\\
           &=\{(\uplus^k_{j=1}\beta_j^1,\cdots,\uplus^k_{j=1}\beta_j^n,(\Lbag (b_1,*),\cdots,(b_k,*)\Rbag,*),*)\para
           \tag*{with $\FDenotationsInd{A_i}{\alpha^i}=\beta^i$,}\\ 
     &\qquad \quad \forall i\in \{1,\cdots,k\}. (\alpha_i^1,\cdots,\alpha_i^n,a_i,*)\in \AtkeyDenotations{P\CPJud\ClientTypeSingle{\Delta},y:A}\} \tag{and $\FDenotationsInd{A}{a}=b$}
    \end{align*}
In the other direction, by the I.H., we obtain:
\begin{align*}
     &\AtkeyDenotations{\TranLauP{\Server{x}{y}P}\CPJud\DualTranLau{(\ClientTypeSingle{\Delta})}, m:\DualTranLau{(\ServerTypeSingle{A})},w:\unit}\\
      &\quad=\{(\uplus^k_{j=1}\beta_j^1,\cdots,\uplus^k_{j=1}\beta_j^n,(\Lbag (b_1,*),\cdots,(b_k,*)\Rbag,*),*)\para\\ 
     &\qquad \quad \forall i\in \{1,\cdots,k\}. (\alpha_i^1,\cdots,\alpha_i^n,a_i,*)\in \AtkeyDenotations{P\CPJud{\ClientTypeSingle{\Delta}},y:A}\}
\end{align*}
when the last rule is cut, we rely on I.H. and \Cref{theorem: synchronizers are well behaved}.
\end{proofsketch}
\begin{proof}

        By induction on the structure of On the typing derivation by case analysis in the last rule applied.

        \begin{itemize}
        \item Case $\Emptyoutt{x}\CPJud x:\unit$.
        By \Cref{def: Tranformation of processes},
        $\TranLauP{\Emptyoutt{x}}=\boutt{x}{m}(\Emptyoutt{x}\para \fwd{m}{n})$, and 
        by \Cref{fig:DenotationalSemantics}:
        \begin{align*}
        \AtkeyDenotations{\Emptyoutt{x}\CPJud x:\unit}&=\{(*)\} \\
        \AtkeyDenotations{\fwd{m}{n}\CPJud m:\bot,n:\unit}&=\{(*,*)\} 
        \end{align*}
    Thus, the equality is checked as follows: 
        \begin{align*}
            \AtkeyDenotations{\boutt{x}{m}(\Emptyoutt{x}\para \fwd{m}{n})\CPJud m :\unit\Tensor \bot, n:\unit}
            &=\{((*,*),*)\}\\
            &=\{(\FDenotationsInd{1}{*},*)\}\\
            &=\{\FDenotations{1}{*}\}\\
            &=\FDenotations{1}{\AtkeyDenotations{\Emptyoutt{x}\CPJud x:\unit}}
            \end{align*}
            
            \item Case  $\fwd{x}{y}\CPJud x:A, y:A^\bot$. 
By \Cref{def: Tranformation of processes}, $\TranLauP{\fwd{x}{y}} =\boutt{x}{m}(\fwd{x}{y}\para \fwd{m}{n})$. By \Cref{fig:DenotationalSemantics} we have:
            \begin{align*}
            \AtkeyDenotations{\fwd{x}{y}\CPJud x:A^\bot,y:A}&=\{(a,a)\para a\in\AtkeyDenotations{A}\}\\
              \AtkeyDenotations{\fwd{m}{n}\CPJud m:\bot,n:\unit}&=\{(*,*)\}
            \end{align*}
Thus, we calculate the denotations of $\boutt{x}{m}(\fwd{x}{y}\para \fwd{m}{n})$ as follows:
            \begin{align*}
            &\AtkeyDenotations{\TranLauP{\fwd{x}{y}}\CPJud m: A\Tensor \bot, y:A^\bot,n:\unit}\\
                &=\AtkeyDenotations{\boutt{x}{m}(\fwd{x}{y}\para \fwd{m}{n})\CPJud m: A\Tensor \bot, y:A^\bot,n:\unit}\\
                &=\{(a,(a,*),*) \para (a,a) \in \AtkeyDenotations{\fwd{x}{y}\CPJud x:A^\bot,y:A},(*,*)\in \AtkeyDenotations{\fwd{m}{n}\CPJud m:\bot,n:\unit}\}\\
                &=\{(\FDenotationsInd{A^\bot}{a},(\FDenotationsInd{A}{a},*)\para a\in\AtkeyDenotations{A}\}\\
               &=\FDenotations{x:^\bot,y:A}{\AtkeyDenotations{\fwd{x}{y}\CPJud x:A^\bot,y:A}} \tag{by \Cref{def: transformation on denotations}}
            \end{align*}

\item Case $\Emptyimpp{x}P$.
By \Cref{def: Tranformation of processes}, $\TranLauP{\Emptyimpp{x}P}=\Emptyimpp{x}\TranLauP{P}$.
We calculate the denotations of $\Emptyimpp{x}\TranLauP{P}$ as follows:
 \begin{align*}
     &\AtkeyDenotations{\Emptyimpp{x}\TranLauP{P}\CPJud\DualTranLau{\Delta},x:\bot,w:\unit}\\
     &=\{(\FDenotationsInd{\Delta}{\delta},*,*) \para (\FDenotationsInd{\Delta}{\delta},*) \in \AtkeyDenotations{\TranLauP{P}\CPJud\DualTranLau{\Delta},w:\unit}\}\\
     &=\{(\FDenotationsInd{\Delta}{\delta},\FDenotationsInd{\bot}{*},*) \para (\FDenotationsInd{\Delta}{\delta},*) \in \AtkeyDenotations{\TranLauP{P}\CPJud\DualTranLau{\Delta},w:\unit}\}\tag{by I.H.}\\
     &=\{(\FDenotationsInd{\Delta}{\delta},\FDenotationsInd{\bot}{*},*) \para {\delta}\in\AtkeyDenotations{P\CPJud \Delta}\}\\
      &=\{(\FDenotationsInd{\Delta}{\delta},\FDenotationsInd{\bot}{*},*) \para ({\delta},*)\in\AtkeyDenotations{\Emptyimpp{x}P\CPJud \Delta}\}\\
      &=\FDenotations{\Delta,x:\bot}{\AtkeyDenotations{\Emptyimpp{x}P\CPJud\Delta,x:\bot}}
 \end{align*}
    
\item Case $\boutt{y}{x}(\textcolor{Black}{P_1} \para \textcolor{Black}{P_2})\CPJud\Delta,\Gamma, x:A\Tensor B$.
By \Cref{def: Tranformation of processes}, we have: 
\[\TranLauP{\boutt{y}{x}(\textcolor{Black}{P_1}\para \textcolor{Black}{P_2})}=\boutt{z_2}{z}(\boutt{z_1}{z_2}(\impp{z_1}{y} \TranLauP{\textcolor{Black}{P_1}}\para\impp{z_2}{x} \TranLauP{\textcolor{Black}{P_2}})\para \fwd{z}{w})\]
By I.H., we have:
                 \begin{align*}
    \FDenotations{\Gamma,y:A}{\AtkeyDenotations{\textcolor{Black}{P_1}\CPJud\Gamma, y:A}}&=\AtkeyDenotations{\TranLauP{\textcolor{Black}{P_1}}\CPJud\DualTranLau{\Gamma},y:\DualTranLau{A},\unit}\\
    \FDenotations{\Delta,x:B}{\AtkeyDenotations{\textcolor{Black}{P_2}\CPJud\Delta,x:B}}&=\AtkeyDenotations{\TranLauP{\textcolor{Black}{P_2}}\CPJud\DualTranLau{\Delta},x:\DualTranLau{B},\unit}
\end{align*}
            Thus, for $D=(({\DualTranLau{A}}\parrType 1)\Tensor({\DualTranLau{B}}\parrType 1))\Tensor \bot$, by repeatedly applying the rules in \Cref{fig:DenotationalSemantics} and I.H. we have: 
               \begin{align*}
      &\AtkeyDenotations{\TranLauP{P}\CPJud \DualTranLau{\Delta},\DualTranLau{\Gamma},z:D,w:\unit}\\
      &=\{(\FDenotationsInd{\Gamma}{\gamma},\FDenotationsInd{\Delta}{\delta},(((\FDenotationsInd{A}{a},*),(\FDenotationsInd{B}{b},*)),*),*) \para\\
       &\qquad(\FDenotationsInd{\Gamma}{\gamma},\FDenotationsInd{A}{a},*)\in\AtkeyDenotations{\TranLauP{\textcolor{Black}{P_1}}\CPJud \DualTranLau{\Delta}, y:\DualTranLau{A},\unit},\\
            &\qquad(\FDenotationsInd{\Delta}{\delta},\FDenotationsInd{B}{b},*)\in \AtkeyDenotations{\TranLauP{\textcolor{Black}{P_2}}\CPJud \DualTranLau{\Gamma}, x:\DualTranLau{B},\unit}\}\tag{by \Cref{fig:DenotationalSemantics}}
            \end{align*}
Now we show that those denotations coincide with those in $$\FDenotations{\Delta,\Gamma,A\Tensor B}{ \AtkeyDenotations{\boutt{y}{x}(\textcolor{Black}{P_1}\para \textcolor{Black}{P_2})\CPJud\Gamma,\Delta,x: A\Tensor B}}$$ We proceed as follows:
      \begin{align*}
      &\AtkeyDenotations{\TranLauP{P}\CPJud \DualTranLau{\Delta},\DualTranLau{\Gamma},z:D,w:\unit}\\
      &=\{(\FDenotationsInd{\Gamma}{\gamma},\FDenotationsInd{\Delta}{\delta},(((\FDenotationsInd{A}{a},*),(\FDenotationsInd{B}{b},*)),*),*) \para\\
       &\qquad(\FDenotationsInd{\Gamma}{\gamma},\FDenotationsInd{A}{a},*)\in\AtkeyDenotations{\TranLauP{\textcolor{Black}{P_1}}\CPJud \DualTranLau{\Delta}, y:\DualTranLau{A},\unit},
       \\
            &\qquad(\FDenotationsInd{\Delta}{\delta},\FDenotationsInd{B}{b},*)\in \AtkeyDenotations{\TranLauP{\textcolor{Black}{P_2}}\CPJud \DualTranLau{\Gamma}, x:\DualTranLau{B},\unit}\}\tag{by \Cref{fig:DenotationalSemantics}}
            \\
             &=\{(\FDenotationsInd{\Gamma}{\gamma},\FDenotationsInd{\Delta}{\delta},(((\FDenotationsInd{A}{a},*),(\FDenotationsInd{B}{b},*)),*),*) \para\\
       &\qquad(\FDenotationsInd{\Gamma}{\gamma},\FDenotationsInd{A}{a},*)\in \FDenotations{\Gamma,y:A}{\AtkeyDenotations{\textcolor{Black}{P_1}\CPJud\Gamma, y:A}} ,\\
            &\qquad(\FDenotationsInd{\Delta}{\delta},\FDenotationsInd{B}{b},*)\in \FDenotations{\Delta,x:B}{\AtkeyDenotations{\textcolor{Black}{P_2}\CPJud\Delta,x:B}} \}\tag{by the I.H.}\\
            &=\{(\FDenotationsInd{\Gamma}{\gamma},\FDenotationsInd{\Delta}{\delta},\FDenotationsInd{A\Tensor B}{a,b},*)\para\\
       &\qquad(\FDenotationsInd{\Gamma}{\gamma},\FDenotationsInd{A}{a},*)\in \FDenotations{\Gamma,y:A}{\AtkeyDenotations{\textcolor{Black}{P_1}\CPJud\Gamma, y:A}} ,\\
            &\qquad(\FDenotationsInd{\Delta}{\delta},\FDenotationsInd{B}{b},*)\in \FDenotations{\Delta,x:B}{\AtkeyDenotations{\textcolor{Black}{P_2}\CPJud\Delta,x:B}} \}\tag{by \Cref{def: transformation on denotations}}\\
        &=\FDenotations{\Delta,\Gamma,A\Tensor B}{ \AtkeyDenotations{\boutt{y}{x}(\textcolor{Black}{P_1}\para \textcolor{Black}{P_2})\CPJud\Gamma,\Delta,x: A\Tensor B}}
\end{align*}

\item Case $\impp{x}{y}P\CPJud\Delta,x:A\parrType B$. By \Cref{def: Tranformation of processes}, $\TranLauP{\impp{x}{y}P}=\impp{x}{y}\TranLauP{P}$.
          We calculate the denotations of $\TranLauP{\impp{x}{y}P}$ as follows:
          \begin{align*}
    &\AtkeyDenotations{\impp{x}{y}\TranLauP{P}\CPJud x:\DualTranLau{B}\parrType\DualTranLau{A},\DualTranLau{\Delta},w:\unit}\\
&=\{(\FDenotationsInd{\Delta}{\delta},(\FDenotationsInd{A}{a},\FDenotationsInd{B}{b}),*)\para  (\FDenotationsInd{\Delta}{\delta},\FDenotationsInd{A}{a},\FDenotationsInd{B}{b})\in \AtkeyDenotations{\TranLauP{P}\CPJud x:\DualTranLau{B}, y:\DualTranLau{A},\DualTranLau{\Delta},w:\unit} \}\tag{by \Cref{fig:DenotationalSemantics} }\\
&=\{(\FDenotationsInd{\Delta}{\delta},(\FDenotationsInd{A}{a},\FDenotationsInd{B}{b}),*)\para \FDenotations{\Delta,x:B,y:A}{\AtkeyDenotations{P\CPJud x:B,y:A,\Delta}} \}\tag{by the I.H.}\\
&= \FDenotations{\Delta,x:A\parrType B}{\AtkeyDenotations{\impp{x}{y}P\CPJud \Delta,x:A\parrType B}}
          \end{align*}

\item Case $\Server{x}{y}P\CPJud\ClientTypeSingle{\Delta},x:\ServerTypeSingle{A}$.
    By \Cref{def: Tranformation of processes},
    \[\TranLauP{\Server{x}{y}P}=\boutt{x}{m}(\Server{x}{w}\impp{w}{y}\TranLauP{P}\para \fwd{m}{n})\]
By the I.H.:
\[\AtkeyDenotations{\TranLauP{P}\CPJud \DualTranLau{(\ClientTypeSingle{\Delta})},y:\DualTranLau{A},\unit}=\FDenotations{\ClientTypeSingle{\Delta},y:A}{\AtkeyDenotations{P\CPJud\ClientTypeSingle{\Delta},y:A}}\]
where $\Delta =x_1:A_1,\dots,x_n:A_n$. Thus:
\begin{align*}
    \AtkeyDenotations{\TranLauP{P} \CPJud \DualTranLau{(\ClientTypeSingle{\Delta})},y:\DualTranLau{A},\unit} & =\{(\FDenotationsInd{A_1}{\alpha^1},\dots,\FDenotationsInd{A_n}{\alpha^n},\FDenotationsInd{A}{a},*) \para \\
   & \qquad (\alpha^1,\dots,\alpha^n,a)\in \AtkeyDenotations{P\CPJud\ClientTypeSingle{\Delta},y:A}\}
    \end{align*}
First we calculate the denotations of $\TranLauP{\Server{x}{y}P}$ by calculating the denotations of its subprocesses $\impp{w}{y}\TranLauP{P}$ and $\Server{x}{w}\impp{w}{y}\TranLauP{P}$ as follows:

      \begin{align*}
        &\AtkeyDenotations{\impp{w}{y}\TranLauP{P}\CPJud,\DualTranLau{?(\Delta)},w:\DualTranLau{A}\parrType1}\\
        &=\{(\beta^1,\cdots,\beta^n,(b,*)) \para (\beta^1,\cdots,\beta^n,b,*)\in  \AtkeyDenotations{\TranLauP{P}\CPJud,?\DualTranLau{\Delta},y:\DualTranLau{A}, w:\unit}\}
             \end{align*}
  where $\FDenotationsInd{A_i}{\alpha^i}=\beta^i$ and $\FDenotationsInd{A}{a}=b$;

   \begin{align*}
            &\AtkeyDenotations{\Server{x}{w}\impp{w}{y}\TranLauP{P}\CPJud,?\DualTranLau{\Delta},x:\ServerType{\DualTranLau{A}\parrType\unit}}\\
            &=\{(\uplus^k_{j=1}\beta_j^1,\cdots,\uplus^k_{j=1}\beta_j^n,\Lbag (b_1,*),\cdots,(b_k,*)\Rbag)\para \\
            &\qquad\forall i\in \{1,\cdots,k\}. (\beta_i^1,\cdots,\beta_i^n,(b_i,*))\in \AtkeyDenotations{\impp{w}{y}\TranLauP{P}\CPJud,?\DualTranLau{\Delta},w:\DualTranLau{A}\parrType\unit}\}
            \end{align*}
 where $\uplus^k_{j=1}\beta_j^i=\bag{\beta_1^i,\dots,\beta_k^i}$ for all $i=1,\dots,n$.

\smallskip

Finally,  the denotations of $\TranLauP{\Server{x}{y}P}$ are calculated as follows::
                    \begin{align*}
    &\AtkeyDenotations{\TranLauP{\Server{x}{y}P}\CPJud\DualTranLau{(\ClientTypeSingle{\Delta})}, m:\DualTranLau{(\ServerTypeSingle{A})},w:\unit}&\\
    &=\AtkeyDenotations{\boutt{x}{m}(\Server{x}{w}\impp{w}{y}\TranLauP{P}\para \fwd{m}{n})\CPJud\DualTranLau{(\ClientTypeSingle{\Delta})}, m:\DualTranLau{(\ServerTypeSingle{A})}, n:\unit} &\text{by \Cref{def: Tranformation of processes}} \\
    &=\{(\uplus^k_{j=1}\beta_j^1,\cdots,\uplus^k_{j=1}\beta_j^n,(\Lbag (b_1,*),\cdots,(b_k,*)\Rbag,*),*)\para&\\ 
    &\qquad(\uplus^k_{j=1}\beta_j^1,\cdots,\uplus^k_{j=1}\beta_j^n,\Lbag (b_1,*),\cdots,(b_k,*)\Rbag)&\\
    &\qquad\in \AtkeyDenotations{\Server{x}{w}\impp{w}{y}\TranLauP{P}\CPJud\DualTranLau{(\ClientTypeSingle{\Delta})},x:\ServerType{\DualTranLau{A}\parrType\unit}},\\
    & \qquad(*,*)\in \AtkeyDenotations{\fwd{m}{n}\CPJud m:\bot,n:\unit}\} &\text{by \Cref{fig:DenotationalSemantics}} \\
     &=\{(\uplus^k_{j=1}\beta_j^1,\cdots,\uplus^k_{j=1}\beta_j^n,(\Lbag (b_1,*),\cdots,(b_k,*)\Rbag,*),*)\para&\\ 
     &\qquad\forall i\in \{1,\cdots,k\}. (\beta_i^1,\cdots,\beta_i^n,b_i,*)\in \AtkeyDenotations{\TranLauP{P}\CPJud ?\DualTranLau{\Delta},y:\DualTranLau{A}, w:\unit}\}&\\
&=\{(\uplus^k_{j=1}\beta_j^1,\cdots,\uplus^k_{j=1}\beta_j^n,(\Lbag (b_1,*),\cdots,(b_k,*)\Rbag,*),*)\para&\\ 
&\qquad\forall i\in \{1,\cdots,k\}. (\alpha_i^1,\cdots,\alpha_i^n,a_i,*)\in \AtkeyDenotations{P\CPJud \ClientTypeSingle{\Delta},y:A}\}&\text{by the I.H.}
\end{align*}
Now we calculate $\FDenotations{\ClientTypeSingle{\Delta},x:\ServerTypeSingle{A}}{\AtkeyDenotations{\Server{x}{y}P\CPJud\ClientTypeSingle{\Delta},x:\ServerTypeSingle{A}}}$ to show that we obtain the same set of denotations. 
\begin{align*}
           &\FDenotations{\ClientTypeSingle{\Delta},x:\ServerTypeSingle{A}}{\AtkeyDenotations{\Server{x}{y}P\CPJud\ClientTypeSingle{\Delta},x:\ServerTypeSingle{A}}}\\
           &=\{\FDenotations{\ClientTypeSingle{\Delta},x:\ServerTypeSingle{A}}{\uplus^k_{j=1}\alpha_j^1,\cdots,\uplus^k_{j=1}\alpha_j^n,\Lbag a_1,,\cdots,a_k\Rbag}\para\\
           &\qquad\forall i\in \{1,\cdots,k\}. (\alpha_i^1,\cdots,\alpha_i^n,a_i)\in \AtkeyDenotations{P\CPJud \ClientTypeSingle{\Delta},y:A}\}\\
           &=\{(\FDenotations{\ClientTypeSingle{A}'_1}{\uplus^k_{j=1}\alpha_j^1},\cdots,\FDenotationsInd{\ClientTypeSingle{A}'_n}{\uplus^k_{j=1}\alpha_j^n},\FDenotationsInd{\ServerTypeSingle{A}}{\bag{a_1,\cdots,a_k}},*)\para\\
           &\qquad\forall i\in \{1,\cdots,k\}. (\alpha_i^1,\cdots,\alpha_i^n,a_i)\in \AtkeyDenotations{P\CPJud \ClientTypeSingle{\Delta},y:A}\} \tag{by def~\ref{def: transformation on denotations}}\\
        &=\{(\uplus^k_{j=1}\FDenotations{\ClientTypeSingle{A}'_1}{\alpha_j^1},\cdots,\uplus^k_{j=1}\FDenotationsInd{\ClientTypeSingle{A}'_n}{\alpha_j^n},\FDenotationsInd{\ServerTypeSingle{A}}{\bag{a_1,\cdots,a_k}},*)\para\\
           &\qquad\forall i\in \{1,\cdots,k\}. (\alpha_i^1,\cdots,\alpha_i^n,a_i)\in \AtkeyDenotations{P\CPJud \ClientTypeSingle{\Delta},y:A}\} \tag{by lemma~\ref{lemma: union of multisets}}\\
           &=\{(\uplus^k_{j=1}\FDenotations{\ClientTypeSingle{A}'_1}{\alpha_j^1},\cdots,\uplus^k_{j=1}\FDenotationsInd{\ClientTypeSingle{A}'_n}{\alpha_j^n},(\Lbag (\FDenotationsInd{A}{a_1},*),\cdots,(\FDenotationsInd{A}{a_k},*)\Rbag,*),*)\para\\
           &\qquad\forall i\in \{1,\cdots,k\}. (\alpha_i^1,\cdots,\alpha_i^n,a_i)\in \AtkeyDenotations{P\CPJud \ClientTypeSingle{\Delta},y:A}\}  \tag{by def~\ref{def: transformation on denotations}}\\
           &=\{(\uplus^k_{j=1}\beta_j^1,\cdots,\uplus^k_{j=1}\beta_j^n,(\Lbag (b_1,*),\cdots,(b_k,*)\Rbag,*),*)\para
           \tag*{since $\FDenotationsInd{A_i}{\alpha^i}=\beta^i$,}\\ 
     &\qquad\forall i\in \{1,\cdots,k\}. (\alpha_i^1,\cdots,\alpha_i^n,a_i,*)\in \AtkeyDenotations{P\CPJud\ClientTypeSingle{\Delta},y:A}\} \tag{and $\FDenotationsInd{A}{a}=b$}
\end{align*}
Thus, we have proved that:
\begin{align*}
     \AtkeyDenotations{\TranLauP{\Server{x}{y}P}\CPJud\DualTranLau{(\ClientTypeSingle{\Delta})}, m:\DualTranLau{(\ServerTypeSingle{A})},w:\unit}=\FDenotations{\ClientTypeSingle{\Delta},x:\ServerTypeSingle{A}}{\AtkeyDenotations{\Server{x}{y}P\CPJud\ClientTypeSingle{\Delta},x:\ServerTypeSingle{A}}}
\end{align*}
which concludes this case.

\item Case $\Client{x}{y}P$.
By \Cref{def: Tranformation of processes}, $\TranLauP{\Client{x}{y}P}=\Client{x}{m}\boutt{w}{m}(\impp{w}{y}\TranLauP{P}\para \fwd{m}{n})$.

By I.H. we have:
\begin{align*}
\FDenotations{\Delta, y: A}{\AtkeyDenotations{P\CPJud \Delta, y: A}}
&=\{ (\FDenotationsInd{\Delta}{\delta},\FDenotationsInd{A}{a},*)\para (\delta,a)\in \AtkeyDenotations{P\CPJud \Delta , y:A}\}\\
 &=\AtkeyDenotations{\TranLauP{P}\CPJud \DualTranLau{\Gamma}, y:\DualTranLau{A},w:\unit}
\end{align*}
Now we calculate the denotations for $\TranLauP{\Client{x}{y}P}$. First, the denotations of the subprocesses $\impp{w}{y}\TranLauP{P}$ and $\boutt{w}{m}(\impp{w}{y}\TranLauP{P}\para \fwd{m}{n})$ are respectively:
\begin{align*}
    (1)\quad&\AtkeyDenotations{\impp{w}{y}\TranLauP{P}\CPJud \DualTranLau{\Delta}, w: \DualTranLau{A}\parrType 1}\\
    &=\{(\FDenotationsInd{\Delta}{\delta},(\FDenotationsInd{A}{a},*))\para (\FDenotationsInd{\Delta}{\delta},\FDenotationsInd{A}{a},*) \in \AtkeyDenotations{\TranLauP{P}\CPJud \DualTranLau{\Gamma}, y:\DualTranLau{A},w:\unit}  \}\tag{by \Cref{fig:DenotationalSemantics}}
    \\
    (2)\quad
    &\AtkeyDenotations{\boutt{w}{m}(\impp{w}{y}\TranLauP{P}\para \fwd{m}{n})\CPJud \DualTranLau{\Delta}, m: (\DualTranLau{A}\parrType 1)\Tensor \bot, n:\unit}\\
    &=\{(\FDenotationsInd{\Delta}{\delta},((\FDenotationsInd{A}{a},*),*),*)\para (\FDenotationsInd{\Delta}{\delta},(\FDenotationsInd{A}{a},*) \in  \AtkeyDenotations{\impp{w}{y}\TranLauP{P}\CPJud \DualTranLau{\Delta}, w: \DualTranLau{A}\parrType 1},\\
    &\hspace{5cm}(*,*)\in \AtkeyDenotations{\fwd{m}{n}\CPJud m:\bot,n:\unit}\}\tag{by \Cref{fig:DenotationalSemantics}}
    \end{align*}
    Finally, we calculate the denotations of $\Client{x}{m}\boutt{w}{m}(\impp{w}{y}\TranLauP{P}\para \fwd{m}{n})$, as follows:
\begin{align*}
&\AtkeyDenotations{\Client{x}{m}\boutt{w}{m}(\impp{w}{y}\TranLauP{P}\para \fwd{m}{n})\CPJud \DualTranLau{\Delta}, x: \ClientType{(\DualTranLau{A}\parrType 1)\Tensor \bot}, n:\unit}\\
    &=\{(\FDenotationsInd{\Delta}{\delta},\bag{((\FDenotationsInd{A}{a},*),*)},*)\para (\FDenotationsInd{\Delta}{\delta},((\FDenotationsInd{A}{a},*),*),*) \in \\
    &\qquad\AtkeyDenotations{\boutt{w}{m}(\impp{w}{y}\TranLauP{P}\para \fwd{m}{n})\CPJud \DualTranLau{\Delta}, m: (\DualTranLau{A}\parrType 1)\Tensor \bot, n:\unit} \}\tag{by \Cref{fig:DenotationalSemantics}}
    \end{align*}
Thus, from the  previous calculation we obtain:
    \begin{align*}
    &\AtkeyDenotations{\Client{x}{m}\boutt{w}{m}(\impp{w}{y}\TranLauP{P}\para \fwd{m}{n})\CPJud \DualTranLau{\Delta}, x: \ClientType{(\DualTranLau{A}\parrType 1)\Tensor \bot}, n:\unit}\\
    &=\{(\FDenotationsInd{\Delta}{\delta},\bag{((\FDenotationsInd{A}{a},*),*)},*)\para (\FDenotationsInd{\Delta}{\delta},\FDenotationsInd{A}{a},*) \in\AtkeyDenotations{\TranLauP{P}\CPJud \DualTranLau{\Gamma}, y:\DualTranLau{A},w:\unit}\}\\
      &=\{(\FDenotationsInd{\Delta}{\delta},\bag{((\FDenotationsInd{A}{a},*),*)},*)\para (\FDenotationsInd{\Delta}{\delta},\FDenotationsInd{A}{a},*) \in \FDenotations{\Delta, y: A}{\AtkeyDenotations{P\CPJud \Delta, y: A}} \}\tag{by the I.H.}\\
      &=\{(\FDenotationsInd{\Delta}{\delta},\bag{((\FDenotationsInd{A}{a},*),*)},*)\para (\delta,a) \in \AtkeyDenotations{P\CPJud \Delta, y: A}\}\tag{by \Cref{def: transformation on denotations}}
\end{align*}
In the other direction we have:
\begin{align*}
    &\FDenotations{\Delta,x:\ClientTypeSingle{A}}{\AtkeyDenotations{\Client{x}{y}P\CPJud \Delta x:\ClientTypeSingle{A}}}\\
 &=\{ (\FDenotationsInd{\Delta}{\delta},\FDenotationsInd{\ClientTypeSingle{A}}{\bag{a}},*)\para (\delta,\bag{a})\in\AtkeyDenotations{\Client{x}{y}P\CPJud \Delta,x:\ClientTypeSingle{A}  }\}\\
  &=\{ (\FDenotationsInd{\Delta}{\delta},\FDenotationsInd{\ClientTypeSingle{A}}{\bag{a}},*)\para (\delta,a)\in\AtkeyDenotations{P\CPJud   \Delta,y:A }\}\\
   &=\{ (\FDenotationsInd{\Delta}{\delta},\bag{((\FDenotationsInd{A}{a},*),*)},*)\para (\delta,a)\in\AtkeyDenotations{P\CPJud   \Delta,y:A }\}
\end{align*}
Thus, we have proved that:
\begin{align*}
    \FDenotations{\Delta,x:\ClientTypeSingle{A}}{\AtkeyDenotations{\Client{x}{y}P\CPJud \Delta x:\ClientTypeSingle{A}}}=\AtkeyDenotations{\TranLauP{\Client{x}{y}P}\CPJud \DualTranLau{\Delta}, x: \ClientType{(\DualTranLau{A}\parrType 1)\Tensor \bot}, n:\unit}
\end{align*}
which concludes this case.

\item Case $\case{x}{P_1}{P_2}$. By \Cref{def: Tranformation of processes},  
$\TranLauP{\case{x}{\textcolor{Black}{P_1}}{\textcolor{Black}{P_2}}}=\case{x}{\textcolor{Black}{\TranLauP{P_1}}}{\textcolor{Black}{\TranLauP{P_2}}}$.
By I.H. we have:
\begin{align*}
    \FDenotations{\Delta,x:A_1}{\AtkeyDenotations{\textcolor{Black}{P_1}\CPJud \Delta, x:A_1}}
    &=\{(\FDenotationsInd{\Delta}{\delta},\FDenotationsInd{A_1}{a_1},*)\para(\delta,a_1) \in \AtkeyDenotations{\textcolor{Black}{P_1}\CPJud \Delta, x:A_1}\}\\
    &=\AtkeyDenotations{\TranLauP{P_1}\CPJud \DualTranLau{\Delta},x:\DualTranLau{A_1},w_1:\unit}\\
    &\\
   \FDenotations{\Delta,x:A_2}{\AtkeyDenotations{\textcolor{Black}{P_2}\CPJud \Delta, x:A_2}}
    &=\{(\FDenotationsInd{\Delta}{\delta},\FDenotationsInd{A_2}{a_2},*)\para (\delta,a_2)\in \AtkeyDenotations{\textcolor{Black}{P_2}\CPJud \Delta, x:A_2}\}\\
    &=\AtkeyDenotations{\TranLauP{P_2}\CPJud \DualTranLau{\Delta},x:\DualTranLau{A_2},w_1:\unit}
\end{align*}
We first calculate  the denotations of $\TranLauP{\case{x}{\textcolor{Black}{P_1}}{\textcolor{Black}{P_2}}}$ as follows:
          \begin{align*}
              &\AtkeyDenotations{\TranLauP{\case{x}{\textcolor{Black}{P_1}}{\textcolor{Black}{P_2}}}\CPJud \DualTranLau{\Delta},x:\DualTranLau{A_1}\caseType\DualTranLau{A_2},w_1:\unit}\\
              &=\AtkeyDenotations{\case{x}{\textcolor{Black}{\TranLauP{P_1}}}{\textcolor{Black}{\TranLauP{P_2}}}\CPJud \DualTranLau{\Delta},x:\DualTranLau{A_1}\caseType\DualTranLau{A_2},w_1:\unit}\\
              &=\bigcup_{i\in\{1,2\}}\{(\FDenotationsInd{\Delta}{\delta},(i,\FDenotationsInd{A_i}{a_i}),*)\para (\FDenotationsInd{\Delta}{\delta},\FDenotationsInd{A_i}{a_i},*)\in \AtkeyDenotations{\TranLauP{P_i}\CPJud \DualTranLau{\Delta},x:\DualTranLau{A_i},w_1:\unit} \}\tag{by \Cref{fig:DenotationalSemantics}}\\
              &=\bigcup_{i\in\{1,2\}}\{(\FDenotationsInd{\Delta}{\delta},(i,\FDenotationsInd{A_i}{a_i}),*)\para (\delta,a_i)\in \AtkeyDenotations{P_i\CPJud \Delta,x:A_i} \}\tag{by \Cref{def: transformation on denotations}}
          \end{align*}
In the other direction we have:
        \begin{align*}
            &\FDenotations{\Delta,x:A_1\caseType A_2}{\AtkeyDenotations{(\case{x}{\textcolor{Black}{P_1}}{\textcolor{Black}{P_2})}\CPJud \Delta,x:A_1\caseType A_2}}\\
            &=\{(\FDenotationsInd{\Delta}{\delta},(i,\FDenotationsInd{A_i}{a_i}),*)\para (\delta,(i,a_i))\in \AtkeyDenotations{(\case{x}{\textcolor{Black}{P_1}}{\textcolor{Black}{P_2})}\CPJud \Delta,x:A_1\caseType A_2} \} \tag{by \Cref{def: transformation on denotations}}\\
             &= \bigcup_{i\in\{1,2\}}\{(\FDenotationsInd{\Delta}{\delta},(i,\FDenotationsInd{A_i}{a_i}),*)\para (\delta,a_i)\in\AtkeyDenotations{P_i\CPJud \Delta,x:A_i} \}
        \end{align*}
Thus, we have proved that
\begin{align*}
    & \FDenotations{\Delta,x:A_1\caseType A_2}{\AtkeyDenotations{(\case{x}{\textcolor{Black}{P_1}}{\textcolor{Black}{P_2})}\CPJud \Delta,x:A_1\caseType A_2}}= \\
    & \qquad \AtkeyDenotations{\TranLauP{\case{x}{\textcolor{Black}{P_1}}{\textcolor{Black}{P_2}}}\CPJud \DualTranLau{\Delta},x:\DualTranLau{A_1}\caseType\DualTranLau{A_2},w_1:\unit}
\end{align*}
which concludes this case.

\item Case $\choice{x}{i}P$.
By \Cref{def: Tranformation of processes}, $\TranLauP{\choice{x}{i}P}=\boutt{w}{m}(\choice{w}{i}\impp{w}{y}\TranLauP{P}\para \fwd{m}{n})$.
 By I.H. we have:
\begin{align*}
    \FDenotations{\Delta,x:A_i}{\AtkeyDenotations{P\CPJud\DualTranLau{\Delta},x:A_i}}
    &=\{(\FDenotationsInd{\Delta}{\delta},\FDenotationsInd{A_i}{a},*)\para (\delta,a) \in \AtkeyDenotations{P\CPJud\DualTranLau{\Delta},x:A_i}\}\\
    &=\AtkeyDenotations{\TranLauP{P}\CPJud \DualTranLau{\Delta},y:\DualTranLau{A_i},w:\unit}
\end{align*}

Thus, we calculate the denotations of $\boutt{w}{m}(\choice{w}{i}\impp{w}{y}\TranLauP{P}\para \fwd{m}{n})$. We first calculate the denotations of its subprocesses $\impp{w}{y}\TranLauP{P}$ and $\choice{w}{i}\impp{w}{y}\TranLauP{P}$: 
\begin{align*}
&\AtkeyDenotations{\impp{w}{y}\TranLauP{P}\CPJud \DualTranLau{\Delta},w:\DualTranLau{A_i}\parrType1}\\
&=\{(\FDenotationsInd{\Delta}{\delta},(\FDenotationsInd{A_i}{a},*))\para (\FDenotationsInd{\Delta}{\delta},\FDenotationsInd{A_i}{a},*) \in \AtkeyDenotations{\TranLauP{P}\CPJud \DualTranLau{\Delta},y:\DualTranLau{A_i},w:\unit}\}\tag{by \Cref{fig:DenotationalSemantics}}
\\
&\AtkeyDenotations{\choice{w}{i}\impp{w}{y}\TranLauP{P}\CPJud \DualTranLau{\Delta},w:(\DualTranLau{A_1}\parrType1)\oplus(\DualTranLau{A_2}\parrType1)}
\\
&=\{(\FDenotationsInd{\Delta}{\delta},(i,(\FDenotationsInd{A_i}{a}),*))\para (\FDenotationsInd{\Delta}{\delta},(\FDenotationsInd{A_i}{a_i},*)) \in \AtkeyDenotations{\impp{w}{y}\TranLauP{P}\CPJud \DualTranLau{\Delta},w:\DualTranLau{A_i}\parrType1}\}\tag{by \Cref{fig:DenotationalSemantics}}
\end{align*}
Then we have:
\begin{align*}
&\AtkeyDenotations{\boutt{w}{m}(\choice{w}{i}\impp{w}{y}\TranLauP{P}\para \fwd{m}{n})\CPJud \DualTranLau{\Delta},m:((\DualTranLau{A_1}\parrType1)\oplus(\DualTranLau{A_2}\parrType1))\Tensor \bot, n:\unit}\\
&=\{((\FDenotationsInd{\Delta}{\delta},((i,(\FDenotationsInd{A_i}{a}),*),*)),*)\para\\
&\qquad(\FDenotationsInd{\Delta}{\delta},(i,(\FDenotationsInd{A_i}{a}),*)) \in \AtkeyDenotations{\choice{w}{i}\impp{w}{y}\TranLauP{P}\CPJud \DualTranLau{\Delta},w:(\DualTranLau{A_1}\parrType1)\oplus(\DualTranLau{A_2}\parrType1)},\\
&\qquad(*,*)\in \AtkeyDenotations{\fwd{m}{n}\CPJud m:\bot,n:\unit}\}\tag{by \Cref{fig:DenotationalSemantics}}\\
&=\{(\FDenotationsInd{\Delta}{\delta},\FDenotationsInd{A_1\choiceType A_2}{i,a},*)\para\\
&\qquad(\FDenotationsInd{\Delta}{\delta},(i,(\FDenotationsInd{A_i}{a}),*)) \in \AtkeyDenotations{\choice{w}{i}\impp{w}{y}\TranLauP{P}\CPJud \DualTranLau{\Delta},w:(\DualTranLau{A_1}\parrType1)\oplus(\DualTranLau{A_2}\parrType1)}\}
\end{align*}
Thus, from the previous calculation we obtain:
\begin{align*}
&\AtkeyDenotations{\boutt{w}{m}(\choice{w}{i}\impp{w}{y}\TranLauP{P}\para \fwd{m}{n})\CPJud \DualTranLau{\Delta},m:((\DualTranLau{A_1}\parrType1)\oplus(\DualTranLau{A_2}\parrType1))\Tensor \bot, n:\unit}\\
&=\{(\FDenotationsInd{\Delta}{\delta},\FDenotationsInd{A_1\choiceType A_2}{i,a},*)\para (\FDenotationsInd{\Delta}{\delta},\FDenotationsInd{A_i}{a},*) \in \AtkeyDenotations{\TranLauP{P}\CPJud \DualTranLau{\Delta},y:\DualTranLau{A_i},w:\unit}\}\\
&=\{(\FDenotationsInd{\Delta}{\delta},\FDenotationsInd{A_1\choiceType A_2}{i,a},*)\para (\delta,a) \in \AtkeyDenotations{P\CPJud\DualTranLau{\Delta},x:A_i}\}\tag{by the I.H.}\\
&=\FDenotations{\Delta,A_1\choiceType A_2}{\AtkeyDenotations{\choice{x}{i}P\CPJud \Delta,A_1\choiceType A_2}}
\end{align*}
which concludes this case.

\item Case $P\CPJud \Delta,x:\ClientTypeSingle{A}$ (weakening rule \rulenamestyle{W}).
\begin{prooftree}
    \AxiomC{$P\CPJud \Delta$}
    \UnaryInfC{$P\CPJud \Delta,x:\ClientTypeSingle{A}$}
\end{prooftree}

By \Cref{fig:DenotationalSemantics}, 
$\AtkeyDenotations{P\CPJud \Delta, x:\ClientTypeSingle{A}}=\{(\delta,\emptyset)\para \delta \in \AtkeyDenotations{P\CPJud \Delta}\}$
and by the I.H. we have:
\begin{align*}
    \FDenotations{\Delta}{\AtkeyDenotations{P\CPJud \Delta}}&=\{(\FDenotationsInd{\Delta}{\delta},*) \para \delta \in \AtkeyDenotations{P\CPJud \Delta}\}\\
  &=\AtkeyDenotations{\TranLauP{P}\CPJud\DualTranLau{\Delta},w:\unit}
\end{align*}

Thus, we calculate $\AtkeyDenotations{\TranLauP{P}\CPJud\DualTranLau{\Delta}, x: \ClientType{(\DualTranLau{A}}\parrType\unit)\Tensor\bot,w:\unit}$ as follows:
\begin{align*}
    &\AtkeyDenotations{\TranLauP{P}\CPJud\DualTranLau{\Delta}, x: \ClientType{(\DualTranLau{A}}\parrType\unit)\Tensor\bot,w:\unit}\\
    &=\{(\FDenotationsInd{\Delta}{\delta},\emptyset,*) \para (\FDenotationsInd{\Delta}{\delta},*) \in \AtkeyDenotations{\TranLau{P}\CPJud\DualTranLau{\Delta},w:\unit}\}\tag{by \Cref{fig:DenotationalSemantics}}\\
    &=\{(\FDenotationsInd{\Delta}{\delta},\emptyset,*) \para \delta \in \AtkeyDenotations{P\CPJud \Delta}\}
\end{align*}

In the other direction we have:
\begin{align*}
&\FDenotations{\Delta,\ClientTypeSingle{A}}{\AtkeyDenotations{P\CPJud \Delta,x:\ClientTypeSingle{A}}}\\
    &=\{(\FDenotationsInd{\Delta}{\delta},\FDenotationsInd{\ClientTypeSingle{A}}{a},*) \para (\delta,a)\in \AtkeyDenotations{P\CPJud \Delta,x:\ClientTypeSingle{A}}\}\tag{by \Cref{fig:DenotationalSemantics}}\\
      &=\{(\FDenotationsInd{\Delta}{\delta},\FDenotationsInd{\ClientTypeSingle{A}}{\emptyset},*) \para (\delta,\emptyset)\in \AtkeyDenotations{P\CPJud \Delta,x:\ClientTypeSingle{A}}\}\\
      &=\{(\FDenotationsInd{\Delta}{\delta},{\emptyset},*) \para (\delta,\emptyset)\in \AtkeyDenotations{P\CPJud \Delta,x:\ClientTypeSingle{A}}\}\tag{by \Cref{fig:DenotationalSemantics}}\\
       &=\{(\FDenotationsInd{\Delta}{\delta},{\emptyset},*) \para \delta\in \AtkeyDenotations{P\CPJud \Delta}\}
\end{align*}
and so we have shown that
\begin{align*}
    \FDenotations{\Delta,\ClientTypeSingle{A}}{\AtkeyDenotations{P\CPJud \Delta,x:\ClientTypeSingle{A}}}=\AtkeyDenotations{\TranLauP{P}\CPJud\DualTranLau{\Delta}, x: \ClientType{(\DualTranLau{A}}\parrType\unit)\Tensor\bot,w:\unit}
\end{align*}

\item Case $P\Substitution{x_1}{x_2}\CPJud x_1:\ClientTypeSingle{A}$. (contraction rule \rulenamestyle{C}).
\begin{prooftree}
    \AxiomC{$P\CPJud\Delta,x_1:\ClientTypeSingle{A},x_2:\ClientTypeSingle{A}$}
    \UnaryInfC{$P\Substitution{x_1}{x_2}\CPJud x_1:\ClientTypeSingle{A}$}
\end{prooftree}
By \Cref{def: Tranformation of processes}, $\TranLauP{P\Substitution{x_1}{x_2}}=\TranLauP{P}\Substitution{x'_1}{x'_2}$.
By \Cref{fig:DenotationalSemantics}:
\begin{mathpar}
   \AtkeyDenotations{P\Substitution{x_1}{x_2}\CPJud \Delta,x_1:\ClientTypeSingle{A}}
      =\left\{
      \begin{aligned}
        &(\delta,\alpha_1\uplus\alpha_2)\para (\delta,\alpha_1,\alpha_2)\\[-1.5em]
          & \in \AtkeyDenotations{P\CPJud \Delta,x_1:\ClientTypeSingle{A},x_2:\ClientTypeSingle{A}}
      \end{aligned}
\right\}
\end{mathpar}
and by the I.H.:
\begin{align*}
 &\FDenotations{\Delta,x_1:\ClientTypeSingle{A},x_2:\ClientTypeSingle{A}}{\AtkeyDenotations{P\CPJud\Delta,x_1:\ClientTypeSingle{A},x_2:\ClientTypeSingle{A}}}\\
    &=\{(\FDenotationsInd{\Delta}{\delta},\FDenotationsInd{\ClientTypeSingle{A}}{\alpha_1},\FDenotationsInd{\ClientTypeSingle{A}}{\alpha_2},*) \para (\delta,\alpha_1,\alpha_2) \in \AtkeyDenotations{P\CPJud\Delta,x_1:\ClientTypeSingle{A},x_2:\ClientTypeSingle{A}}\}\\
    &= \AtkeyDenotations{\TranLauP{P}\CPJud\DualTranLau{\Delta},x_1':\DualTranLau{(\ClientTypeSingle{A})},x_2':\DualTranLau{(\ClientTypeSingle{A})},w:\unit}
\end{align*}
We calculate the denotations of $\TranLauP{P}\Substitution{x'_1}{x'_2}$ as follows:
\begin{align*}
&\AtkeyDenotations{\TranLauP{P}\Substitution{x_1'}{x_2'}\CPJud \DualTranLau{\Delta},x_1':\DualTranLau{(\ClientTypeSingle{A})},w:\unit}\\
    &=\{(\FDenotationsInd{\Delta}{\delta},\FDenotationsInd{\ClientTypeSingle{A}}{\alpha_1}\uplus\FDenotationsInd{\ClientTypeSingle{A}}{\alpha_2},*) \para \\
    &\qquad(\FDenotationsInd{\Delta}{\delta},\FDenotationsInd{\ClientTypeSingle{A}}{\alpha_1},\FDenotationsInd{\ClientTypeSingle{A}}{\alpha_2},*) \in \AtkeyDenotations{\TranLauP{P}\CPJud\DualTranLau{\Delta},x_1':\DualTranLau{(\ClientTypeSingle{A})},x_2':\DualTranLau{(\ClientTypeSingle{A})},w:\unit}
    \}\tag{by \Cref{fig:DenotationalSemantics}}\\
      &=\{(\FDenotationsInd{\Delta}{\delta},\FDenotationsInd{\ClientTypeSingle{A}}{\alpha_1}\uplus\FDenotationsInd{\ClientTypeSingle{A}}{\alpha_2},*) \para 
      (\delta,\alpha_1,\alpha_2) \in \AtkeyDenotations{P\CPJud\Delta,x_1:\ClientTypeSingle{A},x_2:\ClientTypeSingle{A}}
    \}\\
       &=\{(\FDenotationsInd{\Delta}{\delta},\FDenotationsInd{\ClientTypeSingle{A}}{\alpha_1\uplus \alpha_2},*) \para 
      (\delta,\alpha_1,\alpha_2) \in \AtkeyDenotations{P\CPJud\Delta,x_1:\ClientTypeSingle{A},x_2:\ClientTypeSingle{A}}
    \}\tag{by \Cref{lemma: union of multisets}}\\
           &=\{(\FDenotationsInd{\Delta}{\delta},\FDenotationsInd{\ClientTypeSingle{A}}{\alpha_1\uplus \alpha_2},*) \para 
      (\delta,\alpha_1\uplus\alpha_2) \in \AtkeyDenotations{P\Substitution{x_1}{x_2}\CPJud\Delta,x_1:\ClientTypeSingle{A}}
    \}\tag{by \Cref{def: transformation on denotations}}\\
    &=\FDenotations{\Delta,x_1:\ClientTypeSingle{A}}{\AtkeyDenotations{P\Substitution{x_1}{x_2}\CPJud\Delta,x_1:\ClientTypeSingle{A}}}
\end{align*}
        \end{itemize}
        This concludes this case (and the proof).
        \end{proof}
By combining 
\Cref{theorem: laurent transformation on denotations,theorem: adequacy} and \Cref{lemma: injectivity Function L formula}, we obtain our first full abstraction result:
\begin{corollaryapx}[Full Abstraction (I)]
\label{cor:fullabsden}
    Suppose $P,Q\CPJud \Delta$.  Then $P \ObsEquivAtk Q$ iff $\TranLauP{P}\ObsEquivAtk \TranLauP{Q}$.
\end{corollaryapx}
\begin{proof} We have the following equivalences:
   \begin{align*}
     P \ObsEquivAtk Q \quad 
     & 
      \Leftrightarrow\qquad  \AtkeyDenotations{P\CPJud \Delta} =\AtkeyDenotations{Q\CPJud\Delta}\tag{by \Cref{theorem: adequacy}}
      \\
      & \Leftrightarrow\qquad \FDenotations{\Delta}{\AtkeyDenotations{P\CPJud \Delta}}=\FDenotations{\Delta}{\AtkeyDenotations{Q\CPJud \Delta}}\tag{by \Cref{lemma: injectivity Function L formula}}
       \\
& \Leftrightarrow\qquad \AtkeyDenotations{\TranLauP{P}\CPJud\DualTranLau{\Delta},w:\unit}=\AtkeyDenotations{\TranLauP{Q}\CPJud\DualTranLau{\Delta},w:\unit}\tag{by \Cref{theorem: laurent transformation on denotations}}
\\
     &  \Leftrightarrow\qquad  \TranLauP{P}\ObsEquivAtk \TranLauP{Q}\tag{by \Cref{theorem: adequacy}}
   \end{align*}
\end{proof}

\section{An Operational Characterization of Laurent's Translation}
\label{sec:Transformers}

\begin{figure}[t]
    \begin{mathpar}
    \inferrule*[Right=ObsMix]{C\Observation \and C'\Downarrow \gamma}
    {C\para C'\Downarrow (\sigma,\gamma)}
    \and
    \inferrule*[Right=cfgMix]{C_1 \ConfJud{\Gamma_1}{\Sigma_1}\and C_2\ConfJud{\Gamma_2}{\Sigma_2}}
    {C_1\para C_2 \ConfJud{\Gamma_1,\Gamma_2}{\Sigma_1, \Sigma_2}}
        \and
           \AtkeyDenotations{{P}\para{Q}\CPJud\Gamma,\Delta}=
       \left\{(\gamma,\delta)\para \gamma\in \AtkeyDenotations{P\CPJud\Gamma}, \delta\in \AtkeyDenotations{Q\CPJud \Delta}\right\}
         \vspace{-3mm}
  \end{mathpar}
  \caption{Extensions concerning \MixTwo.}
  \label{fig:mixtwo_extensions}
\end{figure}  
In this section, we show that the translation $\TranLauP{-}$ can be \emph{internalized} as an evaluation context.
That is, given $P \CPJud \Delta$, we can define a corresponding \emph{transformer context}, denoted $\TranProNames{-}{\Delta}$.
Using this context, we obtain a process denoted $\TranProNames{P}{\Delta}$, in which the behavior of $P$ is adapted following $\Delta$, so as to produce $\DualTranLau{\Delta}, w:\unit$.
This is clearly different from translating $P$ into $\TranLauP{P}$ by examining its structure.
We shall show that $\TranProNames{P}{\Delta}$ is equivalent to $\TranLauP{P}$ (\Cref{{coro:transformer_correct}}).
As in \secRef{sec:Laurent's denotations}, we will also show a full-abstraction result for $\TranProNames{-}{\Delta}$ (\Cref{corollary: full abstraction transformers}).

This strategy works in presence of Rule \MixTwo (cf.\,\Cref{fig:CP Rules}). 
Hence, in this section we work with typed processes in \CPMixZeroTwo.
Accordingly, we extend the denotational semantics (cf.~\Cref{s:back}) as given in \Cref{fig:mixtwo_extensions}.
It is easy to check that soundness and completeness  (\Cref{theorem: adequacy,ObservationalEquivalenceCorollary}) still hold for \CPMixZeroTwo.
In $\CPMixZeroTwo$, additional observational equivalences arise from permutation of Rule \MixTwo with other rules:
\begin{lemmaapxrep}\label{lemma: observational equivalences rule mix}
Given $P,Q,R\in \CPMixZeroTwo$, 
we have:  
${\case{x}{P}{Q}\para R}  \ObsEquivAtk\case{x}{P\para R}{Q\para R}$,
$\cut{x}{P}{Q}\para R  \ObsEquivAtk \cut{x}{P}{(Q\para R)}\ObsEquivAtk \cut{x}{(P\para R)}{Q}$   and  
$\boutt{y}{x}(P\para (Q\para R))\ObsEquivAtk\boutt{y}{x}(P\para Q)\para R$.
 \end{lemmaapxrep}
\begin{proof}
The proof follows by straightforwardly applying the rules from \Cref{fig:DenotationalSemantics}.
    \begin{itemize}
        \item 
         \begin{align*}
  &\AtkeyDenotations{\cut{x}{P}{Q}\para R\CPJud\Gamma,\Delta,\Sigma}\\
    &\quad=\{(\delta, \gamma,\sigma) \para (\delta,\gamma) \in \AtkeyDenotations{\cut{x}{P}{Q}\CPJud\Gamma,\Delta}, \sigma \in \AtkeyDenotations{R\CPJud\Sigma}   \}\\
     &\quad=\{(\delta, \gamma,\sigma) \para  (\gamma,a)\in \AtkeyDenotations{P\CPJud\Gamma,x:A}, (\delta,a)\in \AtkeyDenotations{Q\CPJud \Delta,x:A^\bot}, \sigma \in \AtkeyDenotations{R\CPJud\Sigma}   \}\\
     &\quad=\{(\delta, \gamma,\sigma) \para  (\gamma,a)\in \AtkeyDenotations{P\CPJud\Gamma,x:A}, (\delta,\sigma,a)\in \AtkeyDenotations{Q\para R\CPJud \Delta,\Sigma,x:A^\bot}   \}\\
    &\quad=\AtkeyDenotations{\cut{x}{P}{(Q\para R)}\CPJud \Gamma,\Delta,\Sigma}
        \end{align*}

    \item
    \begin{align*}
        &\AtkeyDenotations{\case{x}{P}{Q}\para R \CPJud \Gamma,\Sigma,x:A_1\caseType A_2}\\
        &\quad=\bigcup_{i\in\{1,2\}}\{(\gamma,\sigma,(i,a)\para (\gamma,(i,a))\in \AtkeyDenotations{\case{x}{P}{Q}\CPJud \Gamma, x: A_1\caseType A_2}, \sigma \in \AtkeyDenotations{R\CPJud\Sigma} \}\\
        &\quad=\bigcup_{i\in\{1,2\}}\{(\gamma,\sigma,(i,a)\para (\gamma,a)\in \AtkeyDenotations{{P}\CPJud \Gamma, x: A_1}, (\gamma,a)\in \AtkeyDenotations{{Q}\CPJud \Gamma, x: A_2},\sigma \in \AtkeyDenotations{R\CPJud\Sigma} \}\\
          &\quad=\bigcup_{i\in\{1,2\}}\{(\gamma,\sigma,(i,a)\para (\gamma,\sigma,a)\in \AtkeyDenotations{P\para R\CPJud \Gamma,\Sigma,x:A_1}, (\gamma,\sigma,a)\in \AtkeyDenotations{Q\para R\CPJud \Gamma,\Sigma,x:A_2}\}\\
        &\quad=\AtkeyDenotations{\case{x}{P\para R}{Q\para R}\CPJud \Gamma,\Sigma,x:A_1\caseType A_2}
    \end{align*}

    \item 

    \begin{align*}
        &\AtkeyDenotations{\boutt{y}{x}(P\para (Q\para R))\CPJud \Delta,\Gamma,\Sigma, x:A\Tensor B}\\
        &\quad=\{(\delta,\gamma,\sigma,(a,b))\para (\delta,a)\in \AtkeyDenotations{P\CPJud\Delta,y:A}, (\gamma,\sigma,b) \in \AtkeyDenotations{Q\para R\CPJud\Gamma,\Sigma,x:B}\}\\
        &\quad=\{(\delta,\gamma,\sigma,(a,b))\para (\delta,a)\in \AtkeyDenotations{P\CPJud\Delta,y:A}, (\gamma,b) \in \AtkeyDenotations{Q\CPJud\Gamma,x:B}, \sigma \in \AtkeyDenotations{R\CPJud \Sigma}\}\\
         &\quad=\{(\delta,\gamma,\sigma,(a,b))\para (\delta,\gamma,(a,b))\in \AtkeyDenotations{\boutt{y}{x}(P\para Q)\CPJud\Delta,\Gamma,x:A\Tensor B}\sigma \in \AtkeyDenotations{R\CPJud \Sigma}\}\\
         &\quad=\AtkeyDenotations{\boutt{y}{x}(P\para Q)\para R\CPJud \Delta,\Gamma,\Sigma, x:A\Tensor B}
    \end{align*}
        
    \end{itemize}
\end{proof}

We also need to extend $\TranLauP{-}$ (\Cref{def: Tranformation of processes}).
Given processes $P\CPJud\Delta$ and $Q\CPJud\Gamma$ (with their translations $\TranLauP{P}\CPJud\DualTranLau{\Delta},y:\unit$ and
$\TranLauP{Q}\CPJud\DualTranLau{\Gamma},x:\unit$, respectively), we define:
\begin{align*}
\TranLauP{P\para Q}=\cut{x}{\boutt{y}{x}(\TranLauP{P}\para \TranLauP{Q})}{M_x}    
\end{align*}
where $M_x=\impp{x}{y}\Emptyimpp{y}\forward{x}{m}$. It is easy to check that $M_x \CPJud x: \bot \parrType\bot, m:\unit$.
\begin{remark}
  The results about $\FDenotations{}{-}$ in \secRef{sec:Laurent's denotations} can be adapted to  $\CPMixZeroTwo$.  
\end{remark}

\subparagraph*{Transformers.}
We define \emph{transformer processes}, which
adapt the behavior of one session on a given name.
\begin{definition}[Transformers]\label{DefinitionTransformers}
 Given a type $A$ in $\CPMixZeroTwo$,
 we define 
  the \emph{transformer process} $\Transformer{x,x'}{A}\CPJud x:\Dual{A}, x':\DualTranLau{A}$ by induction on the type $A$ as in \Cref{f:trans}.
  With a slight abuse of notation, in the figure we write 
  $\Transformer{x,x'}{A} = P \CPJud \Delta$ to express that 
  $\Transformer{x,x'}{A} = P$ with $P \CPJud \Delta$.
\end{definition}
\begin{figure}[t]
\begin{align*}
\Transformer{x,x'}{\bot} & = \forward{x}{y}\CPJud x:\unit, x':\bot
\\
\Transformer{x,x'}{\unit} & = \boutt{y}{x'}(\forward{y}{x}\para \Emptyimpp{x'}\inact)\CPJud x:\bot, x':\unit\Tensor\bot 
\\
\Transformer{x,z}{A\Tensor B} & = \impp{x}{y}\boutt{z_2}{z}(\boutt{z_2}{z_1}(\impp{z_1}{y}(\Transformer{y,y'}{A}\para \emptyOut{z_{1}})\para \impp{z_2}{x'}(\Transformer{x,x'}{B} \para \emptyOut{z_{2}})\para \Emptyimpp{z}\inact \CPJud \Delta 
\\
& \qquad \text{(with $\Delta = x:A^\bot\parrType B^\bot,z:\DualTranLau{(A\Tensor B)}$)}
\\
    \Transformer{x,x'}{A\parrType B} & = \impp{x'}{y'}\boutt{y}{x}(\Transformer{y,y'}{A}\para \Transformer{x,x'}{B})\CPJud x:A^\bot\Tensor B^\bot, x':\DualTranLau{A\parrType B}
\\
    \Transformer{x,x'}{\ServerTypeSingle{A}} & = \boutt{w'}{x'}(\Server{w'}{w}\Client{x}{y}\impp{w}{y'}(\Transformer{y,y'}{A}\para  \emptyOut{w})\para\Emptyimpp{x'}\inact)\CPJud  x:\ClientType{A^\bot},x':\DualTranLau{(\ServerTypeSingle{A})}
\\
    \Transformer{x,x'}{\ClientTypeSingle{A}} & = \Server{x}{y}\Client{x'}{m}\boutt{z}{m}(\impp{z}{y'}(\Transformer{y,y'}{A}\para  \emptyOut{z}) \para \Emptyimpp{m}\inact)\CPJud x:\ServerType{A^\bot},x':\DualTranLau{(\ClientTypeSingle{A})}
\\
    \Transformer{x,x'}{A_1\caseType A_2} & = \case{x'}{\choice{x}{1}\Transformer{x,x'}{A_1}}{\choice{x}{2}\Transformer{x,x'}{A_2}}\CPJud x:A_1^\bot\choiceType A_2^\bot,x':\DualTranLau{(A_1\caseType A_2)}
    \\
        \Transformer{x,x'}{A_1\choiceType A_2} & = \case{y}{P_1}{P_2}\CPJud x:A_1^\bot\caseType A_2^\bot,x': \DualTranLau{(A_1\choiceType A_2)}
\\ & \qquad \text{(with $P_i=\boutt{w}{m}\big(\choice{w}{i}\impp{w}{y'}(\Transformer{x,x'}{A_i}\para\emptyOut{w})\para \Emptyimpp{m}\inact\big)$)}
\end{align*}
\vspace{-7mm} 
\caption{Transformer processes (\Cref{DefinitionTransformers}). \label{f:trans} }
\end{figure}

We now define \emph{transformer contexts}, which adapt an entire context $\Delta$ using transformer processes.
We first define \emph{typed  contexts}.

\begin{figure}[t]
\vspace{-2mm}
\begin{mathpar}
    \inferrule*[Right=KCut$_1$]{\TypedContext\CPJudContext{\Sigma}{ \Gamma,x:A} \and Q\CPJud{\Delta,x:A^\bot}}{\cut{x}{\TypedContext}{Q}\CPJudContext{\Sigma} \Gamma,\Delta}\\
    \and
    \inferrule*[Right=KMix]{\TypedContext\CPJudContext{\Sigma}{\Gamma}\and P\CPJud \Delta}{\TypedContext\para P\CPJudContext{\Sigma}{\Gamma,\Delta}}
    \and
    \inferrule*[Right=KHole]{ }{\Hole\CPJudContext{\Delta}{\Delta}}
    \end{mathpar}
\vspace{-7mm}
    \caption{Typed Contexts}
    \label{fig: Typed Contexts Rules}
\end{figure}

\subparagraph*{Typed Process Contexts.}
A typed context is a typed process with a typed hole.
We write $\TypedContext\CPJudContext{\Delta}{\Gamma}$ for a typed context which contains a typed hole $\Hole$, and which produces a process of type $\Gamma$.
That is, given a process $P \CPJud \Delta$, we can fill $\TypedContext$ as $\TypedContextKProcess{P}\CPJud \Gamma$, by replacing the unique occurrence of $\Hole$ with $P$.
\Cref{fig: Typed Contexts Rules} gives the rules for forming typed contexts.
As an example, consider the derivation for a parallel context $\cut{x}{\Hole}{\Transformer{x,y}{A}}$:
 \begin{prooftree}
    \AxiomC{}
     \UnaryInfC{$\Hole\CPJudContext{x:A}{x:A}$}
     \AxiomC{$\Transformer{x,y}{A}\CPJud{ x:A^\bot, y:\DualTranLau{A}}$}
  \BinaryInfC{$\cut{x}{\Hole}{\Transformer{x,y}{A}}\CPJudContext{x:A}{y:\DualTranLau{A}}$}
\end{prooftree}  
Above, we can replace the use of Rule KHole with a typing derivation for $P\CPJud x:A$,
thus obtaining $\cut{x}{P}{\Transformer{x,y}{A}}\CPJud{y:\DualTranLau{A}}$.
Such ``filling'' of contexts can be done in general:
\begin{lemmaapxrep}\label{lemmaTypedcontext-processes}
Given $\TypedContext\CPJudContext{\Delta}{\Gamma}$ and $P\CPJud\Delta$, we have that $\TypedContextKProcess{P}\CPJud \Gamma$ is derivable.
\end{lemmaapxrep}
 \begin{proof}
    By induction on the derivation of $\TypedContext\CPJudContext{\Delta}{\Gamma}$.
    \begin{itemize}
        \item If $\TypedContext= \Hole$, then $\TypedContextKProcess{P}=P$, thus the proof of 
        $\TypedContextKProcess{P}\CPJud\Gamma$ is the proof of $P\CPJud\Delta$.

        \item If $\TypedContextProcess{K'}{\Hole}
        =\cut{x}{\TypedContext}{Q}$

        \begin{prooftree}
\AxiomC{$\TypedContext\CPJudContext{\Delta}{\Gamma_1,x:A}$}
            \AxiomC{$Q\CPJud \Gamma_2,x:A^\bot$}
            \BinaryInfC{$\cut{x}{\TypedContext}{Q}\CPJudContext{\Delta}{\Gamma_1,\Gamma_2}$}
        \end{prooftree}

By I.H. $\TypedContextKProcess{P}\CPJud \Gamma_1,x:A$
is derivable, thus we finish by applying the rule Cut.

\item  If $\TypedContextProcess{K'}{\Hole}=\TypedContextKProcess{\Hole}\para P$.

\begin{prooftree}
    \AxiomC{$\TypedContext \CPJudContext{\Delta}{\Gamma_1}$}
    \AxiomC{$Q\CPJud \Gamma_2$}
    \BinaryInfC{$\TypedContext\para Q \CPJudContext{\Delta}{\Gamma_1,\Gamma_2}$}
\end{prooftree}

By I.H. $\TypedContextKProcess{P}\CPJud \Gamma_1$, thus we finish by applying the rule $Mix$.
    \end{itemize}
\end{proof}
\subparagraph*{Transformer contexts.}  As we have seen, contexts in our setting are hardly arbitrary: only type-compatible processes are inserted into holes. Based on this observation, and following the typing rules, we define 
\textit{transformer contexts} and \textit{transformer contexts with closing name}:

\begin{definition}
    Let $\Delta=x_1:A_1,\dots,x_n:A_n$ be a typing context.  We define:
    \begin{itemize}
    \item \textit{transformer contexts}:  
    $\TranProNameswone{\Hole}{\Delta}=\cut{x_n}{\cdots\cut{x_1}{\Hole}{\Transformer{x_1,y_1}{A_1}}}{\cdots)\para \Transformer{x_n,y_n}{A_n}}$
        \item  \textit{transformer contexts with closing name} ($z$ is fresh wrt $\Delta$): 
        $\TranProNames{\Hole}{\Delta}= \TranProNameswone{\Hole}{\Delta} \para \emptyOut{z}.$
    \end{itemize}
\end{definition}

Note that by \Cref{theorem: observational equivalences} the order of the cuts in $\TranProNameswone{-}{\Delta}$ does not matter.
We will show that  transformers are correct: the transformed process $\TranProNames{P}{\Delta}$ is equivalent to the translated process $\TranLauP{P}$ (\Cref{lemmaTransformers}).
We need auxiliary results about transformers.
\begin{lemmaapxrep}\label{lemma: observational equivalence of transformers}
    The following observational equivalences hold:
    \begin{itemize}
    	\item  $\impp{x'}{y'}\TranProNames{P}{\Delta,y:A,x:B}\ObsEquivAtk \TranProNames{\impp{x}{y}P}{\Delta,x:A\parrType B}$ 
    \item 
    $\boutt{w'}{x'}(\Server{w'}{w} \impp{w}{y'}\TranProNames{P}{\ClientType{\Delta},y:A}\para \Emptyimpp{x'}\inact)\para \Emptyoutt{n} \ObsEquivAtk \TranProNames{\Server{x}{y}P}{\Delta,x:\ServerTypeSingle{A}}$
    \item  
    $\Client{x}{m}\boutt{z}{m}( \impp{z}{y}\TranProNames{P}{\Delta,y:A}\para \Emptyimpp{m}\inact)\ObsEquivAtk \TranProNames{\Client{x}{y}P}{\Delta,x:\ClientTypeSingle{A}}$
    \item 
    $\boutt{w}{m}(\choice{w}{i}\impp{w}{y}\TranProNames{P}{\Delta,y:A}\para \Emptyimpp{m}\inact)\para \Emptyoutt{n}\ObsEquivAtk \TranProNames{\choice{y}{i}P}{y:A_1\choiceType A_2}$
    \item 
      $\case{y'}{\TranProNames{\textcolor{Black}{P_1}}{\Delta,y:A_1}}{\TranProNames{\textcolor{Black}{P_2}}{\Delta,y:A_2}}\ObsEquivAtk \TranProNames{\case{y}{\textcolor{Black}{P_1}}{\textcolor{Black}{P_2}}}{\Delta y:A_1\caseType A_2}$
      \item $
            \boutt{z_2}{z}( \boutt{z_1}{z_2}(\impp{z_1}{y'} \TranProNames{\textcolor{Black}{P_1}}{\Delta,y:A} \para\impp{z_2}{x'} \TranProNames{\textcolor{Black}{P_2}}{\Gamma,x:B} )\para \Emptyimpp{z}\inact)\para \Emptyoutt{w} \ObsEquivAtk \TranProNames{\boutt{y}{x}(\textcolor{Black}{P_1}\para\textcolor{Black}{P_2})}{\Delta,\Gamma,x:A\Tensor B}
        $
    \end{itemize}

\end{lemmaapxrep}
\begin{proofsketch}
The proof follows from \Cref{theorem: observational equivalences} and \Cref{lemma: observational equivalences rule mix}. 
\end{proofsketch}

\begin{proof}
We consider each item separately:
\begin{itemize} 
\item 
 \begin{align*}
     &\TranProNames{\impp{x}{y}P}{\Delta,x:A\parrType B}\\
     &=\cut{x}{\TranProNameswone{\impp{x}{y}P}{\Delta,x:A\parrType B}}{\Transformer{x,x'}{A\parrType B}}\para\Emptyoutt{w}\\
     &=\cut{x}{\TranProNameswone{\impp{x}{y}P}{\Delta,x:A\parrType B}}{\impp{x'}{y'}\boutt{y}{x}(\Transformer{y,y'}{A}\para \Transformer{x,x'}{B})}\para\Emptyoutt{w}\tag{by \Cref{DefinitionTransformers}}\\
     &\ObsEquivAtk^*\impp{x'}{y'}\cut{x}{\cut{y}{\TranProNameswone{P}{\Delta,y:A}}{\Transformer{y,y'}{A}}}{\Transformer{x,x'}{B}}\para\Emptyoutt{w} \tag{by \Cref{theorem: observational equivalences}}\\
     &=\impp{x'}{y'}\TranProNames{P}{\Delta,y:A,x:B}
     \end{align*}
 \item 
\begin{align*}
    &\TranProNames{\boutt{y}{x}(\textcolor{Black}{P_1}\para \textcolor{Black}{P_2})}{\Delta,\Gamma,x:A\Tensor B}\\
    &=\cut{x}{\TranProNameswone{\boutt{y}{x}(\textcolor{Black}{P_1}\para \textcolor{Black}{P_2})}{\Gamma,\Delta}}{\Transformer{x,z}{A\Tensor B}}\para \Emptyoutt{w}\\
    &= \cut{x}{\TranProNameswone{\boutt{y}{x}(\textcolor{Black}{P_1}\para \textcolor{Black}{P_2})}{\Gamma,\Delta}}{\impp{x}{y}\boutt{z_2}{z}(\boutt{z_2}{z_1}(\impp{z_1}{y}\TransformerHat{y,y'}{A}\para \impp{z_2}{x'}\TransformerHat{x,x'}{B})\para \Emptyimpp{z}}\para \Emptyoutt{w}\tag{by \Cref{DefinitionTransformers}}\\
&\ObsEquivAtk^*\boutt{z_2}{z}( \boutt{z_1}{z_2}(\impp{z_1}{y'}\TranProNames{\textcolor{Black}{P_1}}{\Delta,y:A} \para\impp{z_2}{x'} \TranProNames{\textcolor{Black}{P_2}}{\Gamma,x:B} )\para \Emptyimpp{z})\para \Emptyoutt{w} \tag{by \Cref{theorem: observational equivalences}}
\end{align*}

\item 
\begin{align*}
    &\TranProNames{\Server{x}{y}P}{\ClientTypeSingle{\Delta},x:\ServerTypeSingle{A}}\\
    &=\cut{x}{\TranProNameswone{\Server{x}{y}P}{\Delta}}{\Transformer{x,x'}{\ServerTypeSingle{A}}}\para\Emptyoutt{n}\\
    &=\cut{x}{\TranProNameswone{\Server{x}{y}P}{\ClientTypeSingle{\Delta}}}{\boutt{w'}{x'}(\Server{w'}{w}\Client{x}{y}\impp{w}{y'}\TransformerHat{y,y'}{A}\para\Emptyimpp{x'})}\para \Emptyoutt{n}\tag{by \Cref{DefinitionTransformers}}\\
    &\ObsEquivAtk^*\boutt{w'}{x'}(\Server{w'}{w} \impp{w}{y'}\cut{y}{\TranProNameswone{P}{\ClientTypeSingle{\Delta}}}{\TransformerHat{y,y'}{A}}\para \Emptyimpp{x'})\para \Emptyoutt{n} \tag{by \Cref{theorem: observational equivalences}}\\
    &=\boutt{w'}{x'}(\Server{w'}{w} \impp{w}{y'}\TranProNames{P}{\ClientTypeSingle{\Delta},y:A}\para \Emptyimpp{x'})\para \Emptyoutt{n} 
\end{align*}

\item 
\begin{align*}
    &\TranProNames{\Client{x}{y}P}{\Delta,x:\ClientTypeSingle{A}}\\
    &=\cut{x}{\TranProNameswone{\Client{x}{y}P}{\Delta}}{\Transformer{x,x'}{\ClientTypeSingle{A}}}\para \Emptyoutt{n}\\
    &=\cut{x}{\TranProNameswone{\Client{x}{y}P}{\Delta}}{\Server{x}{y}\Client{x'}{m}\boutt{z}{m}(\impp{z}{y'}\TransformerHat{y,y'}{A} \para \Emptyimpp{m})}\para \Emptyoutt{n}\tag{by \Cref{DefinitionTransformers}}\\
    &\ObsEquivAtk^* \Client{x'}{m}\boutt{z}{m}(\impp{z}{y'}\cut{y}{\TranProNameswone{P}{\Delta}}{\Transformer{y,y'}{A}}\para \Emptyimpp{m})\para \Emptyoutt{n} \tag{by \Cref{theorem: observational equivalences}}\\
    &=\Client{x'}{m}\boutt{z}{m}(\impp{z}{y'}\TranProNames{P}{\Delta,y:A}\para \Emptyimpp{m})\para \Emptyoutt{n} \\
\end{align*}

\item  Let
$P_i=\boutt{w}{m}(\choice{w}{i}\impp{w}{y'}\TransformerHat{y,y'}{A_i}\para \Emptyimpp{m}\inact)$, thus we have:
\begin{align*}
    &\TranProNames{\choice{y}{i}P}{\Delta,y:A_1\choiceType A_2}\\
    &=\cut{y}{\TranProNameswone{\choice{y}{i}P}{\Delta}}{\Transformer{y,m}{A_1\choiceType A_2}}\para\Emptyoutt{n}\\
    &=\cut{y}{\TranProNameswone{\choice{y}{i}P}{\Delta}}{\case{y}{P_1}{P_2}}\para\Emptyoutt{n}\tag{by \Cref{DefinitionTransformers}}\\
    &\ObsEquivAtk \cut{y}{\TranProNameswone{P}{\Delta}}{\boutt{w}{m}(\choice{w}{i}\impp{w}{y'}\TransformerHat{y,y'}{A_i}\para \Emptyimpp{m})}\para\Emptyoutt{n}\tag{by \Cref{theorem: observational equivalences}}\\
    &\ObsEquivAtk^* \boutt{w}{m}(\choice{w}{i}\impp{w}{y'}\cut{y}{\TranProNameswone{P}{\Delta}}{\TransformerHat{y,y'}{A_i}}\para \Emptyimpp{m})\para\Emptyoutt{n}\tag{by \Cref{theorem: observational equivalences}}\\
    &=\boutt{w}{m}(\choice{w}{i}\impp{w}{y'}\TranProNames{P}{\Delta,y:A_i}\para \Emptyimpp{m})\para\Emptyoutt{n}
\end{align*}

\item
The following observational equivalence follows by the application of the observational equivalences arising from the permutation of cuts and cut elimination (\Cref{theorem: observational equivalences}) and permutation of \MixTwo (\Cref{lemma: observational equivalences rule mix}).
\begin{align*}
    &\TranProNames{\case{y}{P_1}{P_2}}{\Delta,y:A_1\caseType A_2}\\
    &=\cut{y}{\TranProNameswone{\case{y}{P_1}{P_2}}{\Delta}}{\Transformer{y,y'}{A_1\caseType A_2}}\para \Emptyoutt{n}\\
    &=\cut{y}{\TranProNameswone{\case{y}{P_1}{P_2}}{\Delta}}{ \case{y'}{\choice{y}{1}\Transformer{y,y'}{A_1}}{\choice{y}{2}\Transformer{y,y'}{A_2}}}\para \Emptyoutt{n}\tag{by \Cref{DefinitionTransformers}}\\
   &\ObsEquivAtk \case{y'}
   {
    \cut{y}
    {\TranProNameswone{\case{y}{P_1}{P_2}}{\Delta}}
    {\choice{y}{1}\Transformer{y,y'}{A_1}}
    }{
    \cut{y}{\TranProNameswone{\case{y}{P_1}{P_2}}{\Delta}}
    {\choice{y}{2}\Transformer{y,y'}{A_2}}
    }\para \Emptyoutt{n}\\
    &\ObsEquivAtk^* \case{y'}{
    \cut{y}{\TranProNameswone{P_1}{\Delta}}
    {\Transformer{y,y'}{A_1}}}
    {\cut{y}{\TranProNameswone{P_2}{\Delta}}
    {\Transformer{y,y'}{A_2}}}\para \Emptyoutt{n}\\
    &= \case{y'}
   {\TranProNameswone{P_1}{\Delta,y:A_1}}
    {{\TranProNameswone{P_2}{\Delta,y:A_2}}
    }
    \para \Emptyoutt{n}\\
    &\ObsEquivAtk \case{y'}
   {\TranProNameswone{P_1}{\Delta,y:A_1}\para \Emptyoutt{n}}
    {{\TranProNameswone{P_2}{\Delta,y:A_2}\para \Emptyoutt{n}}
    }\\
    &= \case{y'}
   {\TranProNames{P_1}{\Delta,y:A_1}}
    {{\TranProNames{P_2}{\Delta,y:A_2}}
    }\\
\end{align*}
\end{itemize}
\end{proof}
\begin{lemmaapx}\label{lemma: observational equivalence forwarders}
    $\AtkeyDenotations{\forward{x}{y}\CPJud x:\unit,y:\bot}=\AtkeyDenotations{\Emptyoutt{x}\para \Emptyimpp{y}\inact\CPJud x:\unit,y:\bot}$
\end{lemmaapx}

\begin{lemmaapxrep}\label{lemma: observation transformer client}
Let $\Delta= x_1: \ServerTypeSingle{\Dual{A}}, x_1':\ClientTypeSingle{((A\parrType\unit)\Tensor\bot)}$. We have:
\begin{mathpar}
        \AtkeyDenotations{\Transformer{x_1,x_1'}{\ClientTypeSingle{A}}\CPJud \Delta}
    =\left\{
      \begin{aligned}
        &(\bag{a_1,\dots,a_k},\uplus_{j=1}^k\bag{((a',*),*)})\\[-1.5em]
         \para & \forall i\in\{1,\dots,k\}.(a_i,a_i')\in \AtkeyDenotations{\Transformer{y,y'}{A}\CPJud y:\Dual{A}, y':\DualTranLau{A}}
      \end{aligned}
\right\}
\end{mathpar}
\end{lemmaapxrep}

\begin{proof}
The proof follows straightforwardly  by applying the  rules from \Cref{fig:DenotationalSemantics} to calculate the denotations of $\Transformer{x_1,x_1'}{\ClientTypeSingle{A}}$. In the following calculation, in each steps we calculate the denotations of a larger subprocess of $\Transformer{x_1,x_1'}{\ClientTypeSingle{A}}$. We proceed as follows:
    \begin{align*}
       (1)\quad &\AtkeyDenotations{\Transformer{y,y'}{A}\para \Emptyoutt{z}\CPJud y:\Dual{A}, y':\DualTranLau{A},z:\unit}\\
        &=\{(a,a',*)\para (a,a') \in \AtkeyDenotations{\Transformer{y,y'}{A}\CPJud y:\Dual{A}, y':\DualTranLau{A}}\}\\
        (2)\quad
        &\AtkeyDenotations{\impp{z}{y'}(\Transformer{y,y'}{A}\para \Emptyoutt{z})\CPJud y:\Dual{A}, z:\DualTranLau{A}\parrType\unit}\\
        &=\{(a,(a',*)) \para (a,a',*)\in \AtkeyDenotations{\Transformer{y,y'}{A}\para \Emptyoutt{z}\CPJud y:\Dual{A}, y':\DualTranLau{A},z:\unit}\}\\
        (3)\quad 
        &\AtkeyDenotations{\boutt{z}{w}(\impp{z}{y'}(\Transformer{y,y'}{A}\para \Emptyoutt{z})\para \Emptyimpp{w}\inact)\CPJud y:\Dual{A}, w:(\DualTranLau{A}\parrType\unit)\Tensor\bot}\\
        &=\{(a,((a',*),*)) \para (a,(a',*)) \in \AtkeyDenotations{\impp{z}{y'}(\Transformer{y,y'}{A}\para \Emptyoutt{z})\CPJud y:\Dual{A}, z:\DualTranLau{A}\parrType\unit}\}\\
        (4)\quad
        &\AtkeyDenotations{\Client{x'}{w}\boutt{z}{w}(\impp{z}{y'}(\Transformer{y,y'}{A}\para \Emptyoutt{z})\CPJud y:\Dual{A}, x':\ClientType{((\DualTranLau{A}\parrType\unit)\Tensor\bot)}}\\
        &=\{(a,\bag{((a',*),*)}) \para\\
        &\quad\quad  (a,((a',*),*)) \in \AtkeyDenotations{\boutt{z}{w}(\impp{z}{y'}(\Transformer{y,y'}{A}\para \Emptyoutt{z})\para \Emptyimpp{w}\inact)\CPJud y:\Dual{A}, w:(\DualTranLau{A}\parrType\unit)\Tensor\bot}\}\\
        (5)\quad
        &\AtkeyDenotations{\overbrace{\Server{x}{y}\Client{x'}{w}\boutt{z}{w}(\impp{z}{y'}(\Transformer{y,y'}{A}\para \Emptyoutt{z})}^{\Transformer{x,x'}{\ClientTypeSingle{A}}}\CPJud x:\ServerType{\Dual{A}}, x':\ClientType{((\DualTranLau{A}\parrType\unit)\Tensor\bot)}}\\
        &=\{(\bag{a_1,\dots,a_k},\uplus_{j_1}^k\bag{((a_j',*),*)}) \para \forall i\in\{1,\dots,k\}.\\
        &\quad\quad (a_i,\bag{((a_i',*),*)}) \in \AtkeyDenotations{\Client{x'}{w}\boutt{z}{w}(\impp{z}{y'}(\Transformer{y,y'}{A}\para \Emptyoutt{z})\CPJud y:\Dual{A}, x':\ClientType{((\DualTranLau{A}\parrType\unit)\Tensor\bot)}}\}
    \end{align*}

    Thus, from the previous calculations we have:
    \begin{align*}
        &\AtkeyDenotations{\overbrace{\Server{x}{y}\Client{x'}{w}\boutt{z}{w}(\impp{z}{y'}(\Transformer{y,y'}{A}\para \Emptyoutt{z})}^{\Transformer{x,x'}{\ClientTypeSingle{A}}}\CPJud x:\ServerType{\Dual{A}}, x':\ClientType{((\DualTranLau{A}\parrType\unit)\Tensor\bot)}}\\
        &=\{(\bag{a_1,\dots,a_k},\uplus_{j_1}^k\bag{((a_j',*),*)}) \para \forall i\in\{1,\dots,k\}.\\
        &\quad\quad (a_i,a_i')\in \AtkeyDenotations{\Transformer{y,y'}{A}\CPJud y:\Dual{A}, y':\DualTranLau{A}}\}
    \end{align*}
\end{proof}

We may now establish the correctness of transformers:
\begin{lemmaapxrep}\label{lemmaTransformers} Suppose $P\CPJud \Gamma$. Then $\AtkeyDenotations{\TranLauP{P}\CPJud\DualTranLau{\Gamma},w:\unit}=\AtkeyDenotations{\TranProNames{P}{\Gamma}\CPJud\DualTranLau{\Gamma},w:\unit}$.
\end{lemmaapxrep}
\begin{proofsketch}
By induction on the structure of $P$. We consider a number of illustrative cases.
If $P=\Server{x}{y}P'$, then:
\begin{align*}
\TranLauP{\Server{x}{y}P'}&=\boutt{x}{m}(\Server{x}{w}\impp{w}{y}\TranLauP{P'}\para\forward{m}{n})\\
&\ObsEquivAtk
\boutt{x}{m}(\Server{x}{w}\impp{w}{y}\TranProNames{P'}{\Delta,y:A}\para\Emptyimpp{m}\inact)\para \Emptyoutt{n}\tag{by I.H. and  \Cref{lemma: observational equivalence forwarders} }\\
    &\ObsEquivAtk \TranProNames{P'}{\Delta,x:\ServerType{A}}\tag{by \Cref{lemma: observational equivalence of transformers}}
\end{align*}
If $P=\cut{x}{R}{Q}$, then we need to show:
\[
  \TranLauP{P} \ObsEquivAtk 
\cut{z}{\impp{z}{y}\TranProNames{Q}{\Gamma,x:\Dual{A}}}{\cut{w}{\impp{w}{x'}\TranProNames{R}{\Delta,x:A}}{\Synchronizer{z,w}{A}}}\ObsEquivAtk \TranProNames{\cut{x}{R}{Q}}{\Gamma,\Delta}\]
By I.H. and \Cref{theorem: laurent transformation on denotations} we know that: 
\begin{align*}
     &(\FDenotationsInd{\Delta}{\delta},\FDenotationsInd{\Gamma}{\gamma},*)\in \AtkeyDenotations{\cut{z}{\impp{z}{y}\TranProNames{Q}{\Gamma,x:\Dual{A}}}{\cut{w}{\impp{w}{x'}\TranProNames{R}{\Delta,x:A}}{\Synchronizer{z,w}{A}}}}\\
     &\quad\quad \Leftrightarrow\,(\delta,\gamma)\in\AtkeyDenotations{\cut{x}{R}{Q}}
\end{align*}
Thus, the observations on the left-hand side are exactly those from $\cut{x}{R}{Q}$ under some transformation; that transformation being the one induced by the transformers, having thus the same observation as $\TranProNames{\cut{x}{R}{Q}}{\Gamma,\Delta}$. When the last rule applied is either  \rulenamestyle{W} or \rulenamestyle{C}, we rely on the I.H. and \Cref{lemma: observation transformer client}.
\end{proofsketch} 
\begin{proof}
 By induction on the structure of $P$.
\begin{itemize}

\item Case $\inact$.
Note that $\TranLauP{\inact}=\Emptyoutt{x}$, and $\TranProNames{\inact}{\cdot}=\TranProNameswone{\inact}{\cdot}\para\Emptyoutt{x}=\Emptyoutt{x}$.

\item Case $\Emptyoutt{x}$.
By \Cref{def: Tranformation of processes}
$\TranLauP{\Emptyoutt{x}}=\boutt{x}{m}(\Emptyoutt{x}\para\fwd{m}{n})$.  By \cref{lemma: transformation on processes}, we have: 
$$\boutt{x}{m}(\Emptyoutt{x}\para\fwd{m}{n})\CPJud m:\unit\Tensor\bot,n:\unit$$ 
On the other hand, $\Transformer{x,m}{\unit}  = \boutt{y}{m}(\forward{y}{x}\para \Emptyimpp{m}\inact)\CPJud  m:\unit\Tensor\bot, x:\bot$, and $$\TranProNames{\Emptyoutt{x}}{m:\unit\Tensor\bot, x:\bot}=\cut{x}{\Emptyoutt{x}}{\boutt{y}{m}(\forward{y}{m}\para \Emptyimpp{x}\inact)}\para\Emptyoutt{n}\CPJud m:\unit\Tensor\bot,n:\unit$$

Thus, the result follows by straightforwardly calculating the denotations.
\item Case $\Emptyimpp{x}P$. We have:
\begin{prooftree}
    \AxiomC{$P\CPJud\Delta$}
    \UnaryInfC{$\Emptyimpp{x}P\CPJud\Delta,x:\bot$}
\end{prooftree}

By the I.H. we know $\AtkeyDenotations{\Emptyimpp{x}\TranLauP{P}\CPJud\DualTranLau{\Delta},x:\bot,w:\unit}=\AtkeyDenotations{\Emptyimpp{x}\TranProNames{P}{\Delta}\CPJud\DualTranLau{\Delta},x:\bot,w:\unit}$. By \Cref{lemma: observational equivalences rule mix,theorem: observational equivalences} we know that 
$\Emptyimpp{x}\TranProNames{P}{\Delta}\ObsEquivAtk \TranProNames{\Emptyimpp{x}P}{\Delta}$. 
This way, $\TranProNames{\Emptyimpp{x}P}{\Delta}\ObsEquivAtk \TranProNames{\Emptyimpp{x}P}{\Delta,x:\bot}$ follows from the fact that $\Transformer{x,x'}{\bot}$ is a forwarder.

\item Case $P\para Q$.
\begin{prooftree}
    \AxiomC{$P\CPJud\Gamma$}
    \AxiomC{$Q\CPJud \Delta$}
    \BinaryInfC{$P\para Q\CPJud \Gamma,\Delta$}
\end{prooftree}

By the I.H.: 
\begin{align*}
    \AtkeyDenotations{\TranLauP{P}\CPJud\DualTranLau{\Gamma},z:\unit}&=
    \AtkeyDenotations{\TranProNames{P}{\Gamma}\CPJud\DualTranLau{\Gamma},z:\unit}\\
    \AtkeyDenotations{\TranLauP{Q}\CPJud\DualTranLau{\Delta},w:\unit}&=
    \AtkeyDenotations{\TranProNames{Q}{\Delta}\CPJud\DualTranLau{\Delta},w:\unit}
\end{align*}
By \Cref{def: Tranformation of processes}:
\[\TranLauP{P\para Q}=\cut{w}{\boutt{z}{w}(\TranLauP{P}\para \TranLauP{Q})}{\impp{w}{z}\Emptyimpp{z}\forward{w}{m}}\]
Thus, we have the following equivalences:
\begin{align*}
    &\cut{w}{\boutt{z}{w}(\TranLauP{P}\para \TranLauP{Q})}{\impp{w}{z}\Emptyimpp{z}\forward{w}{m}}\\
    &\ObsEquivAtk\cut{w}{\boutt{z}{w}(\TranProNames{P}{\Gamma}\para \TranProNames{Q}{\Delta})}{\impp{w}{z}\Emptyimpp{z}\forward{w}{m}}\tag{by I.H.}\\
    &\ObsEquivAtk\cut{w}{\boutt{z}{w}(\TranProNames{P}{\Gamma}\para \TranProNames{Q}{\Delta})}{\impp{w}{z}\Emptyimpp{z}(\Emptyimpp{w}\inact\para\Emptyoutt{m})}\tag{by \Cref{lemma: observational equivalence forwarders}}\\
     &\ObsEquivAtk\cut{w}{\boutt{z}{w}(\TranProNames{P}{\Gamma}\para \TranProNames{Q}{\Delta})}{\impp{w}{z}\Emptyimpp{z}\Emptyimpp{w}\inact}\para\Emptyoutt{m}\tag{by \Cref{lemma: observational equivalences rule mix}}\\
      &=\cut{w}{\boutt{z}{w}((\TranProNameswone{P}{\Gamma}\para\Emptyoutt{z})\para (\TranProNameswone{Q}{\Delta}\para\Emptyoutt{w}))}{\impp{w}{z}\Emptyimpp{z}\Emptyimpp{w}\inact}\para\Emptyoutt{m}
\end{align*}

We finish by showing that:
\[\cut{w}{\boutt{z}{w}((\TranProNameswone{P}{\Gamma}\para\Emptyoutt{z})\para (\TranProNameswone{Q}{\Delta}\para\Emptyoutt{w}))}{\impp{w}{z}\Emptyimpp{z}\Emptyimpp{w}\inact}\para\Emptyoutt{m}\ObsEquivAtk \TranProNames{P\para Q}{\Gamma,\Delta}\]

However, note that: 
\begin{align*}
    \TranProNames{P\para Q}{\Gamma,\Delta}&= \TranProNameswone{P\para Q}{\Gamma,\Delta}\para \Emptyoutt{m}\\
    &\ObsEquivAtk\TranProNameswone{P}{\Gamma}\para\TranProNameswone{Q}{\Delta}\para \Emptyoutt{m}
\end{align*}

Thus, it is enough to show that:
\[\cut{w}{\boutt{z}{w}((\TranProNameswone{P}{\Gamma}\para\Emptyoutt{z})\para (\TranProNameswone{Q}{\Delta}\para\Emptyoutt{w}))}{\impp{w}{z}\Emptyimpp{z}\Emptyimpp{w}\inact}\para\Emptyoutt{m}\ObsEquivAtk \TranProNameswone{P}{\Gamma}\para\TranProNameswone{Q}{\Delta}\para \Emptyoutt{m}\]

which comes out straightforwardly  from the fact that the observations of both processes are exactly those transformed with $\TranProNameswone{P}{\Gamma}$ and $\TranProNameswone{Q}{\Delta}$.
\item Case $\cut{x}{P}{Q}$. We have:
\begin{prooftree}
    \AxiomC{$P\CPJud \Delta,x:A$}
    \AxiomC{$Q\CPJud \Gamma,x:\Dual{A}$}
    \BinaryInfC{$\cut{x}{P}{Q}\CPJud \Delta,\Gamma$}
\end{prooftree}

By \cref{fig:DenotationalSemantics}:
\begin{align*}
    \AtkeyDenotations{\cut{x}{P}{Q}\CPJud\Delta,\Gamma}
    &=\{(\delta,\gamma)\para (\delta,a)\in \AtkeyDenotations{P\CPJud\Delta,x:A}, (\gamma,a)\in \AtkeyDenotations{Q\CPJud \Gamma,x:A^\bot}\}
\end{align*}

By the I.H. we have:
\begin{align*}
\AtkeyDenotations{\TranLauP{Q}\CPJud\DualTranLau{\Gamma},x':\DualTranLau{\Dual{A}},z:\unit} &=\AtkeyDenotations{\TranProNames{Q}{\Gamma}\CPJud\DualTranLau{\Gamma},x':\DualTranLau{\Dual{A}},z:\unit}\\
    \AtkeyDenotations{\TranLauP{P}\CPJud \DualTranLau{\Delta},x':\DualTranLau{A},w:\unit}&=
    \AtkeyDenotations{\TranProNames{P}{\Delta,x:A}\CPJud \DualTranLau{\Delta},x':\DualTranLau{A},w:\unit}
\end{align*}
Thus, we have the following equality:
\begin{align*}
    &\AtkeyDenotations{\cut{z}{\impp{z}{x'}\TranLauP{Q}}{\cut{w}{\impp{w}{x'}\TranLauP{P}}{\Synchronizer{w,z}{A}}}\CPJud \DualTranLau{\Delta},\DualTranLau{\Gamma},s:\unit}\\
    &=\AtkeyDenotations{\cut{z}{\impp{z}{x'}\TranProNames{Q}{\Gamma,x:\Dual{A}}}{\cut{w}{\impp{w}{x'}\TranProNames{P}{\Delta,x:A}}{\Synchronizer{w,z}{A}}}\CPJud \DualTranLau{\Delta},\DualTranLau{\Gamma},s:\unit}
\end{align*}

Note that by \Cref{theorem: laurent transformation on denotations}:
\begin{align*}
    \FDenotations{\Delta,\Gamma}{\AtkeyDenotations{\cut{x}{P}{Q}\CPJud\Delta,\Gamma}}=\AtkeyDenotations{\TranLauP{(\cut{x}{P}{Q})}\CPJud\DualTranLau{\Delta},\DualTranLau{\Gamma},s:\unit}\label{case:cut}
\end{align*}
thus we have both:
\begin{align}
(\delta,\gamma)\in\AtkeyDenotations{\cut{x}{P}{Q}} \, &\Leftrightarrow \,(\FDenotationsInd{\Delta}{\delta},\FDenotationsInd{\Gamma}{\gamma},*)\in \AtkeyDenotations{\cut{w}{\impp{w}{x'}\TranLauP{P}}{\cut{z}{\impp{z}{y}\TranLauP{Q}}{\Synchronizer{z,w}{A}}}}\\
  \,&\Leftrightarrow \,(\FDenotationsInd{\Delta}{\delta},\FDenotationsInd{\Gamma}{\gamma},*)\in \AtkeyDenotations{\cut{w}{\impp{w}{x'}\TranProNames{P}{\Delta}}{\cut{z}{\impp{z}{y}\TranProNames{Q}{\Gamma}}{\Synchronizer{z,w}{A}}}}\label{tag: 2}
\end{align}
  
We finish the proof by showing that:
\[\AtkeyDenotations{\TranProNames{\cut{x}{P}{Q}}{\Delta,\Gamma}}=\AtkeyDenotations{\cut{w}{\impp{w}{x'}\TranProNames{P}{\Delta}}{\cut{z}{\impp{z}{y}\TranProNames{Q}{\Gamma}}{\Synchronizer{z,w}{A}}}}\]

Note that, by \cref{lemma: denotations of synchronizers}, the denotations in $\AtkeyDenotations{\cut{w}{\impp{w}{x'}\TranProNames{P}{\Delta}}{\cut{z}{\impp{z}{y}\TranProNames{Q}{\Gamma}}{\Synchronizer{z,w}{A}}}}$, are those that are composable between $\TranProNames{P}{\Delta}$ and $\TranProNames{Q}{\Gamma}$, and by \eqref{tag: 2}, those are the transformed  denotations of $\cut{x}{P}{Q}$. Moreover, the denotations of $\TranProNames{\cut{x}{P}{Q}}{\Delta,\Gamma}$ are also the transformed denotations of $\cut{x}{P}{Q}$.

\item Case $\impp{x}{y}P$.
By the I.H., $\TranLauP{P}\ObsEquivAtk \TranProNames{P}{\Delta,y:A,x:B}$. Also, by \Cref{def: Tranformation of processes}, $\TranLauP{\impp{x}{y}P.}=\impp{x'}{y}\TranLauP{P}$.
Thus, we have
\begin{align*}
    \impp{x'}{y}\TranLauP{P}&\ObsEquivAtk\impp{x'}{y}\TranProNames{P}{\Delta,y:A,x:B}\tag{by I.H.}\\
    &\ObsEquivAtk \TranProNames{\impp{x}{y}P}{\Delta,x':A\parrType B}
    \tag{by \Cref{lemma: observational equivalence of transformers} }
\end{align*}

\item Case $\boutt{y}{x}(\textcolor{Black}{P_1} \para \textcolor{Black}{P_2})$.
By the I.H., $\TranLauP{P_1}\ObsEquivAtk \TranProNames{P_1}{\Delta,y:A}$ and$\TranLauP{P_2}\ObsEquivAtk \TranProNames{P_2}{\Gamma,x:B}$.
By \Cref{def: Tranformation of processes}:
\[\TranLauP{\boutt{y}{x}(\textcolor{Black}{P_1} \para \textcolor{Black}{P_2})}=\boutt{z_2}{z}(\boutt{z_1}{z_2}(\impp{z_1}{y'} \TranLauP{\textcolor{Black}{P_1}}\para\impp{z_2}{x'} \TranLauP{\textcolor{Black}{P_2}})\para \fwd{z}{w})\]

Thus we have the following equivalences:
\begin{align*}
    &\boutt{z_2}{z}(\boutt{z_1}{z_2}(\impp{z_1}{y'} \TranLauP{\textcolor{Black}{P_1}}\para\impp{z_2}{x'} \TranLauP{\textcolor{Black}{P_2}})\para \fwd{z}{w})\\
    &\quad\ObsEquivAtk  \boutt{z_2}{z}(\boutt{z_1}{z_2}(\impp{z_1}{y'} \TranLauP{\textcolor{Black}{P_1}}\para\impp{z_2}{x'} \TranLauP{\textcolor{Black}{P_2}})\para (\Emptyimpp{m}\inact\para\Emptyoutt{n}))\tag{by \Cref{lemma: observational equivalence forwarders}}\\
     &\quad\ObsEquivAtk  \boutt{z_2}{z}(\boutt{z_1}{z_2}(\impp{z_1}{y'} \TranLauP{\textcolor{Black}{P_1}}\para\impp{z_2}{x'} \TranLauP{\textcolor{Black}{P_2}})\para \Emptyimpp{m}\inact)\para\Emptyoutt{n}\tag{by \Cref{lemma: observational equivalences rule mix}}\\
     &\quad\ObsEquivAtk  \boutt{z_2}{z}(\boutt{z_1}{z_2}(\impp{z_1}{y'} \TranProNames{P_1}{\Delta,y:A}\para\impp{z_2}{x'} \TranProNames{P_2}{\Gamma,x:B})\para \Emptyimpp{m}\inact)\para\Emptyoutt{n} \tag{by I.H.}\\
     &\quad\ObsEquivAtk \TranProNames{\boutt{y}{x}(\textcolor{Black}{P_1}\para \textcolor{Black}{P_2})}{\Delta,\Gamma,x:A\Tensor B}\tag{by \Cref{lemma: observational equivalence of transformers}}
\end{align*}

\item Case $\Server{x}{y}P$. We have:
\begin{prooftree}
    \AxiomC{$P\CPJud \ClientTypeSingle{\Delta},y:A$}
    \UnaryInfC{$\Server{x}{y}P\CPJud\ClientTypeSingle{\Delta},x:\ServerTypeSingle{A}$}
\end{prooftree}

By the I.H.:
\[\AtkeyDenotations{\TranLauP{P}\CPJud \ClientTypeSingle{\Delta},y:A}=\AtkeyDenotations{\TranProNames{P}{\ClientTypeSingle{\Delta},y:A}\CPJud \ClientTypeSingle{\Delta},y:A}\]

By \Cref{def: Tranformation of processes} we have
$\TranLauP{\Server{x}{y}P}=\boutt{x}{m}(\Server{x}{w}\impp{w}{y}\TranLauP{P}\para\fwd{m}{n})$, and so 
 we have the following equivalences:
\begin{align*}
&\boutt{x}{m}(\Server{x}{w}\impp{w}{y}\TranLauP{P}\para\fwd{m}{n})\\ &\ObsEquivAtk \boutt{x}{m}(\Server{x}{w}\impp{w}{y}\TranLauP{P}\para (\Emptyimpp{m}\inact \para \Emptyoutt{n})) \tag{by \Cref{lemma: observational equivalence forwarders}}\\
&\ObsEquivAtk \boutt{x}{m}(\Server{x}{w}\impp{w}{y}\TranLauP{P}\para \Emptyimpp{m}\inact )\para \Emptyoutt{n}\tag{by \Cref{lemma: observational equivalences rule mix}} \\
&\ObsEquivAtk \boutt{x}{m}(\Server{x}{w}\impp{w}{y}\TranProNames{P}{\Delta,y:A}\para\Emptyimpp{m}\inact)\para \Emptyoutt{n}\tag{by I.H.}\\
    &\ObsEquivAtk \TranProNames{P}{\Delta,x:\ServerTypeSingle{A}}\tag{by \Cref{lemma: observational equivalence of transformers}}
\end{align*}
 
    \item Case $\Client{x}{y}P$.
    By I.H.,
\[\TranLauP{P}\CPJud \DualTranLau{\Delta}, y: \DualTranLau{A}, z:\unit\ObsEquivAtk\TranProNames{P}{\Delta,y:A}\CPJud \DualTranLau{\Delta}, y: \DualTranLau{A}, z:\unit\]

By \Cref{def: Tranformation of processes}, $\TranLauP{\Client{x}{y}P} =\Client{x'}{m}\boutt{z}{m}( \impp{z}{y}\TranLauP{P}\para \fwd{n}{m})$. Then, we have the following equivalences:
\begin{align*}
&\Client{x'}{m}\boutt{z}{m}( \impp{z}{y}\TranLauP{P}\para \fwd{n}{m})\\
&\ObsEquivAtk \Client{x'}{m}\boutt{z}{m}( \impp{z}{y}\TranLauP{P}\para(\Emptyimpp{m}\inact\para \Emptyoutt{n}) )\tag{by \Cref{lemma: observational equivalence forwarders}}\\
&\ObsEquivAtk \Client{x}{m}\boutt{z}{m}( \impp{z}{y}\TranLauP{P}\para \Emptyimpp{m}\inact)\para \Emptyoutt{n}\tag{by \Cref{lemma: observational equivalences rule mix}}\\
&\ObsEquivAtk \Client{x}{m}\boutt{z}{m}( \impp{z}{y}\TranProNames{P}{\Delta,y:A}\para \Emptyimpp{m}\inact)\para \Emptyoutt{n}\tag{by I.H.}\\
    &\ObsEquivAtk \TranProNames{\Client{x}{y}P}{\Delta,x:\ClientTypeSingle{A}}\tag{by \Cref{lemma: observational equivalence of transformers}}
\end{align*}

\item Case $\choice{y}{i}P$.
 By the I.H.,
 $\TranLauP{P}\ObsEquivAtk\TranProNames{P}{\Delta,y:A} $, and by \Cref{def: Tranformation of processes}
$$\TranLauP{\choice{y}{i}P}=\boutt{w}{m}(\choice{w}{i}\impp{w}{y}\TranLauP{P}\para \fwd{m}{n})$$
Thus, we have:
   \begin{align*}
   &\boutt{w}{m}(\choice{w}{i}\impp{w}{y}\TranLauP{P}\para \fwd{m}{n})\\
   &\ObsEquivAtk \boutt{w}{m}(\choice{w}{i}\impp{w}{y}\TranLauP{P}\para (\Emptyimpp{m}\inact\para \Emptyoutt{n}) )\tag{by \Cref{lemma: observational equivalence forwarders}}\\
    &\ObsEquivAtk \boutt{w}{m}(\choice{w}{i}\impp{w}{y}\TranLauP{P}\para \Emptyimpp{m}\inact)\para \Emptyoutt{n}\tag{by \Cref{lemma: observational equivalences rule mix}}\\
      &\ObsEquivAtk \boutt{w}{m}(\choice{w}{i}\impp{w}{y}\TranProNames{P}{\Delta,y:A}\para \Emptyimpp{m})\para \Emptyoutt{n} \tag{by I.H.}\\
       &\ObsEquivAtk \TranProNames{\choice{y}{i}P}{y:A_1\choiceType A_2} \tag{by \Cref{lemma: observational equivalence of transformers}}
   \end{align*}

   \item Case $\case{x}{\textcolor{Black}{P_1}}{\textcolor{Black}{P_2}}$.
    By I.H. $\TranLauP{P_i}\ObsEquivAtk\TranProNames{P_i}{\Delta,x:A_i}$, and
    by \Cref{def: Tranformation of processes}
   $\TranLauP{\case{x}{\textcolor{Black}{P_1}}{\textcolor{Black}{P_2}}}=\case{x'}{\TranLauP{\textcolor{Black}{P_1}}}
        {\TranLauP{\textcolor{Black}{P_2}}}$,
 thus, we have:
     \begin{align*}
         &\case{x'}{\TranLauP{\textcolor{Black}{P_1}}}
        {\TranLauP{\textcolor{Black}{P_2}}}\\
        & \ObsEquivAtk\case{x}{\TranProNames{\textcolor{Black}{P_1}}{\Delta,x:A_1}}{\TranProNames{\textcolor{Black}{P_2}}{\Delta,x:A_2}}\tag{by I.H.}\\
         &\ObsEquivAtk \TranProNames{\case{x}{\textcolor{Black}{P_1}}{\textcolor{Black}{P_2}}}{\Delta x:A_1\caseType A_2}\tag{by \Cref{lemma: observational equivalence of transformers}}
     \end{align*}

     \item Case $P\Substitution{x_1}{x_2}$.
We have:
     \begin{prooftree}
         \AxiomC{$P\CPJud \Delta,x_1:\ClientTypeSingle{A},x_2:\ClientTypeSingle{A}$}
         \UnaryInfC{$P\Substitution{x_1}{x_2}\CPJud \Delta,x_1:\ClientTypeSingle{A}$}
         \end{prooftree}

By \Cref{def: Tranformation of processes}, $\TranLauP{P\Substitution{x_1}{x_2}} =\TranLauP{P}\Substitution{x_1'}{x_2'} $. By the I.H.:
     \begin{align*}
         &\AtkeyDenotations{\TranLauP{P}\CPJud\DualTranLau{A}, x_1':\DualTranLau{(\ClientTypeSingle{A})},x_2':\DualTranLau{(\ClientTypeSingle{A})},w:\unit}\\
         & =\AtkeyDenotations{\TranProNames{P}{ \Delta,x_1:\ClientTypeSingle{A},x_2:\ClientTypeSingle{A}}\CPJud\DualTranLau{A}, x_1':\DualTranLau{(\ClientTypeSingle{A})},x_2':\DualTranLau{(\ClientTypeSingle{A})},w:\unit}
     \end{align*}
Thus, as a consequence the following equalities hold:
     \begin{align*}
         &\AtkeyDenotations{\TranLauP{P}\Substitution{x_1'}{x_2'}\CPJud\DualTranLau{\Delta}, x_1':\DualTranLau{(\ClientTypeSingle{A})},w:\unit}\\
           &=\{(\delta',\alpha'_1\uplus\alpha'_2) \para (\delta',\alpha'_1,\alpha'_2) \in \AtkeyDenotations{\TranLauP{P}\CPJud\DualTranLau{\Delta}, x_1':\DualTranLau{(\ClientTypeSingle{A})},x_2':\DualTranLau{(\ClientTypeSingle{A})},w:\unit}\}\\
            &=\{(\delta',\alpha'_1\uplus\alpha'_2) \para (\delta',\alpha'_1,\alpha'_2) \in  \AtkeyDenotations{\TranProNames{P}{ \Delta,x_1:\ClientTypeSingle{A},x_2:\ClientTypeSingle{A}}\CPJud\DualTranLau{\Delta}, x_1':\DualTranLau{(\ClientTypeSingle{A})},x_2':\DualTranLau{(\ClientTypeSingle{A})},w:\unit} \}\\
            &=\AtkeyDenotations{\TranProNames{P}{ \Delta,x_1:\ClientTypeSingle{A},x_2:\ClientTypeSingle{A}}\Substitution{x_1'}{x_2'}\CPJud\DualTranLau{\Delta}, x_1':\DualTranLau{(\ClientTypeSingle{A})},w:\unit}
     \end{align*}

Thus, we finish the proof by showing that:
\[\TranProNames{P}{ \Delta,x_1:\ClientTypeSingle{A},x_2:\ClientTypeSingle{A}}\Substitution{x_1'}{x_2'}\ObsEquivAtk \TranProNames{P\Substitution{x_1}{x_2}}{ \Delta,x_1:\ClientTypeSingle{A},x_2:\ClientTypeSingle{A}}.\]

However, since transformers are modular, we have the following two equalities: 
\begin{align*}
    \TranProNames{P}{\Delta,x_1:\ClientTypeSingle{A},x_2:\ClientTypeSingle{A}}\Substitution{x_1'}{x_2'}&=\TranProNames{\TranProNameswone{P}{\Delta}}{x_1:\ClientTypeSingle{A},x_2:\ClientTypeSingle{A}}\Substitution{x_1'}{x_2'}\\
     \TranProNames{P\Substitution{x_1}{x_2}}{\Delta,x_1:\ClientTypeSingle{A}}&= \TranProNames{\TranProNameswone{P\Substitution{x_1}{x_2}}{\Delta}}{x_1:\ClientTypeSingle{A}}
\end{align*}

Thus, we proceed by showing that: 
\begin{align}
\AtkeyDenotations{\TranProNames{\TranProNameswone{P}{\Delta}}{x_1:\ClientTypeSingle{A},x_2:\ClientTypeSingle{A}}\Substitution{x_1'}{x_2'}}=\AtkeyDenotations{\TranProNames{\TranProNameswone{P\Substitution{x_1}{x_2}}{\Delta}}{x_1:\ClientTypeSingle{A}}}\label{equality proof client}
    \end{align}
By \Cref{lemma: observation transformer client}:
\begin{align*}&\AtkeyDenotations{\TranProNames{\TranProNameswone{P}{\Delta}}{x_1:\ClientTypeSingle{A},x_2:\ClientTypeSingle{A}}\CPJud\DualTranLau{\Delta}, x_1':\DualTranLau{(\ClientTypeSingle{A})}, x_2':\DualTranLau{(\ClientTypeSingle{A})},w:\unit}\\
    &=\{(\delta',\bag{((a,*),*)},\bag{((a',*),*)}) \para (\delta',\bag{a},\bag{a'},*) \in \AtkeyDenotations{\TranProNameswone{P}{\Delta}\CPJud\DualTranLau{\Delta},x_1:\ClientTypeSingle{A},x_2:\ClientTypeSingle{A}}
    \end{align*}

    For the left-hand side of \eqref{equality proof client} we have:
    \begin{align*}
    &\quad \AtkeyDenotations{\TranProNames{\TranProNameswone{P}{\Delta}}{x_1:\ClientTypeSingle{A},x_2:\ClientTypeSingle{A}}\Substitution{x_1'}{x_2'}\CPJud\DualTranLau{\Delta}, x_1':\DualTranLau{(\ClientTypeSingle{A})},w:\unit }\\
    &\quad=\{ (\delta',\bag{((a,*),*)}\uplus\bag{((a',*),*)}) \para (\delta',\bag{((a,*),*)},\bag{((a',*),*)}) \\
    &\qquad\in
    \AtkeyDenotations{\TranProNames{\TranProNameswone{P}{\Delta}}{x_1:\ClientTypeSingle{A},x_2:\ClientTypeSingle{A}}\CPJud\DualTranLau{\Delta}, x_1':\DualTranLau{(\ClientTypeSingle{A})}, x_2':\DualTranLau{(\ClientTypeSingle{A})},w:\unit}\}\\
     &\quad=\{ (\delta',\bag{((a,*),*),((a',*),*)}) \para (\delta',\bag{((a,*),*)},\bag{((a',*),*)}) \\
    &\qquad\in
    \AtkeyDenotations{\TranProNames{\TranProNameswone{P}{\Delta}}{x_1:\ClientTypeSingle{A},x_2:\ClientTypeSingle{A}}\CPJud\DualTranLau{\Delta}, x_1':\DualTranLau{(\ClientTypeSingle{A})}, x_2':\DualTranLau{(\ClientTypeSingle{A})},w:\unit}\}\tag{by \Cref{lemma: union of multisets}}\\
         &\quad=\{ (\delta',\bag{((a,*),*),((a',*),*)}) \para (\delta,\bag{a},\bag{a'}) \in
    \AtkeyDenotations{P\CPJud\Delta, x_1:\ClientTypeSingle{A}, x_2:\ClientTypeSingle{A}}\}
\end{align*}

And, for the right-hand side of \eqref{equality proof client}:
 \begin{align*}
     &\AtkeyDenotations{\TranProNames{\TranProNameswone{P\Substitution{x_1}{x_2}}{\Delta}}{x_1:\ClientTypeSingle{A}}\CPJud\DualTranLau{\Delta}, x_1':\DualTranLau{(\ClientTypeSingle{A})},w:\unit}\\
     &=\{(\delta',\bag{((a,*),*)}\uplus\bag{((a',*),*)}) \para (\delta',\bag{a}\uplus\bag{a'}) \in \AtkeyDenotations{\TranProNameswone{P\Substitution{x_1}{x_2}}{\Delta}\CPJud\DualTranLau{\Delta},x_1:\ClientTypeSingle{A}}\tag{by \Cref{lemma: observation transformer client}}\\
     &=\{(\delta',\bag{((a,*),*),((a',*),*)}) \para (\delta',\bag{a,a'}) \in \AtkeyDenotations{\TranProNameswone{P\Substitution{x_1}{x_2}}{\Delta}\CPJud\DualTranLau{\Delta},x_1:\ClientTypeSingle{A}}\\
        &=\{(\delta',\bag{((a,*),*),((a',*),*)}) \para (\delta,\bag{a,a'}) \in \AtkeyDenotations{P\Substitution{x_1}{x_2}\CPJud\Delta,x_1:\ClientTypeSingle{A}}\\
          &=\{ (\delta',\bag{((a,*),*),((a',*),*)}) \para (\delta,\bag{a},\bag{a'}) \in
    \AtkeyDenotations{P\CPJud\Delta, x_1:\ClientTypeSingle{A}, x_2:\ClientTypeSingle{A}}\}
 \end{align*}

 \item Case $P\CPJud\Delta,x:\ClientTypeSingle{A}$.
 We have:
 \begin{prooftree}
     \AxiomC{$P\CPJud\Delta$}
     \UnaryInfC{$P\CPJud\Delta,x:\ClientTypeSingle{A}$}
 \end{prooftree}

By I.H.  $\AtkeyDenotations{\TranProNames{P}{\Delta}\CPJud \DualTranLau{\Delta},w:\unit}=\AtkeyDenotations{\TranLauP{P}\CPJud \DualTranLau{\Delta},w:\unit}$. As a consequence we have:
\begin{align*}
    &\AtkeyDenotations{\TranLauP{P}\CPJud \DualTranLau{\Delta},x':\ClientType{((\DualTranLau{A}\parrType\unit)\Tensor\bot)},w:\unit}\\
    &=\{(\delta',\emptyset,*) \para (\delta',*) \in \AtkeyDenotations{\TranLauP{P}\CPJud \DualTranLau{\Delta},w:\unit}\}\\
    &=\{(\delta',\emptyset,*) \para (\delta',*) \in \AtkeyDenotations{\TranProNames{P}{\Delta}\CPJud \DualTranLau{\Delta},w:\unit}\}\\
    & =\AtkeyDenotations{\TranProNames{P}{\Delta}\CPJud \DualTranLau{\Delta},x':\ClientType{((\DualTranLau{A}\parrType\unit)\Tensor\bot)},w:\unit}
\end{align*}

 The proof finishes by showing that:
 \begin{align*}
    &\AtkeyDenotations{\TranProNames{P}{\Delta,x:\ClientTypeSingle{A}}\CPJud \DualTranLau{\Delta},x':\ClientType{((\DualTranLau{A}\parrType\unit)\Tensor\bot)},w:\unit}\\
    &=\AtkeyDenotations{\TranLauP{P}\CPJud \DualTranLau{\Delta},x':\ClientType{((\DualTranLau{A}\parrType\unit)\Tensor\bot)},w:\unit}
\end{align*}

 Since transformers are modular, we have that  $\TranProNames{P}{\Delta,x:\ClientTypeSingle{A}}\ObsEquivAtk \TranProNameswone{\TranProNames{P}{\Delta}}{x:\ClientTypeSingle{A}}$. Thus, calculating the denotations of $\TranProNameswone{\TranProNames{P}{\Delta}}{x:\ClientTypeSingle{A}}$, we obtain:
\begin{align*}
    &\AtkeyDenotations{\TranProNameswone{\TranProNames{P}{\Delta}}{x:\ClientTypeSingle{A}}\CPJud \DualTranLau{\Delta},x':\ClientType{((\DualTranLau{A}\parrType\unit)\Tensor\bot)},w:\unit}\\
    &=\{(\delta',\emptyset,*)\para (\delta',\emptyset,*) \in \AtkeyDenotations{\TranProNames{P}{\Delta}\CPJud \DualTranLau{\Delta},x:\ClientTypeSingle{A},w:\unit} \}
\end{align*}
what concludes the proof.
        \end{itemize}
    \end{proof}

\begin{corollaryapx}
 \label{coro:transformer_correct}
    Given $P\CPJud\Delta$, then $\TranLauP{P}\ObsEquivAtk \TranProNames{P}{\Delta}$.
\end{corollaryapx}
\begin{proof}
    The proof follows from  \Cref{ObservationalEquivalenceCorollary} and  \Cref{lemmaTransformers}.
\end{proof}

\subparagraph*{Transformer contexts and $\FDenotations{\Delta}{-}$.}
Transformer contexts induce a function on denotations, similar to $\FDenotations{\Delta}{-}$  (\Cref{def: transformation on denotations}).
In general, we have the following result, for any context.
\begin{definition}\label{def: Transformation On denotations typed contexts}
For $\TypedContext\CPJudContext{\Delta}{\Gamma}$, we define $\AtkeyDenotations{\TypedContext}: \PowerSet{\AtkeyDenotations{\Delta}}\mapsto \PowerSet{\AtkeyDenotations{\Gamma}}$ inductively:
\begin{align*}
&\AtkeyDenotations{\Hole\CPJudContext{\Delta}{\Delta}}(X)=X\\
&\AtkeyDenotations{\cut{x}{\TypedContextKProcess{\Hole}}{Q}\CPJudContext{\Delta}{\Sigma,\Gamma}}(X)
          =\left\{(\sigma,\gamma)\para
          \begin{gathered}[c]
                (\sigma,a)\in \AtkeyDenotations{\TypedContextKProcess{\Hole}\CPJudContext{\Delta}{\Sigma,x:A}}(X),\\
           (\gamma,a)\in \AtkeyDenotations{{Q}\CPJud{\Gamma,x:\Dual{A}}}
          \end{gathered}
         \right\}\\
& \AtkeyDenotations{\TypedContext\para Q\CPJudContext{\Delta}{\Sigma,\Gamma}}(X)
              =\left \{(\sigma,\gamma)\para 
             \sigma\in 
\AtkeyDenotations{\TypedContext\CPJudContext{\Delta}{\Sigma}}(X), ~~
\gamma \in \AtkeyDenotations{Q\CPJud \Gamma}
         \right\}
\end{align*}
\end{definition}
\begin{lemmaapxrep}\label{lemma: typed context denotation function}
    Let $\TypedContext\CPJudContext{\Delta} {\Gamma}$ and $P\CPJud\Delta$ be a typed context and process, respectively. Then
$\AtkeyDenotations{\TypedContext\CPJudContext{\Delta}{\Gamma}}(\AtkeyDenotations{P\CPJud\Delta})=\AtkeyDenotations{\TypedContextKProcess{P}\CPJud\Gamma}$.
\end{lemmaapxrep}

\begin{proof}
      By induction on the derivation of $\TypedContext\CPJudContext{\Delta}{\Gamma}$, with a case analysis in the last rule applied (\Cref{fig: Typed Contexts Rules}). There are three cases:
   \begin{itemize}
      \item Case $\Hole\CPJudContext{\Delta}{\Delta}$. We have
      $\AtkeyDenotations{\Hole\CPJudContext{\Delta}{\Delta}}(\AtkeyDenotations{P\CPJud\Delta})=\AtkeyDenotations{P\CPJud\Delta}$ and so the thesis follows by \Cref{def: Transformation On denotations typed contexts}.

      \item Case $\cut{x}{\TypedContextProcess{K}{\Hole}}{Q}\CPJudContext{\Delta}{\Gamma,\Sigma}$.
We have:
         \begin{align*}
            &\AtkeyDenotations{\cut{x}{\TypedContextProcess{K'}{\Hole}}{Q}\CPJudContext{\Delta}{\Gamma,\Sigma}}(\AtkeyDenotations{P\CPJud \Delta})\\
            &\quad=\{(\gamma,\sigma) \para (\sigma,a)\in\AtkeyDenotations{Q\CPJud{\Sigma,x:A}},\\
            &\quad\qquad (\gamma,a) \in  \AtkeyDenotations{\TypedContextProcess{K'}{\Hole}\CPJudContext{\Delta}{\Gamma,x:A^\bot}}(\AtkeyDenotations{P\CPJud \Delta}) \tag{by \Cref{def: Transformation On denotations typed contexts}}
            &\\
            &\quad=\{(\sigma,\gamma) \para (\sigma,a)\in\AtkeyDenotations{Q\CPJud{\Sigma,x:A}},\\
            &\quad\qquad (\gamma,a) \in  \AtkeyDenotations{\TypedContextProcess{K'}{P} \CPJud{\Gamma,x:A^\bot}}\} \tag{by the I.H.}\\
            &\quad =\AtkeyDenotations{\cut{x}{\TypedContextProcess{K'}{P}}{Q}\CPJud{\Gamma,\Sigma}} \tag{by \Cref{fig:DenotationalSemantics}}
        \end{align*}

        \item Case ${\TypedContextProcess{K'}{\Hole}}\para{Q}\CPJudContext{\Delta}{\Gamma,\Sigma}$. We have:
        \begin{align*}
            &\AtkeyDenotations{{\TypedContextProcess{K'}{\Hole}}\para{Q}\CPJudContext{\Delta}{\Gamma,\Sigma}}(\AtkeyDenotations{P\CPJud \Delta})\\
            &\quad=\{(\gamma,\sigma) \para (\sigma)\in\AtkeyDenotations{Q\CPJud{\Sigma}},\\
            &\quad\qquad (\gamma) \in  \AtkeyDenotations{\TypedContextProcess{K'}{\Hole}\CPJudContext{\Delta}{\Gamma}}(\AtkeyDenotations{P\CPJud \Delta}) \tag{by \Cref{def: Transformation On denotations typed contexts}}\\
            &\quad=\{(\sigma,\gamma) \para (\sigma)\in\AtkeyDenotations{Q\CPJud{\Sigma}},\\
            &\quad\qquad (\gamma) \in  \AtkeyDenotations{\TypedContextProcess{K'}{P} \CPJud{\Gamma}}\} \tag{by the I.H.}\\
            &\quad =\AtkeyDenotations{{\TypedContextProcess{K'}{P}}\para{Q}\CPJud{\Gamma,\Sigma}} \tag{by \Cref{fig:DenotationalSemantics}}
        \end{align*}
  \end{itemize}
\end{proof}

 \Cref{def: Transformation On denotations typed contexts} can be specialized to transformer contexts, so as to obtain the following function: $\AtkeyDenotations{\TranProNames{\Hole}{\Delta}}: \PowerSet{\AtkeyDenotations{\Delta}}\to \PowerSet{\AtkeyDenotations{\DualTranLau{\Delta},w:\unit}}$.
Putting all these elements together, we can show that transformers also internalize Laurent's translation on the level of denotations.
\begin{lemma}
  For any type $A \in \CPMixZeroTwo$, and for any $a \in \AtkeyDenotations{A}$, $b \in  \AtkeyDenotations{\DualTranLau{A}}$, we have
  $$(a,b) \in \AtkeyDenotations{\Transformer{x,y}{A}} \iff \FDenotationsInd{A}{a} = b$$
\end{lemma}
\begin{proof}
  By induction on the type $A$.
\end{proof}
\begin{theorem}
  \label{thm:transformers_L}
  For any typing context $\Delta$, and for any set $X \subseteq \AtkeyDenotations{\Delta}$,
  $$\AtkeyDenotations{\TranProNames{\Hole}{\Delta}\CPJudContext {\Delta}{\DualTranLau{\Delta},w:\unit}}(X) = \FDenotations{\Delta}{X}$$
\end{theorem} 
\begin{proof}
  Let $\Delta = x_{1}:A_{1},\dots,x_{n}:A_{n}$.
  Then, $\AtkeyDenotations{\DualTranLau{\Delta}, w:\unit} = \AtkeyDenotations{\DualTranLau{A_{1}}} \times \dots \times \AtkeyDenotations{\DualTranLau{A_{n}}} \times\{\ast\}$.
  Then, we reason as follows:
  \begin{align*}
    & (\delta_{1},\dots,\delta_{n},\ast)\in  \AtkeyDenotations{\TranProNames{\Hole}{\Delta}}(X) \\
    \iff & (\delta_{1},\dots,\delta_{n})\in \AtkeyDenotations{\TranProNameswone{\Hole}{\Delta}}(X) = \AtkeyDenotations{\cut{x_n}{\cdots\cut{x_1}{\Hole}{\Transformer{x_1,x'_1}{A_1}}}{\cdots)\para \Transformer{x_n,x'_n}{A_n}}}(X)\\
   \iff & \exists d_{n} \in \AtkeyDenotations{A_{n}}.
    \begin{array}[c]{l}
      (d_{n},\delta_{n})\in \AtkeyDenotations{\Transformer{x_{n},x'_{n}}{A_{n}}}\quad \land\\      
      (\delta_{1}, \dots, \delta_{n-1}, d_{n})\in \AtkeyDenotations{\cut{x_{n-1}}{\cdots\cut{x_1}{\Hole}{\Transformer{x_1,x'_1}{A_1}}}{\cdots)}}(X)
    \end{array}\\
   \iff & \exists d_{n} \in \AtkeyDenotations{A_{n}}, \dots, d_{1}\in\AtkeyDenotations{A_{1}}.
    \begin{array}[c]{l}
      (d_{n},\delta_{n})\in \AtkeyDenotations{\Transformer{x_{n},x'_{n}}{A_{n}}}\ \land \dots \land\ (d_{1},\delta_{1})\in \AtkeyDenotations{\Transformer{x_{1},x'_{1}}{A_{1}}}\land
\\      
      (d_{1}, \dots, d_{n})\in \AtkeyDenotations{\Hole}(X)
    \end{array}\\
    \iff&\exists (d_{1}, \dots, d_{n})\in X.\ \FDenotationsInd{A_{1}}{d_{1}}=\delta_{1}\land \dots\land \FDenotationsInd{A_{n}}{d_{n}}=\delta_{n}\\
    \iff&(\delta_{1},\dots,\delta_{n},\ast) \in \FDenotations{\Delta}{X}
\end{align*}
\end{proof}
Thus,  \Cref{thm:transformers_L}  shows that $\TranProNames{\Hole}{\Delta}$ is the proper internalization of Laurent's translation as a typed context in $\CPMixZeroTwo$.
In \Cref{sec:Laurent's denotations} we have shown that  Laurent's translation preserves and reflects equivalence of processes.
\Cref{thm:transformers_L} allows us to lift that result to processes with transformers, thus also obtaining a full abstraction result:
\begin{corollaryapx}[Full Abstraction (II)]\label{corollary: full abstraction transformers}
    For all $P\in \CPMixZeroTwo$,
\[P\ObsEquivAtk Q\CPJud \Delta \quad \Leftrightarrow\quad \TranProNames{P}{\Delta}\ObsEquivAtk\TranProNames{Q}{\Delta}\CPJud \DualTranLau{\Delta},w:\unit\]
\end{corollaryapx}
\begin{proof}
  Follows from the soundness of denotational semantics, \Cref{lemma: typed context denotation function,lemma: injectivity Function L formula,thm:transformers_L}.
\end{proof}

\section{Concluding Remarks}
\label{s:conc}

This paper has brought the  translation $\TranLau{(-)}: \CLL \to \ILL$ (due to Laurent~\cite{Laurent}), into the realm of concurrent interpretations of linear logic (\PaS). 
As we have seen, under the \PaS interpretation, the translation converts a classical process $P$ into an intuitionistic process $\TranLauP{P}$ (cf. \Cref{def: Tranformation of processes}); then, exploiting the fact that $\TranLauP{P}$ can be analyzed without changes in the classical setting, we contrast the  behavior of ${P}$ and $\TranLauP{P}$ using Atkey's observational semantics for \CP~\cite{ObservationsCP}.
Our two full abstraction results (\Cref{cor:fullabsden} and \Cref{corollary: full abstraction transformers}) give denotational and operational characterizations that extend the scope of Laurent's translation, and connect purely logical results with their corresponding computational interpretations. To our knowledge, ours is the first formal relationship of its kind. 

Differences between classical and intuitionistic variants of \PaS have already been observed by Caires and Pfenning~\cite{DBLP:conf/concur/CairesP10} and by Wadler~\cite{PropositionsAsSessions}.
There are superficial differences, such as the nature/reading of typing judgments (already discussed) and the number of typing rules---classical interpretations have one rule per connective, whereas intuitionistic ones have two: one for expressing the reliance on a behavior, another for expressing an offer.
But there are also more subtle differences, in particular the \emph{locality} principle, which, 
informally speaking, ensures that received names can only be used for sending. 
Intuitionistic interpretations enforce locality for shared names.
Consider, e.g., the process 
$P = \impp{x}{y}\Server{y}{z}Q$, which uses the name $y$ received on $x$ to define a server behavior.
Because $P$ does not respect locality, it is not typable in the intuitionistic system of~\cite{DBLP:conf/concur/CairesP10}.

Prior work by Van den Heuvel and P\'{e}rez~\cite{PiUll,DBLP:journals/corr/abs-2401-14763} studies this specific difference: they study the sets of processes typable under classical and intuitionistic interpretations, and use non-local processes such as $P$ to prove that the intuitionistic set is strictly included in the classical one.
Crucially, this prior work focuses on typing, and does not formally relate the behavioral equivalences in the two classes, as we achieve here by coupling $\TranLauP{-}$ with typed observational equivalences on processes.
In fact, our results go beyond~\cite{PiUll,DBLP:journals/corr/abs-2401-14763} in that \Cref{def: Tranformation of processes} stipulates how to translate a process $P$ with non-local servers into a corresponding process $\TranLauP{P}$ with localized servers.
This translation not only follows directly the logical translation by Laurent, but is also correct in a strong sense under the two different perspectives (denotational and operational) given by our full abstraction results.

The issue of translating non-local processes into local processes was studied, albeit in a different setting, by Boreale~\cite{Boreale}, who considers a calculus with locality as an intermediate language between the asynchronous $\pi$-calculus and the internal $\pi$-calculus.
His work makes heavy use of link processes, which are closely related to the forwarding process $\forward{x}{y}$ of \CP.
More fundamentally, because Borale's translations and results are framed in the untyped, asynchronous setting, comparisons with our work in the typed setting are difficult to draw.

We find it remarkable that our results leverage two separate, well-established technical ingredients, namely Laurent's translation and Atkey's observational equivalence and denotational semantics~\cite{ObservationsCP}.
In particular, Atkey's denotational semantics, based on the relational semantics of \CLL, is simple and  effective for our purposes, and also amenable to extensions (like incorporating support for \MixTwo).
Indeed, our denotational characterization  $\FDenotationsInd{A}{-}$ of Laurent's translation (\Cref{def: transformation on denotations}) benefits from this simplicity.

Our technical results make use of the mix principles---we use Rule~\MixZero in \secRef{sec:Laurent's denotations} and also and Rule~\MixTwo in \secRef{sec:Transformers}.
The use of mix principles in the context of \PaS has been  analyzed by Atkey et al.~\cite{Conflation}.
Already, one difference between Wadler's presentation in~\cite{PropositionsAsSessions} and Atkey's observational semantics in~\cite{ObservationsCP} is the use of Rule~\MixZero.
As we have briefly mentioned, our results in \secRef{sec:Laurent's denotations} hold also for \CPMixZeroTwo, the extension with both \MixZero and \MixTwo.

Finally, we note that Caires and Pfenning based their interpretation on Barber's Dual Intuitionistic Linear Logic (\DILL)~\cite{DILL}, which is based on sequents of the form $\Gamma;\Delta\DillSequent A$, where $\Gamma$ and $\Delta$ specify unrestricted and linear assignments, respectively.
This is a bit different from \ILL as considered by Laurent.
However, the two systems are equivalent (as logics), and so this difference does not jeopardize our results.
Barber~\cite{DILL} provides translations between \DILL and \ILL and shows that \ILL is isomorphic to the sub-system of \DILL with sequents of the form $\cdot\,; \Delta\DillSequent A$.
From the point of view of $\TranLau{(-)}$, this means that $\CPJud \Delta$ is a provable in \CLL iff $\cdot\,; \TranLau{\Delta}\DillSequent \unit$ is provable in \DILL.
Hence we can regard $\TranLauP{P}$ as a \DILL process.

This observation, together with the fact that  $\TranLauP{P}$ is typable with both $\TranLau{\Delta}\PiDiLLJud \unit$ and $\CPJud \DualTranLau{\Delta},\unit$, provides us with a solid groundwork for the computational interpretation of the translation.

\subparagraph*{Future Work}
We intend to adapt our approach to other denotational semantics for typed languages under \PaS, such as the one  by Kokke et al.~\cite{DBLP:journals/pacmpl/KokkeMP19}, whose  definition is inspired by Brzozowski derivatives and includes the polarity of names/channels. Also, we plan to study the potential of our full abstraction results as a tool for a better understanding of the locality principle for shared names in the session-typed setting. Moreover, it would be worthwhile exploring the consequences of varying the parameter $\myR$ in Laurent’s translation, which we currently instantiate with the simplest possible proposition/type.

From a more applied perspective, we believe that our work can shed light on connections between different existing implementation strategies for process calculi with session types based on linear logic. On the intuitionistic side, the work by Pfenning and Griffith develops SILL, a language based on \ILL~\cite{DBLP:conf/fossacs/PfenningG15}; on the classical side, recent work by Caires and Toninho develops a Session Abstract Machine based on \CLL~\cite{DBLP:conf/esop/CairesT24}. It would be interesting to establish to what extent our work can be applied to connect such language implementations.

\end{document}